\documentclass[12pt]{amsart}
\usepackage{a4wide}
\usepackage[latin9]{inputenc}
\usepackage{amsmath}
\usepackage{amsfonts}
\usepackage{amssymb}
\usepackage{amsthm}
\usepackage{subfigure}
\usepackage[usenames,dvipsnames]{color}
\usepackage{pstricks}
\usepackage{graphicx}
\usepackage[all]{xy}
\newtheorem{theo}{Theorem}[section]
\newtheorem{prop}[theo]{Proposition}

\newtheorem{lemm}[theo]{Lemma}

\newtheorem{ex}[theo]{Example}
\theoremstyle{definition}
\newtheorem{def1}[theo]{Definition}
\theoremstyle{remark}
\newtheorem{rema}[theo]{Remark}
\newcommand{\Op}{\operatorname{Op}}

\newcommand{\nwc}{\newcommand}
\nwc{\eps}{\epsilon}
\nwc{\ep}{\epsilon}
\nwc{\vareps}{\varepsilon}
\nwc{\Oph}{\operatorname{Op}_\hbar}
\nwc{\la}{\langle}
\nwc{\ra}{\rangle}

\nwc{\mf}{\mathbf} %Latex (as in \bf not tilted math letters)
\nwc{\blds}{\boldsymbol} %Latex 
\nwc{\ml}{\mathcal} %Latex

\nwc{\defeq}{\stackrel{\rm{def}}{=}}

\nwc{\cE}{\ml{E}}
\nwc{\cN}{\ml{N}}
\nwc{\cO}{\ml{O}}
\nwc{\cP}{\ml{P}}
\nwc{\cU}{\ml{U}}
\nwc{\cV}{\ml{V}}
\nwc{\cW}{\ml{W}}
\nwc{\tU}{\widetilde{U}}
\nwc{\IN}{\mathbb{N}}
\nwc{\IR}{\mathbb{R}}
\nwc{\IZ}{\mathbb{Z}}
\nwc{\IC}{\mathbb{C}}
\nwc{\IT}{\mathbb{T}}
\nwc{\IS}{\mathbb{S}}
\nwc{\tP}{\widetilde{P}}
\nwc{\tPi}{\widetilde{\Pi}}
\nwc{\tV}{\widetilde{V}}
\nwc{\supp}{\operatorname{supp}}
\nwc{\rest}{\restriction}

\begin{document}

\title[Spectral analysis of Morse-Smale flows I]{Spectral analysis of Morse-Smale flows I:\\
Construction of the anisotropic spaces}


\author[Nguyen Viet Dang]{Nguyen Viet Dang}

\address{Institut Camille Jordan (U.M.R. CNRS 5208), Universit\'e Claude Bernard Lyon 1, B\^atiment Braconnier, 43, boulevard du 11 novembre 1918, 
69622 Villeurbanne Cedex }

\email{dang@math.univ-lyon1.fr}

\author[Gabriel Rivi\`ere]{Gabriel Rivi\`ere}

\address{Laboratoire Paul Painlev\'e (U.M.R. CNRS 8524), U.F.R. de Math\'ematiques, Universit\'e Lille 1, 59655 Villeneuve d'Ascq Cedex, France}

\email{gabriel.riviere@math.univ-lille1.fr}

\begin{abstract} 
We prove the existence of a discrete correlation spectrum for Morse-Smale flows acting on smooth forms on a compact manifold. This is done by constructing spaces of currents with anisotropic Sobolev regularity on which the Lie derivative has a discrete spectrum.  
\end{abstract}

\maketitle

\section{Introduction}

Given a smooth ($\ml{C}^{\infty}$), compact, oriented, boundaryless manifold $M$ of dimension $n\geq 1$ and a smooth flow $\varphi^t:M\rightarrow M$, a basic question from dynamical systems is to understand the 
long time behaviour of the flow. There are many ways to approach this problem. For instance, one can define 
the \emph{correlation function}~:
\begin{equation}\label{e:correlation}
\boxed{C_{\psi_1,\psi_2}(t):=\int_M\varphi^{-t*}(\psi_1)\wedge\psi_2,}
\end{equation}
with $\psi_1\in\Omega^k(M)$ and $\psi_2\in\Omega^{n-k}(M)$. Then, if one can describe the limit of this quantity as $t\rightarrow+\infty$, then it gives some information 
on the weak limit of $\varphi^{-t*}(\psi_1)$ in the sense of currents. 
Studying directly the limit of $C_{\psi_1,\psi_2}(t)$ as $t\rightarrow +\infty$ is not often possible and one may first introduce its Laplace transform~:
\begin{equation}\label{e:Laplacecorrel}
\boxed{ \hat{C}_{\psi_1,\psi_2}(z):=\int_0^{+\infty}e^{-tz}C_{\psi_1,\psi_2}(t)dt.}
\end{equation}

Note that this is well defined for $\text{Re}(z)>c$ with $c>0$ depending only on the flow $\varphi^t$. Instead of studying the long time limit, one could then try to understand if 
this holomorphic function has a meromorphic extension to a larger half-plane. If so, the poles and their residues also give some informations on the long time dynamics of the flow. Therefore, in the sequel, the set of poles
and residues of the Laplace transformed correlation functions will be called
\textbf{correlation (or Pollicott-Ruelle) spectrum} of the flow.  
These kind 
of questions were for instance considered by Pollicott~\cite{Po85} and Ruelle~\cite{Ru87a} in the framework of Axiom $A$ flows that we shall now discuss.

In fact, this type of problem is very hard at this level of generality and some assumptions on the nature of the flow should be made to obtain some nontrivial results. 
A natural situation where one may expect some answer is when some hyperbolicity is involved in the system, e.g. for Axiom A flows in 
the sense of Smale~\cite{Sm67}. In that framework, one can decompose the nonwandering set of the flow into finitely many invariant hyperbolic 
subsets $(\Lambda_j)_{j=1}^K$, which are called the basic sets of the flow. Examples of such flows are geodesic flows on negatively curved 
manifolds or gradient flows associated with a Morse function. As most of the time on a given orbit
is spent in some neighborhood of these basic sets, it is natural to first restrict to test forms which are supported in a small neighborhood of a given $\Lambda_j$. For a slightly 
different correlation function associated with a Gibbs measure of $\Lambda_i$, Pollicott~\cite{Po85} and Ruelle~\cite{Ru87a} proved the meromorphic extension 
of the Laplace transform to some half-plane $\text{Re}(z)\geq -\delta$ with $\delta>0$. Their proof relies on the symbolic coding by Markov partitions of such flows that was constructed by Bowen~\cite{Bo73}. 
In the last fifteen years, many progresses have been made towards this problem by adopting a slightly different point of view. Namely, 
one can observe that, for
a $k$-form $\psi_1$, its pull--back 
$\varphi^{-t*}(\psi_1)$ by the flow solves 
the following partial differential equation:
$$\boxed{\partial_t\psi=-\ml{L}_V^{(k)}\psi,\quad\psi(t=0)=\psi_1,}$$
where $\ml{L}_V^{(k)}=(d+\iota_V)^2$ is the Lie derivative along the vector field $V$ associated with the flow $\varphi^t$.
In particular, if one can find an appropriate Banach space on which 
$-\ml{L}_V$ has a discrete spectrum (with finite multiplicity) on the half-plane $\text{Re}(z)\geq-\delta$ for some positive $\delta$, then one can verify that $\hat{C}_{\psi_1,\psi_2}(z)$ 
has a meromorphic extension to the same half-plane. This approach has been initiated by Liverani~\cite{Li04} in the context of contact Anosov flows and it was further developed 
in~\cite{BuLi07, GiLiPo13} where it is proved, among other things,
that $\hat{C}_{\psi_1,\psi_2}(z)$ has a meromorphic extension to the entire complex plane for Anosov flows. In these references, one of the 
key ingredient is the construction of Banach spaces with anisotropic H\"older regularity on which $-\ml{L}_V$ has good spectral properties. 
Alternative spaces based on microlocal 
tools were developed by Dyatlov, Faure, Sj\"ostrand, Tsujii and Zworski~\cite{Ts10, FaSj11, Ts12, DyZw13, FaTs13}. This complementary approach 
allowed to bring new perspectives on the fine structure of this correlation spectrum in the Anosov case. Coming back to the case of Axiom A flows, Dyatlov and Guillarmou proved that 
$\hat{C}_{\psi_1,\psi_2}(z)$ admits a meromorphic extension to $\IC$ provided that we only consider test forms which are compactly supported in a neighborhood of a fixed basic 
set $\Lambda_j$~\cite{DyGu14}. We should point that progress 
for flows follow from 
earlier results for hyperbolic diffeomorphisms 
that we will not discuss here. We refer the reader 
to the book of Baladi for a recent detailed account on that case~\cite{Ba16}.

If we look at the general framework of Axiom A flows, we already noticed that these different results do not say at 
first sight much things on the \emph{global dynamics} of the 
flow. In fact, in all the results we mentioned so far, it is important that we restrict ourselves to test forms which are compactly 
supported near a fixed basic set $\Lambda_j$ 
of the flow. In the case of geodesic flows on negatively curved manifolds, this restriction is of course artificial as there is only one basic set which 
covers the entire manifold. However, in general, there may be several basic sets that are far from covering the entire manifold and, if we 
remember that Axiom A flows arised as far-reaching generalizations of gradient flows associated with a Morse function~\cite{Sm67}, then understanding 
the global dynamics sounds also important as it provides some informations on the topology of the manifold~\cite{Sm60, Fr82}. Several difficulties appear 
if we want to consider this global question and let us mention at least two of them: (1) the flow is not topologically transitive on $M$ (while 
it is on a fixed $\Lambda_j$), (2) hyperbolicity only holds on the basic sets. 

\textbf{From Morse-Smale gradient flows to Morse--Smale flows.}
In the case of certain Morse-Smale gradient flows, the fact that $C_{\psi_1,\psi_2}(t)$ admits 
a limit as $t\rightarrow +\infty$ for \emph{any} choice of $\psi_1$ and $\psi_2$ was 
proved by Harvey and Lawson~\cite{HaLa00, HaLa01} -- see also~\cite{Lau92} for earlier related results. 
While the proof of Harvey and Lawson was based on the theory of currents \`a la Federer, we 
recently showed how to develop an appropriate \emph{global} spectral theory for such gradient flows~\cite{DaRi16} and 
to derive a complete asymptotic expansion of the correlation function -- see 
also~\cite{FrenkelLosevNekrasov1} for related results in the context of quantum field theory. 
In particular, this shows that $\hat{C}_{\psi_1,\psi_2}$ has a meromorphic extension to $\mathbb{C}$. Moreover, as a byproduct of our spectral analysis, 
we obtained a new spectral interpretation of the Thom-Smale-Witten complex as the kernel of the operator $-\ml{L}_V$ 
acting on certain anisotropic 
spaces of currents from which one can easily deduce the finiteness of Betti numbers, the Poincar\'e duality and the classical Morse 
inequalities. 

The goal of the present work is to continue to \emph{explore the global dynamics} of Axiom A flows by focusing on the particular case of 
Morse-Smale flows~\cite{Sm60, PaDeMe82} for which the basic sets are either closed orbits or fixed points. 
These flows 
also satisfy 
some transversality assumptions necessary to develop proper topological applications~\cite{Sm60} -- see section~\ref{s:smale} for more details.  
These flows are more general than the families of gradient flows we considered in~\cite{DaRi16} as 
they may have closed periodic orbits. 
They also form a natural subfamily of simple Axiom A flows~\cite{Sm67} 
where one may expect to develop a proper \emph{global} spectral theory. Recall that Peixoto proved that, 
in dimension $2$, these flows form an open and dense (in the $\ml{C}^{\infty}$ topology) family of all smooth vector fields~\cite{Pe62} 
while in higher dimension, Palis showed that they form an open subset of all smooth vector fields~\cite{Pa68}.

This article is the first in a series. Here, we develop a convenient \emph{global} functional framework for Morse-Smale flows in order to prove 
that the Laplace 
transformed correlator $\hat{C}_{\psi_1,\psi_2}(z)$
defined by equation (\ref{e:Laplacecorrel}) has a meromorphic extension to the entire complex plane. In~\cite{DaRi17b}, 
we will elaborate more on Morse-Smale flows and show how to give an explicit description of the poles and of 
residues of $\hat{C}_{\psi_1,\psi_2}(z)$ provided certain non resonance assumptions are satisfied. 
Finally, in~\cite{DaRi17c}, 
we will 
explain how to give topological interpretations for the correlation spectrum of
a class of flows which have a proper \emph{global} spectral theory that we call
microlocally tame. 
This class contains Anosov and Morse--Smale flows. 
We will show how to extract 
Morse inequalities for Pollicott-Ruelle resonant states
in the kernel of $\mathcal{L}_V$ and we will also
give some new identities 
relating regularized products of
Pollicott-Ruelle resonances on the imaginary
axis with a torsion function introduced by Fried which coincides with Reidemeister torsion
when $V$ is non singular~\cite{F87}.

\section{Statement of the main results}

In the article, $M$ will denote a smooth ($\ml{C}^{\infty}$), compact, oriented manifold without boundary and of dimension $n\geq 1$. We will 
also endow this manifold with a smooth ($\ml{C}^{\infty}$) Riemannian metric that will be denoted by $g$.

\subsection{Discrete correlation spectrum}
Our main result shows the existence of a meromorphic extension to $\IC$ of $\hat{C}_{\psi_1,\psi_2}(z)$~:

\fbox{
\begin{minipage}{0,94\textwidth} 
\begin{theo}[Resonances]\label{t:maintheo1} 
Let $\varphi^t$ be a $\ml{C}^{\infty}$ Morse-Smale flow which is $\ml{C}^1$ linearizable. 
Denote by $V$ the corresponding vector field and let $0 \leq k\leq n$. 

Then, there exists a minimal\footnote{We just mean that, for every $z_0\in\ml{R}_k(V)$, one can find $(\psi_1,\psi_2)$ such that there is indeed a pole at 
$z_0$.} 
discrete subset 
$\ml{R}_k(V)\subset\IC$ such that, given any $(\psi_1,\psi_2)\in\Omega^k(M)\times\Omega^{n-k}(M)$, the map
 $$z\mapsto \hat{C}_{\psi_1,\psi_2}(z)$$
 has a meromorphic extension from $\{\operatorname{Re}(z)>c\}$ to $\IC$ whose poles are of finite order and contained inside $\ml{R}_k(V)$.
\end{theo}
\end{minipage}
}
\\
\\
We refer to section~\ref{s:smale} for a precise definition of a Morse-Smale flow. By discrete, we mean that $\ml{R}_k(V)$ has no accumulation points. In particular, it is at most countable.
Elements inside $\ml{R}_k(V)$ are often referred as 
\emph{Pollicott-Ruelle resonances} or as the \emph{correlation spectrum} 
of the flow. 
Besides the fact that the flow is Morse-Smale, 
we need to make an assumption on the fact that the flow is $\ml{C}^1$-linearizable -- see paragraph~\ref{ss:C1-lin} for the precise definition. 
Roughly speaking, it means that the flow is $\ml{C}^1$-conjugated to a linear flow in a neighborhood of each basic set of the flow. This may 
sound like a big constraint. Yet, thanks to the Sternberg-Chen Theorem~\cite{Ch63, Ne69, WWL08}, it is satisfied as soon as
a certain number 
of non resonance assumptions are made on the Lyapunov exponents of the basic set. In particular, they are satisfied for a generic 
choice of Morse--Smale flow. We refer to appendix~\ref{aa:Sternberg} for a brief account on these non resonance hypotheses. This 
assumption of being $\ml{C}^1$-linearizable may be artificial at this stage of our analysis but it does not look obvious to us how to remove 
it in an easy manner. 

Note that, if $\psi_1$ 
and $\psi_2$ were supported in 
a small neighborhood of some given basic set, then the existence of this discrete correlation spectrum 
could be deduced near critical points from~\cite{BaTs07, BaTs08, GoLi08} and near closed orbits from~\cite{DyGu14}. Here, 
the main novelty is that the result holds \textbf{globally} on the manifold, i.e. without any restriction on the supports of $\psi_1$ and $\psi_2$. It also 
generalizes our previous results from~\cite{DaRi16} which were only valid for Morse-Smale \textbf{gradient flows which are not allowed 
to have periodic orbits}. Observe that, even if the nonwandering set is the union of finitely many basic sets, it is not obvious that the 
\emph{global correlation spectrum} should be the union of the correlation spectra associated with each individual basic set. We shall see 
in~\cite{DaRi17b} that this is indeed the case if enough nonresonance conditions are satisfied by the Lyapunov exponents. Without these assumptions, 
it is not completely obvious if some unexpected phenomenon may occur in the correlation spectrum.

In fact, thanks to its spectral nature, our proof will not only give the meromorphic extension of $\hat{C}_{\psi_1,\psi_2}(z)$ but also some information on its residues~:
\begin{theo}[Resonant states]\label{t:maintheo2}  Let $\varphi^t$ be a Morse-Smale flow which is $\ml{C}^1$ linearizable. 
Denote by $V$ the corresponding vector field 
and let $0 \leq k\leq n$. 

Then, for every $z_0\in\ml{R}_k(V)$, there exists an integer $m_k(z_0)\geq 1$ and a linear map of finite rank
$$\pi_{z_0}^{(k)}:\Omega^k(M)\rightarrow\ml{D}^{\prime k}(M)$$
 such that, given any $(\psi_1,\psi_2)\in\Omega^k(M)\times\Omega^{n-k}(M)$, one has, in a small neighborhood of $z_0$,
 $$\boxed{\hat{C}_{\psi_1,\psi_2}(z)=\sum_{l=1}^{m_k(z_0)}(-1)^{l-1}\frac{\left\la(\ml{L}_V^{(k)}+z_0)^{l-1}\pi_{z_0}^{(k)}(\psi_1),\psi_2\right\ra}{(z-z_0)^l}+
 R_{\psi_1,\psi_2}(z),}$$
 where $R_{\psi_1,\psi_2}(z)$ is a holomorphic function.
Moreover, any element $u$ in the range of $\pi_{z_0}^{(k)}$ satisfies the generalized eigenvalue equation
 $$\boxed{(\ml{L}_V^{(k)}+z_0)^{m_k(z_0)}(u)=0.}$$
\end{theo}
Here, $\ml{D}^{\prime k}(M)$ denotes the 
De Rham currents of degree $k$, i.e. the topological dual of $\Omega^{n-k}(M)$. Elements inside the 
range of $\pi_{z_0}^{(k)}$ are called the \emph{Pollicott-Ruelle resonant states}. Our proof will say more on the Sobolev 
regularity of these currents along stable and unstable directions. In~\cite{DaRi17b}, we will show how to exploit this Sobolev regularity to give a rather precise description of these 
resonant states in a neighborhood of the basic sets of the flow. In~\cite{DaRi17c}, we will show that some of these 
states have a deep topological meaning related to the De Rham 
complex~\cite{dRh80} and to its Reidemeister torsion~\cite{F87}. 
Finally, we emphasize that Dyatlov and Guillarmou characterized elements inside the range of $\pi_{z_0}^{(k)}$ in terms of currents solving the 
generalized eigenvalue equation and verifying some support and wavefront assumptions~\cite{DyGu14}. This was valid for general Axiom A flows but only 
near a fixed basic set. It is plausible that a similar property holds globally for Morse-Smale flows but we shall not discuss this question here.

\subsection{About the proofs of Theorems~\ref{t:maintheo1} and~\ref{t:maintheo2}}

We will consider a slightly more general framework than the one we described so far. Fix a complex vector bundle 
$\ml{E}\rightarrow M$ of rank $N$ and some connection $\nabla:\Omega^0(M,\ml{E})\rightarrow\Omega^1(M,\ml{E})$~\cite{Lee09} -- see also 
paragraph~\ref{ss:connection} for a brief reminder. Then, 
one can define a covariant derivative $d^{\nabla}:\Omega^{\bullet}(M,\ml{E})\rightarrow\Omega^{\bullet+1}(M,\ml{E})$ and introduce the operator
$$\forall 0\leq k\leq n,\ \ml{L}_{V,\nabla}^{(k)}:=d^{\nabla}\circ\iota_V+\iota_V\circ d^{\nabla}:\Omega^{k}(M,\ml{E})\rightarrow\Omega^{k}(M,\ml{E}).$$
Note that, in the present article, we will not make the assumption that $\nabla$ is flat, i.e. that $d^{\nabla}\circ d^{\nabla}=0$. Our goal is to 
introduce anisotropic Sobolev spaces of currents adapted to the dynamics of the Morse-Smale vector field $V$ in the sense that $-\ml{L}_{V,\nabla}^{(k)}$ has 
a discrete spectrum on this space at least for $\text{Re}(z)\geq -\sigma$ with $\sigma>0$. 
Introducing this algebraic framework is in fact important for the applications to topology we have in mind~\cite{DaRi17c}. In particular, it will allow us to 
 formulate generalizations of the Morse inequalities for general vector bundles and to relate the Pollicott-Ruelle spectrum to the Reidemeister 
 torsion appearing in the works of Fried~\cite{F87}.

The proofs of Theorems~\ref{t:maintheo1} and~\ref{t:maintheo2} will follow the 
microlocal approach of Faure and Sj\"ostrand~\cite{FaSj11}. Recall that their construction is based on the fact that the 
operator\footnote{They do not actually deal with this generalized geometric framework but their construction can be adapted as 
the operators under consideration have scalar principal symbols.} $-\ml{L}_{V,\nabla}^{(k)}$ is a differential operator whose principal symbol is $H(x;\xi)\textbf{Id}_{\Lambda^k(T^*M)\otimes\ml{E}}$ where
\begin{equation}\label{e:hamiltonian-lie}
 \forall (x;\xi)\in T^*M,\quad H(x;\xi):=\xi(V(x)).
\end{equation}
Then, they show that the spectrum of the operator $\ml{L}_{V,\nabla}^{(k)}$ (hence the correlation spectrum) can be obtained in a 
similar manner as in the theory of semiclassical resonances~\cite{HeSj86, DyZw16}. In particular, this requires to understand the dynamical 
properties of the Hamiltonian flow induced by $H$ on $T^*M$, equivalently of the symplectic lift of $\varphi^t$ to the cotangent space. Precisely, it 
requires to describe the topological and dynamical properties of the set of points which are trapped by 
the Hamiltonian dynamics either in the future or in the past. 
Due to the fact that we want to deal with the global correlation spectrum, 
we need to analyze precisely 
the global Hamiltonian dynamics which is the content of section~\ref{s:cotangent} and which is 
probably the main new difficulty compared to the Anosov case 
treated in~\cite{FaSj11}. We emphasize that this part of the proof 
is implicitely related to the classical results of Smale~\cite{Sm60} on Morse-Smale flows -- see also~\cite{Web06} for a formulation closer to ours 
in the case of gradient flows. The major difference is that we are interested here in the Hamiltonian dynamics
on $T^*M$ rather than the dynamics 
on the base space $M$. 
Note that we already had to deal with similar difficulties in the context of gradient flows~\cite{DaRi16} 
and we give here a more systematic approach which allows to deal with closed orbits. We should also point out that the results we obtain in that direction 
are in some sense related to some results of Laudenbach who gave a very precise description of the  
closure 
of unstable manifolds for Morse--Smale gradient flows~\cite{Lau92}. 
Once this is well understood, we construct in Lemma~\ref{l:escape-function} an appropriate 
escape (or Lyapunov) function for the Hamiltonian dynamics. 
This construction combines our analysis of 
the global dynamics with some results due to Meyer on the existence of energy functions for Morse-Smale flows~\cite{Me68}. Given this 
escape function, we can follow the spectral construction of Faure and Sj\"ostrand whose crucial ingredient is in fact the existence 
of such a function -- see section~\ref{s:sobolev} for more details.

\subsection{Organization of the article}
In section~\ref{s:smale}, we review some classical facts on Morse-Smale flows and introduce some conventions that we will use all along the article. 
Then, in section~\ref{s:cotangent}, we explain how to define the symplectic lift of a flow and we introduce coordinates systems adapted to 
our problem. Along the way, we collect a few facts from Floquet theory. In section~\ref{s:conormal}, we introduce the conormal bundle to 
the unstable manifolds and we state the main new dynamical results of the article, namely Theorems 
\ref{t:attractor} and \ref{t:compactness}. Sections~\ref{s:proofcompact} and~\ref{s:proof-attractor} are 
devoted to the proof of these results. In section~\ref{s:escape}, we use these results to construct an escape function for the 
symplectic lift of our flow. Finally, we prove Theorems~\ref{t:maintheo1} and~\ref{t:maintheo2} by 
constructing anisotropic Sobolev spaces adapted to the vector field $V$. 
In the
appendix~\ref{a:hyperbolic}, we collect some classical results 
on hyperbolic fixed points and closed orbits that we use in our proofs and, in appendix~\ref{a:proof-smale}, we briefly recall the proofs of some results due 
to Smale~\cite{Sm60} which may be helpful to understand the proofs of our results on the Hamiltonian dynamics.

\subsection*{Acknowledgements} We warmly thank Fr\'ed\'eric Faure for many explanations on his works with Johannes Sj\"ostrand and Masato Tsujii. We 
also acknowledge 
useful discussions related to this article and its companion articles~\cite{DaRi17b, DaRi17c} 
with Yannick Bonthonneau, Livio Flaminio, Colin Guillarmou, Benoit Merlet, Fr\'ed\'eric Naud and Patrick Popescu Pampu. The presentation 
of the article highly benefited from a detailed and careful report of an anonymous referee that we kindly acknowledge. The second author 
is partially supported by the Agence Nationale de la Recherche through the Labex CEMPI (ANR-11-LABX-0007-01) and the 
ANR project GERASIC (ANR-13-BS01-0007-01).

\section{Review on Morse-Smale flows}\label{s:smale}

The purpose of this preliminary section is to collect some well-known facts on the so-called Morse-Smale flows which were introduced by Smale 
in~\cite{Sm60} as a generalization of gradient flows induced by a Morse function. Besides the seminal work of Smale, good references 
on the subject are~\cite[Ch.~8]{Fr82} and~\cite[Ch.~4]{PaDeMe82}. 

\begin{rema}
All along the article, we implicitely assume that $M$ is endowed 
with a $\ml{C}^{\infty}$ Riemannian structure $g$ which plays an auxiliary role, all the results being independent of the choice of $g$. 
\end{rema}

\subsection{Definition and examples} We say that $\Lambda\subset M$ is an elementary critical element if $\Lambda$ is either a fixed point or a closed orbit of $\varphi^t$. Such an element is said to be hyperbolic if 
the fixed point or the closed orbit is hyperbolic -- see appendix~\ref{a:hyperbolic} for a brief reminder. Following~\cite[p.~798]{Sm67}~:

\fbox{
\begin{minipage}{0,94\textwidth} 
\begin{def1}
\label{d:morsesmale}
A flow $\varphi^t$ is called \textbf{a Morse-Smale flow} if the following properties hold~:

\begin{enumerate}
 \item the non-wandering set $\operatorname{NW}(\varphi^t)$ is the 
union of finitely many elementary critical elements 
$\Lambda_1,\ldots,\Lambda_K$ which are hyperbolic,
 \item for every $(i,j)\in \{1,\dots,K\}^2$ and for every $x$ in $W^u(\Lambda_i)\cap W^s(\Lambda_j)$, one has
\begin{equation}\label{e:transversality0}
 T_xM=T_xW^u(\Lambda_i)+T_xW^s(\Lambda_j)
\end{equation}  
where $W^{u}(\Lambda_i)$ (resp $W^s(\Lambda_j)$) denotes the
unstable (resp stable) manifold of $\Lambda_*$.
\footnote{See the appendix~\ref{a:hyperbolic} for the definition of the stable/unstable manifolds $W^{s/u}(\Lambda)$.} 
\end{enumerate}
\end{def1}
\end{minipage}
}
\\
\\
The second assumption should be understood as a transversality property between stable and unstable manifolds. This hypothesis is crucial in our analysis as 
it was already the case in the works of Smale~\cite{Sm60}. Note that these two submanifolds may in fact be equal at some points.
For instance, if we consider the gradient flow associated with the height function on $\mathbb{S}^2$
endowed with the canonical induced metric, any points outside the poles $\mathbf{n}$ and $\mathbf{s}$ belongs to $W^u(\mathbf{s})$ and $W^s(\mathbf{n})$ 
which are both two-dimensional.

As was already mentioned, these flows generalize the so-called Morse-Smale gradient flows, 
and they are the simplest examples of Axiom A flows in 
the sense of Smale~\cite[p.~803]{Sm67}. Before giving some remarkable properties of these flows, let us start with the following observation (which follows only from the first part of the definition):
\begin{lemm}\label{l:partition}
For every $x$ in $M$, there exists a unique pair 
$(i,j)\in \{1,\dots,K\}^2$ 
such that $x\in W^u(\Lambda_i)\cap W^s(\Lambda_j)$.
\end{lemm}

In particular, we have~:
\begin{lemm}\label{l:partitionprop}
The unstable manifolds $(W^u(\Lambda_j))_{j=1,\ldots, K}$ form a partition of $M$, i.e.
$$M=\bigcup_{j=1}^KW^u(\Lambda_j),\ \ \text{and}\ \ \forall i\neq j,\ W^u(\Lambda_i)\cap W^u(\Lambda_j)=\emptyset.$$
\end{lemm}
The same property holds for stable manifolds.
\begin{proof} Fix some element $x$ in $M$. We aim at 
proving that there exists a unique $i$ such that $x\in W^{u}(\Lambda_i)$. 
By reversing the time, we would get the same conclusion for 
stable manifolds. Note that the hyperbolicity assumption 
ensures that the sets $(\Lambda_j)_j$ are disjoint. Hence, we can assume that $x\notin \cup_{j=1}^K\Lambda_j$.
Set $J$ to be the subset of $\{1,\dots,K\}$ such that $\forall j\in J, x\in W^u(\Lambda_j)$ and
$\forall j\notin J$, $x\notin W^u(\Lambda_j)$.
Assume that $J$ contains more than two elements. Then, we pick a small enough open neighborhood $O$ of 
$\cup_{j\in J}\Lambda_j$ having at least two connected components. By continuity of
$\varphi^{-t}(x)$ in $t$, we can extract a sequence $(t_n)_n$ of times tending to $+\infty$ such that
$\varphi^{-t_n}(x)\in M\setminus \tilde{O}$ ($\forall n$) where $\tilde{O} \supset O$ is some open neighborhood 
of the compact set $NW(\varphi^t)$. Therefore, by extracting again, we get
a sequence which converges to some point $y\in M$ belonging to $NW(\varphi^t)$ 
as a limit point of $\varphi^{-t_n}(x)$. On the 
other, $y\notin \tilde{O}$ by construction which yields the contradiction. Hence, $|J|=1$. 
\end{proof}

\subsection{Classical results on Morse-Smale flows} 
We shall now collect some other useful facts on Morse-Smale flows following the seminal article of Smale~\cite{Sm60}. 
As the proofs are not very long and as they are instructive for the arguments of the upcoming sections, we briefly 
recall how to prove most of them in appendix~\ref{a:proof-smale}.

\subsubsection{First properties}
We start with the following direct consequence from the definition -- see appendix~\ref{a:proof-smale} for details.

\begin{lemm}\label{l:index-increase} Let $x$ be an element in $W^u(\Lambda_i)\cap W^s(\Lambda_j)$. Then, one has
 $$\operatorname{dim} W^u(\Lambda_i)\geq \operatorname{dim}W^u(\Lambda_j).$$
 Moreover, if $x\notin\Lambda_j$, equality can occur only if $\Lambda_j$ is a closed orbit.

\end{lemm}

It roughly means that the dimension of the unstable manifolds must decrease along the limit sets of the flow. 
We continue our description of the properties of Morse-Smale flow with the no-cycle property.

\begin{lemm}[No-cycle]\label{l:nocycle}
 If $x$ belongs to $W^u(\Lambda_j)\cap W^s(\Lambda_j)$, then $x\in\Lambda_j$.
\end{lemm}

Note that, in his original article~\cite{Sm60}, Smale took this property as one of the axioms satisfied by his flows. Yet, in~\cite{Sm67}, 
this assumption was removed as it can be deduced from the other axioms using the so-called $\lambda$-Lemma (or inclination Lemma) recalled in the 
appendix~\ref{a:proof-smale}.

\subsubsection{Ordering unstable manifolds.}

Let us now turn to the main feature of these flows which was proved by Smale in~\cite{Sm60} -- see also~\cite[p.752]{Sm67}:

\fbox{
\begin{minipage}{0,94\textwidth} 
\begin{theo}[Smale]\label{t:smale} Suppose that $\varphi^t$ is a Morse-Smale flow. Then, for every $1\leq j\leq K$, 
the closure of $W^u(\Lambda_j)$ is the union of certain $W^u(\Lambda_{j'})$. 

Moreover, if we define a relation $\preceq$ as follows~: 
$W^u(\Lambda_{j'})\preceq 
 W^u(\Lambda_j)$ whenever $W^u(\Lambda_{j'})$ is contained in the closure of $W^u(\Lambda_{j})$, then $\preceq$ 
 is a \textbf{partial order} on the set $W^u(\Lambda_j)_{j=1}^K$.
 
 Finally, if $W^u(\Lambda_{j'})\preceq
 W^u(\Lambda_j)$, then $\operatorname{dim}W^u(\Lambda_{j'})\leq \operatorname{dim}W^u(\Lambda_{j}).$
 \end{theo}
\end{minipage}
}
\\
\\
We shall also denote $W^u(\Lambda_{j'})\prec
 W^u(\Lambda_j)$ if $W^u(\Lambda_{j'})\preceq
 W^u(\Lambda_j)$ and $j^\prime\neq j$.
We refer to appendix~\ref{a:proof-smale} for a reminder on the proof of this result. This partial 
order relation was 
crucial for Smale to construct a filtration of the manifold in order to prove his
Morse inequalities~\cite{Sm60} and it is  
related to the concept of topological stratification. A stratum is less than a bigger stratum if it 
lies in the closure of the bigger stratum. Following~\cite[p.~753]{Sm67} -- see also~\cite[Ch.~4]{PaDeMe82}, we can set

\begin{def1}[Smale quiver.]
The partial order relation on the collection of subsets $W^u(\Lambda_j)_{j=1}^K$ defined above is called \textbf{Smale causality relation}. 
We define an oriented graph $D$, called \textbf{Smale quiver}, whose $K$ vertices are given by $W^u(\Lambda_j)_{j=1}^K$. Two  
vertices $W^u(\Lambda_j), W^u(\Lambda_i) $ are connected
by an oriented edge starting at $W^u(\Lambda_j)$ and ending at $W^u(\Lambda_i) $
iff $W^u(\Lambda_j)\succ W^u(\Lambda_i)$.
\end{def1}
The set $\left(W^u(\Lambda_j)_{j=1}^K,\preceq\right)$ being partially ordered by Smale's Theorem, 
the oriented graph $D$ is sometimes referred to as the Hasse 
diagram of this partially ordered set.

\subsection{$\ml{C}^1$ Linearizable flows}\label{ss:C1-lin}

In order to construct anisotropic Sobolev spaces adapted to a Morse-Smale flow, we will need to make some extra assumption that will roughly says that the flow is linearizable in a 
$\ml{C}^1$ chart near every critical element $\Lambda_j$. Such an assumption may sound quite restrictive. Yet, thanks to Sternberg-Chen's Theorem (see appendix~\ref{aa:Sternberg} for 
a brief reminder), it is automatically satisfied as soon as certain (generic) non resonance assumptions on the Lyapunov exponents are satisfied.

For a hyperbolic fixed point $\Lambda$, we say that $\varphi^t$ is \emph{$\ml{C}^{1}$-linearizable near $\Lambda$} if there exists a 
$\ml{C}^1$ diffeomorphism 
$h: B_n(0,r)\rightarrow O$ (where $O$ is a small open neighborhood of $\Lambda$ and $B_n(0,r)$ is a small ball of radius $r$ centered at $0$ in $\IR^n$) 
s.t. $h(0)=\Lambda$ and $h^*(V)(x)=D_xh(L(x))$ where $L(x)$ is a vector field on $B_n(0,r)$ defined as
$$L(x):=Ax.\partial_x,$$
for some $A\in M_n(\IR)$.

For a hyperbolic closed orbit $\Lambda$ of minimal period $\ml{P}_{\Lambda}$, we say that $\varphi^t$ is \emph{$\ml{C}^{1}$-linearizable 
near $\Lambda$} if there exists a $\ml{C}^1$ diffeomorphism $h: (z,\theta)\in B_{n-1}(0,r)\times(\IR/\ml{P}_{\Lambda}\IZ)\rightarrow O$ (where $O$ 
is a small open neighborhood of $\Lambda$ and $r>0$ is small) and a 
$\ml{C}^1$ map $A:\IR/\ml{P}_{\Lambda}\IZ\rightarrow M_{n-1}(\IR)$ such that
$$h\left(\{0\}\times(\IR/\ml{P}_{\Lambda}\IZ)\right)=\Lambda,$$
and such that $h^*(V)=d_{z,\theta}h (L(z,\theta))$ where $L(z,\theta)$ is 
a vector field on $B_{n-1}(0,r)\times(\IR/\ml{P}_{\Lambda}\IZ)$ defined as
$$L(z,\theta)=A(\theta)z.\partial_z+\partial_{\theta}.$$

Finally, we say that a Morse-Smale flow is \emph{$\ml{C}^{1}$-linearizable} if it is 
$\ml{C}^1$-linearizable near each  
critical element 
$(\Lambda_j)_{j=1,\ldots,K}$. More generally, for any $k\geq 1$, we will say that it is \emph{$\ml{C}^k$-linearizable} if the ``linearizing'' diffeomorphism can be chosen of class $\ml{C}^k$ for every critical element.

\subsection{Existence of an energy function}

To conclude our short review on Morse-Smale flows, we mention the following result due to Meyer~\cite[p.~1034]{Me68} that 
will be central in our construction of anisotropic Sobolev spaces:
\begin{theo}[Meyer]\label{t:meyer} Let $V$ be a vector field generating a $\ml{C}^{\infty}$ Morse-Smale flow. 
Then, there exists a smooth function $E:M\rightarrow \IR$ such that
$$\ml{L}_VE\geq 0\ \text{on}\ M,\ \text{and}\ \ml{L}_VE>0\ \text{on}\ M\setminus (\cup_{i=1}^K\Lambda_i).$$ 
Moreover, $E$ is constant on every connected component of $\cup_{j=1}^K\Lambda_j$ and one can choose $E$ 
such that $E\rceil_{\Lambda_j}>E\rceil_{\Lambda_i}$ whenever $W^u(\Lambda_j)\prec W^u(\Lambda_i)$.
\end{theo}
In the articles of Meyer~\cite{Me68} and Smale~\cite{Sm67}, such a function is called an energy function for the flow. 
It is also quite common to call $E$ a Lyapunov function for the flow.

\section{Lifted dynamics on the cotangent space}\label{s:cotangent}

The results of Smale described the structure of the closure of unstable manifolds for Morse-Smale flows. In this article, we aim at applying 
microlocal methods for the study of such flows following for instance the strategy developed in~\cite{FaSj11}. For that purpose, we 
need to understand more things about these flows, more specifically we need to describe the properties of their symplectic lift
to $T^*M$ in order to be able to construct the anisotropic Sobolev spaces.  

 Following the strategy initiated in~\cite{DaRi16} for Morse-Smale gradient flows, this includes describing 
the closure of the conormals to unstable manifolds. Before entering this delicate issue, we recall how to define this symplectic lift 
and we collect in this section a few tools from Floquet theory that will be used in our proofs.
Our main emphasis here is on giving explicit coordinates representation
of the flow near critical elements.

\subsection{Hamiltonian lift}\label{ss:hamiltonian}
A flow on $M$ can be lifted to the cotangent space $T^*M$ as follows.
In this paragraph, we shall denote by $(x;\xi)$ the elements of $T^*M$ where $x\in M, \xi\in T^*_xM$. 
We associate to the vector field $V$ a Hamiltonian function,
$$\forall (x;\xi)\in T^*M,\ H(x;\xi):=\xi\left(V(x)\right).$$
This Hamiltonian function also induces a Hamiltonian flow that we denote by $\Phi^t:T^*M\rightarrow T^*M$ and whose vector field will be denoted by $X_H$. 
We note that, by construction, 
\begin{equation}\label{e:basicliftcotangentflow}
\Phi^t(x;\xi):=\left(\varphi^t(x);\left(d\varphi^t(x)^T\right)^{-1}\xi\right),
\end{equation}
and that \emph{this flow induces a diffeomorphism between $T^*M\backslash 0$ 
and $T^*M\backslash 0$}. Observe that we denoted by $\varphi^{t*}$ the action 
of the flow by pull--back
on smooth differential form in the introduction. 
This lifted flow also induces a 
smooth flow on the unit cotangent bundle $S^*M$, i.e.
$$\forall t\in\IR,\ \forall(x;\xi)\in S^*M,\ \tilde{\Phi}^t(x;\xi)=\left(\varphi^t(x);\frac{\left(d\varphi^t(x)^T\right)^{-1}\xi}{\left\|\left(d\varphi^t(x)^T\right)^{-1}\xi\right\|_x}\right).$$
Here, $\|.\|_x$ represents the metric induced on $T_x^*M$ by $g_x$. 
Note that this is an auxiliary datum in our analysis~: 
introducing it allows us to work with compact subsets 
of $T^*M$ rather
than with conical subsets. 
We denote by $\tilde{X}_{H}$ the induced smooth vector field on $S^*M$.

\subsection{Writing the flow in local coordinates near critical points}\label{ss:coordinates}

Let us rewrite the Morse-Smale flow and its Hamiltonian lift near the critical elements of $\varphi^t$ in some well-chosen local coordinates. For that purpose, we now make 
the assumption that $\varphi^t$ is a Morse-Smale flow which is $\ml{C}^1$-linearizable. 
Near a \emph{fixed point} $\Lambda$, we can 
choose local $\ml{C}^1$ coordinates $(\tilde{x},\tilde{y})\in\IR^{n_s}\times\IR^{n_u}$ such that 
\begin{equation}\label{e:linearize-flow-critical-point}
 \varphi^t(\tilde{x},\tilde{y})=(e^{-t\Omega_s}\tilde{x},e^{t\Omega_u}\tilde{y}),
\end{equation}
with $\Omega_s$ (resp. $\Omega_u$) an element in $M_{n_s}(\IR)$ (resp. $M_{n_u}(\IR)$) all of whose eigenvalues have positive real part from 
the hyperbolicity assumption. In particular, there exist some positive constants 
$0<C_1<C_2$ and $0<\chi_-<\chi_+$ such that, for every $t\geq 0$ and for every $(\tilde{x},\tilde{y})$,
\begin{equation}\label{e:lyapunov-bound-fixed-point}C_1e^{-t\chi_+}\|(\tilde{x},\tilde{y})\|
\leq\|(e^{-t\Omega_{s}}\tilde{x},e^{-t\Omega_{u}}\tilde{y})\|\leq C_2e^{-t\chi_-}\|(\tilde{x},\tilde{y})\|
\end{equation}
where $\|.\|$ is the Euclidean norm on $\mathbb{R}^n$.
As the chart is of class $\ml{C}^1$, by equation 
(\ref{e:basicliftcotangentflow}) we can also write the
corresponding 
Hamiltonian flow (in the induced local coordinates $(\tilde{x},\tilde{y};\tilde{\xi},\tilde{\eta})$ on $T^*M$) under the form:
\begin{equation}\label{e:linearize-hamiltonian-flow-critical-point}
 \Phi^t(\tilde{x},\tilde{y};\tilde{\xi},\tilde{\eta})=(e^{-t\Omega_s}\tilde{x},e^{t\Omega_u}\tilde{y};e^{t\Omega_s^T}\tilde{\xi},e^{-t\Omega_u^T}\tilde{\eta}).
\end{equation}

\subsection{ODE with periodic coefficients}\label{ss:floquet}

We now turn to the case where $\Lambda$ 
is a \emph{closed orbit} of minimal period $\ml{P}_{\Lambda}>0$. 
For that purpose, we shall first recall a few 
facts from Floquet theory~\cite[p.~91]{Te12}. For every $\theta_0\in\IR/\ml{P}_{\Lambda}\IZ$, 
we consider the matrix valued ordinary differential equation:
 \begin{equation}\label{e:floquet-matrix}\frac{dU(\theta,\theta_0)}{d\theta}=A(\theta)U(\theta,\theta_0),\ \ \ U(\theta_0,\theta_0)=\text{Id},\end{equation}
 where $A(\theta)$ is given by the $\ml{C}^1$-linearization hypothesis near $\Lambda$. The solution to this ODE 
 satisfies
 the periodicity
 condition 
 $U(\theta+\ml{P}_\Lambda,\theta_0+\ml{P}_\Lambda)=U(\theta,\theta_0)$. 
 The matrix $M(\theta_0)=U(\theta_0+\ml{P}_\Lambda,\theta_0)$ 
is called the \textbf{monodromy matrix} of the system. Recall that $M(\theta_1)=U(\theta_1,\theta_0)M(\theta_0)U(\theta_1,\theta_0)^{-1}$. 
In particular, its eigenvalues are independent of 
$\theta_0$. We denote by $(\lambda_1,\ldots,\lambda_n)$ the eigenvalues of $M=M(0)$. We write these eigenvalues under an exponential form, 
i.e. $\lambda_j:=e^{\ml{P}_{\Lambda} \mu_j(\Lambda)}$. The $\mu_j(\Lambda)$ are called the \textbf{Floquet exponents of the closed orbit} 
$\{0\}\times(\IR/\ml{P}_{\Lambda}\IZ)$ while their real parts are the so-called \textbf{Lyapunov exponents}.

\subsection{Results from Floquet theory}\label{aaa:floquet} 
According to~\cite[Lemma~3.34, Corollary~3.16]{Te12}, the following Proposition
gives one of the main result from
Floquet theory concerning the
reduction of the solutions to 
the ODE with periodic coefficients
(\ref{e:floquet-matrix})~: 
\begin{prop}\label{p:floquet}
With the conventions of paragraph~\ref{ss:floquet},
the squared matrix $M(\theta_0)^2$ can be put under 
the real exponential form~: 
$$M(\theta_0)^2=e^{2\ml{P}_{\Lambda}\Omega(\theta_0)}$$ with 
$\Omega(\theta_0)$ real valued. 
Moreover, there exists a \emph{real valued} 
matrix $P(\theta_0+\theta,\theta_0)$ which is $2\ml{P}_{\Lambda}$-periodic in $\theta$ and such that, 
for every $\theta$ in $[0,2\ml{P}_\Lambda]$,
\begin{equation}\label{e:formulaU}
U(\theta,\theta_0)=P(\theta,\theta_0)e^{(\theta-\theta_0)\Omega(\theta_0)}.
\end{equation}
\end{prop}
Note that we have to take 
the square of $M(\theta_0)$ in order to take into account the fact that $M(\theta_0)$ may have negative eigenvalues. 
We next use this Proposition to
find some explicit coordinate representation
of the flow near periodic orbits.

\subsection{Writing the flow in local coordinates near closed orbits}\label{ss:coordinates-orbits}

By the assumption that the periodic orbit
$\Lambda$ is hyperbolic and using the $\ml{C}^1$-linearization assumption,
we can fix from Proposition~\ref{p:floquet} a system of $\ml{C}^1$ local coordinates 
$(z,\theta)\in\IR^{n-1}\times(\IR/\ml{P}_{\Lambda}\IZ)$ such that
\begin{equation}\label{e:linearize-flow-periodic-orbit}
 \varphi^t(z,\theta)=\left(P(\theta+t,0)e^{t\Omega_{\Lambda}} P(\theta,0)^{-1}z,\theta+t\ \text{mod} \left(\ml{P}_{\Lambda}\right)\right),
\end{equation}
where $t\mapsto P(\theta+t,0)$ is $2\ml{P}_{\Lambda}$-periodic and where
the
eigenvalues of the matrix $\Omega_{\Lambda}=\Omega(0)$ have a nonzero real part.
 Equivalently, one has
\begin{equation}\label{e:linearize-flow-periodic-orbit-2}\varphi^t(z,\theta)=\left(U(t+\theta,\theta)z,t+\theta\ 
 \text{mod}\left(\ml{P}_{\Lambda}\right)\right),
\end{equation}
where 
\begin{equation}\label{e:UPomegafloquetrelation}
\boxed{U(t+\theta,\theta)=P(\theta+t,\theta)e^{t\Omega(\theta)}=P(\theta+t,0)e^{t\Omega(0)}P(\theta,0)^{-1}}
\end{equation}
is the fundamental solution for the Floquet problem associated with $A(t)$.

 Furthermore, 
we can split $\IR^{n-1}=E_{ss}\oplus E_{uu}$ where $E_{uu}$ 
(resp. $E_{ss}$) corresponds to the eigenvalues of $\Omega_{\Lambda}=\Omega(0)$ 
with positive (resp. negative) real part. 
This allows to rewrite the matrix $\Omega_{\Lambda}$ 
under a block diagonal form, i.e. $\Omega_{\Lambda}=\text{diag}(-\Omega_s,\Omega_u)\in M_{n_s+n_u}(\IR)$ with 
$\Omega_{u/s}$ of size $n_{u/s}$ and having eigenvalues with positive real parts. 
We set $z=(\tilde{x},\tilde{y})$ for the system of coordinates adapted to the above decomposition 
and $(\tilde{\xi},\tilde{\eta})$ the corresponding 
dual coordinates. 
In particular, there are some positive constants 
$0<C_1<C_2$ and $0<\chi_-<\chi_+$ such that, for every $t\geq 0$,
\begin{equation}\label{e:lyapunov-bound-closed-orbit}C_1e^{-t\chi_+}\|(\tilde{x},\tilde{y})\|\leq\|(e^{-t\Omega_{s}}\tilde{x},e^{-t\Omega_{u}}\tilde{y})\|
\leq C_2e^{-t\chi_-}\|(\tilde{x},\tilde{y})\|.
\end{equation}
In the next paragraph, we use the above Floquet coordinates so that, for each $\theta\in \mathbb{R}/\mathcal{P}_\Lambda\mathbb{Z}$, 
$\mathbb{R}^{n-1}$ decomposes as a sum of
stable and unstable subspaces of $M(\theta)$ which gives some coordinates 
representation of the stable and unstable manifolds near the periodic orbits.

\subsection{Coordinate representation of stable and unstable manifolds near periodic orbits}
 
From this expression, one can verify that the unstable and stable manifolds are given
in local coordinates 
\emph{near the closed orbit $\Lambda$}, by
\begin{equation}\label{e:unstablemfd}
W^u(\Lambda):=\left\{\left(P(\theta,0)(0,\tilde{y}),\theta\right):\theta\in\IR/\ml{P}_{\Lambda}\IZ,\ \tilde{y}\in \IR^{n_u}\right\}
\end{equation}
and
\begin{equation}\label{e:stablemfd}
W^s(\Lambda):=\left\{\left(P(\theta,0)(\tilde{x},0),\theta\right):\theta\in\IR/\ml{P}_{\Lambda}\IZ,\ \tilde{x}\in \IR^{n_s}\right\}.
\end{equation}
Similarly, we can verify that the strong unstable and stable manifolds at a given point $(0,0,\theta)$ of the closed orbit $\Lambda$ is
\begin{equation}\label{e:uumfd}
W^{uu}(0,0,\theta):=\left\{\left(P(\theta,0)(0,\tilde{y}),\theta\right):\tilde{y}\in \IR^{n_u}\right\}
\end{equation}
and
\begin{equation}\label{e:ssmfd}
W^{ss}(0,0,\theta):=\left\{\left(P(\theta,0)(\tilde{x},0),\theta\right):\tilde{x}\in \IR^{n_s}\right\}.
\end{equation}
%Our construction requires to understand the topological and dynamical properties of the conormal of these submanifolds -- see paragraph~\ref{ss:conormal} below  
%for the precise definition. 
%Let us give the expression of the conormal of the (strong) unstable manifolds\footnote{The case of stable manifolds is similar.}:
%$$N^*(W^u(\Lambda)):=\left\{\left(P(\theta,0)(0,\tilde{y}),\theta;(P(\theta,0)^T)^{-1}(\tilde{\xi},0),0\right):(0,\tilde{y},\theta;\tilde{\xi},0,0)\in T^*(\IR^{n-1}
%\times\IR/\ml{P}_{\Lambda}\IZ)\right\},$$
%and, for every $\theta$ in $\IR/\ml{P}_{\Lambda}\IZ$,
%$$N^*(W^{uu}(0,0,\theta)):=\left\{\left(P(\theta,0)(0,\tilde{y}),\theta;(P(\theta,0)^T)^{-1}(\tilde{\xi},0),\Theta\right):(0,\tilde{y},\theta;\tilde{\xi},0,\Theta)\in T^*(\IR^{n-1}\times\IR/\ml{P}_{\Lambda}\IZ)\right\}.$$
We next lift the Floquet representation to $T^*M$.

\subsection{The Hamiltonian flow near closed orbits}
We used Floquet theory to write 
the flow in nice coordinates near closed orbits. The purpose of the present paragraph
is to lift the coordinate representation to
cotangent space.
Following the notations of the previous paragraph, we consider the system of coordinates
$(z,\theta;\zeta,\Theta)$ in $T^*\left( \IR^{n-1}\times(\IR/\ml{P}_{\Lambda}\IZ)\right)$.
From equation 
\eqref{e:basicliftcotangentflow}, we need to compute the
element $(d\varphi^t(x)^T)^{-1}(\zeta,\theta)$ in our coordinate system.
A straightforward calculation in the Floquet coordinates $(z,\theta)$ yields 
$$d\varphi^t(z,\theta)=\left(\begin{array}{cc}
U(t+\theta,\theta) & \partial_\theta U(t+\theta,\theta)z\\
0 & 1 
\end{array}  \right) \implies d\varphi^t(z,\theta)^T=\left(\begin{array}{cc}
U(t+\theta,\theta)^T & 0\\
(\partial_\theta U(t+\theta,\theta)z)^T & 1 
\end{array}  \right) $$
where we wrote the $((n-1)+1)\times ((n-1)+1)$ block decomposition of the differential $d\varphi^t$.
This
tells us that
$$((d\varphi^t)^T(z,\theta))^{-1}=\left(\begin{array}{cc}(U(\theta+t,\theta)^{-1})^T & 0\\
-\left(U(\theta+t,\theta)^{-1}\partial_{\theta}\left(U(\theta+t,\theta)z\right)\right)^{T} & 1
\end{array}\right).$$
Fix now a point $(z,\theta;\zeta,\Theta)=(P(\theta,0)(\tilde{x},\tilde{y}),\theta;(P(\theta,0)^T)^{-1}(\tilde{\xi},\tilde{\eta}),\Theta)$ in $T^*(\IR^{n-1}\times\IR/\ml{P}_{\Lambda}\IZ)$.
Applying the previous formula and~\eqref{e:UPomegafloquetrelation} to calculate $(d\varphi^t)^T(z,\theta))^{-1}(\zeta,\Theta)$, 
we find that~:
\begin{eqnarray*}
\left(\begin{array}{cc}(U(\theta+t,\theta)^{-1})^T & 0\\
-\left(U(\theta+t,\theta)^{-1}\partial_{\theta}\left(U(\theta+t,\theta)z\right)\right)^{T} & 1
\end{array}\right)_{(z=P(\theta,0)(\tilde{x},\tilde{y}),\theta)}((P(\theta,0)^T)^{-1}(\tilde{\xi},\tilde{\eta}),\Theta)\\
=\left(\begin{array}{cc}( (P(\theta+t,0)^{-1 })^T(e^{-t\Omega_\Lambda })^TP(\theta,0)^T   & 0\\
-\left(\partial_{\theta}\left(U(\theta+t,\theta)z\right)\right)^{T}\left(U(\theta+t,\theta)^{-1}\right)^T & 1
\end{array}\right)_{(z=P(\theta,0)(\tilde{x},\tilde{y}),\theta)}((P(\theta,0)^T)^{-1}(\tilde{\xi},\tilde{\eta}),\Theta).
\end{eqnarray*}
%Using $\partial_{\theta}U(\theta+t,\theta)z=\partial_{\theta}(P(\theta+t,0)
%e^{t\Omega_\Lambda}P(\theta,0)^{-1}
%\underset{=x}{\underbrace{P(\theta,0)(\tilde{x},\tilde{y})}})=\partial_\theta P(\theta+t,0) (e^{t\Omega_s}\tilde{x},e^{-t\Omega_u}\tilde{y}) $,
%we obtain that 
%\begin{eqnarray*}
%\left(\partial_{\theta}\left(U(\theta+t,\theta)z\right)\right)^{T}\left(U(\theta+t,\theta)^{-1}\right)^TP(\theta,0)^T)^{-1}(\tilde{\xi},\tilde{\eta})\\ =\left(\partial_\theta\left(P(\theta+t,0)\right) (e^{t\Omega_s}\tilde{x},e^{-t\Omega_u}\tilde{y})\right)^T(P(\theta+t,0)^{-1 })^Te^{-t\Omega_\Lambda^T } (\tilde{\xi},\tilde{\eta})
%\end{eqnarray*}
Using the block decomposition of $\Omega_\Lambda=\left(\begin{array}{cc} -\Omega_s & 0\\0 &\Omega_u \end{array} \right)$ with equation~\eqref{e:UPomegafloquetrelation},
this yields~:
\begin{eqnarray}\label{e:tangent-map}\left(\left(d_{P(\theta,0)(\tilde{x},\tilde{y}),\theta}\varphi^t\right)^T\right)^{-1}\left((P(\theta,0)^T)^{-1}(\tilde{\xi},\tilde{\eta}),\Theta\right)\\
\nonumber=\left((P(\theta+t,0)^T)^{-1}(e^{t\Omega_s^T}\tilde{\xi},e^{-t\Omega_u^T}\tilde{\eta}),\Theta+R(t,\tilde{x},\tilde{y},\theta,\tilde{\xi},\tilde{\eta})\right),\end{eqnarray}
where
\begin{equation}\label{e:remainderR}
R(t,\tilde{x},\tilde{y},\theta,\tilde{\xi},\tilde{\eta})=-\left\la\partial_{\theta}(U(\theta+t,\theta))P(\theta,0)(\tilde{x},\tilde{y}),
(P(\theta+t,0)^T)^{-1}(e^{t\Omega_s^T}\tilde{\xi},e^{-t\Omega_u^T}\tilde{\eta})\right\ra
\end{equation}
where $\langle .,. \rangle$ denotes the natural 
duality pairing between $\mathbb{R}^{n-1}$ and 
$\mathbb{R}^{n-1*}$.
We can now write the corresponding Hamiltonian flow:
\begin{equation}\label{e:linearize-hamiltonian-flow-periodic-orbit}
 \Phi^t(z,\theta;\zeta,\Theta)=\left(\varphi^t(z,\theta); (U(t+\theta,\theta)^T)^{-1}\zeta,\Theta+R(t,z,\theta,\zeta)\right),
\end{equation}
with $R(t,z;\theta,\zeta)$ which is defined above. Equivalently, one has
\begin{equation}\label{e:linearize-hamiltonian-flow-periodic-orbit-2}
 \Phi^t(z,\theta;\zeta,\Theta)=\left(\varphi^t(z,\theta); (P(\theta+t,0)^T)^{-1}e^{-t\Omega_{\Lambda}^T} P(\theta,0)^T(\zeta),\Theta+R(t,z,\theta,\zeta)\right),
\end{equation}
where $R$ can be split as follows~:
\begin{eqnarray*}
 R(t,\tilde{x},\tilde{y},\theta,\tilde{\xi},\tilde{\eta}) & = &-\left\la\partial_{\theta}P(\theta+t,0)(e^{-t\Omega_s}\tilde{x},e^{t\Omega_u}\tilde{y}),
(P(\theta+t,0)^T)^{-1}(e^{t\Omega_s^T}\tilde{\xi},e^{-t\Omega_u^T}\tilde{\eta})\right\ra\\
& + &\left\la\partial_{\theta}P(\theta,0)(\tilde{x},\tilde{y}),
(P(\theta,0)^T)^{-1}(\tilde{\xi},\tilde{\eta})\right\ra.
\end{eqnarray*}

\section{Conormals of stable and unstable manifolds}\label{s:conormal}

Smale's Theorem~\ref{t:smale} describes the closure of any unstable manifold $W^u(\Lambda)$ of a Morse-Smale flow. In particular, 
this theorem shows that each 
closure 
$\overline{W^u(\Lambda)}$ has a stratified structure whose strata are given by unstable manifolds. In the 
following, we need to understand the fine properties of the conormal $N^*(W^u(\Lambda))$,
in particular we need to describe precisely 
its closure $\overline{N^*(W^u(\Lambda))}$ inside 
$T^*M$. Analyzing the conormals of 
unstable manifolds is the core of this article and this is the content of this section. The main results of this section are 
Theorems~\ref{t:compactness} and~\ref{t:attractor}. Note that these results were already obtained in the case of Morse-Smale 
gradient flows in~\cite{DaRi16}. The main novelty here is that we allow closed orbits in our analysis.

\subsection{Conormals of the unstable/stable manifolds}\label{ss:conormal}

Given a smooth submanifold $S$ inside $M$, one can define its \emph{conormal} (bundle) as follows~:
$$N^*S:=\{(x;\xi)\in T^*M:\ x\in S,\ \xi\neq 0,\ \text{and}\ \forall v\in T_xS,\ \xi(v)=0\}.$$
Observe that, for a submanifold of dimension $n$, the conormal is empty. For instance, this is the case for the unstable manifold of an 
expanding fixed point. Here, the relevant sets for the lifted dynamics will be
$$\Sigma_{uu}:=\bigcup_{j=1}^K N^*(W^s(\Lambda_j))\cap S^*M\ \text{and}\ \Sigma_{ss}:=\bigcup_{j=1}^K N^*(W^u(\Lambda_j))\cap S^*M.$$
In the forthcoming Lemma \ref{l:limit-attractor}, we will
see that the set $\bigcup_{j=1}^K N^*(W^u(\Lambda_j))\cap S^*M$ is in fact an attractor
for the lifted dynamics $(\tilde{\Phi}^t)_t$, this is why it is denoted
by
$\Sigma_{ss}$ to emphasize it is the
stable set for the flow $(\tilde{\Phi}^t)_t$ in $S^*M$. From appendix~\ref{a:hyperbolic}, we also know that the unstable (resp. stable) manifolds are invariantly fibered by smooth submanifolds, i.e.
$$\forall 1\leq j\leq K,\ W^u(\Lambda_j):=\bigcup_{x\in\Lambda_j} W^{uu}(x),\ \text{and}\ W^s(\Lambda_j):=\bigcup_{x\in\Lambda_j} W^{ss}(x).$$
We then define 
$$\Sigma_{u}:=\bigcup_{j=1}^K \bigcup_{x\in\Lambda_j}N^*(W^{ss}(x))\cap S^*M\ \text{and}\ \Sigma_{s}:=\bigcup_{j=1}^K\bigcup_{x\in\Lambda_j} N^*(W^{uu}(x))\cap S^*M.$$
Note that, in the case of a fixed point $\Lambda_j=\{x\}$, $W^{ss/uu}(x)$ coincides with $W^{s/u}(\Lambda_j)$. Finally, we can translate the Smale transversality assumption from definition~\ref{d:morsesmale} in this
microlocal setting~:
\begin{equation}\label{e:transversality}
\Sigma_{uu}\cap\Sigma_{ss}=\Sigma_{uu}\cap\Sigma_s=\Sigma_{ss}\cap\Sigma_u=\emptyset.
\end{equation}

\subsection{Coordinate representation of $N^*(W^u(\Lambda_j))$ near $\Lambda_j$}
\label{ss:unstable-conormal-coord} We previously defined the set $\Sigma_{ss}$ as the union of conormals of
unstable manifolds intersected with $S^*M$. 
In order to have a concrete representation of these sets, let us write the conormal bundle
$N^*(W^u(\Lambda_j))$ in the local coordinates we have introduced in paragraphs~\ref{ss:coordinates},~\ref{ss:coordinates-orbits} 
and~\ref{aaa:floquet}. Note already that this does not represent the whole set $\Sigma_{ss}$ near $\Lambda_j$ as one can also find 
points inside $N^*(W^u(\Lambda_i))$ with $i\neq j$ arbitrarily close to $\Lambda_j$. Yet, we will only need to use the exact representation 
of $N^*(W^u(\Lambda_j))$ near $\Lambda_j$ in our proofs.

In the $(\tilde{x},\tilde{y})$ coordinates near a fixed point $\Lambda_j$, using the coordinates system from subsection 
\ref{ss:coordinates}, the set $N^*(W^u(\Lambda_j))$ can be represented as~:
\begin{equation}\label{e:conormunstable}
\left\{(0,\tilde{y};\tilde{\xi},0):\tilde{y}\in\IR^{n_u},\ \tilde{\xi}\in\IR^{n_s}\setminus\{0\}\right\}.
\end{equation}
Near a closed orbit $\Lambda_j$, using the coordinates system from paragraph~\ref{ss:coordinates-orbits}, $N^*(W^u(\Lambda_j))$ can be written as:

\begin{equation}\label{e:conormunstable-orbit}
\left\{\left(P(\theta,0)(0,\tilde{y}),\theta;(P(\theta,0)^T)^{-1}(\tilde{\xi},0),-\la \partial_\theta P(\theta,0)(0,\tilde{y}),
(P(\theta,0)^{T})^{-1}(\tilde{\xi},0)\ra\right):(*)\right\},
\end{equation}
where $(*)$ means $\tilde{y}\in\IR^{n_u}$, $\theta\in \IR/\ml{P}_{\Lambda}\IZ$ and $\tilde{\xi}\in\IR^{n_s}\setminus\{0\}$. To see this, recall that he weak unstable manifold reads~:
$$W^u(\Lambda_j)=\{ (P(\theta,0)(0,\tilde{y}),\theta) : \tilde{y}\in\mathbb{R}^{n_u}, \theta\in \mathbb{R}/\mathcal{P}_{\Lambda_j}\mathbb{Z} \}.$$
Hence tangent vectors to $W^u(\Lambda_j)$ are generated by vectors of the form 
$ \left(P(\theta,0)(0,Y),0 \right)$ and $(\partial_{\theta}P(\theta,0)(0,\tilde{y}),1)$ for some $Y\in\IR^{n_u}$.
So $(P(\theta,0)(0,\tilde{y}),\theta;\zeta,\Theta  ) \in N^*\left(W^u(\Lambda_j)\right)$ iff, for every $Y\in\IR^{n_u}$,
\begin{eqnarray*}
\la \partial_\theta P(\theta,0)(0,\tilde{y}),\zeta\ra+\Theta=0\\
\la P(\theta,0)(0,Y),\zeta \ra=0.
\end{eqnarray*} 
So $\zeta=\left((P(\theta,0)^{T})^{-1}(\tilde{\xi},0) \right)$ by the second equation and the first equation implies
$\Theta$ satisfies $\la \partial_\theta P(\theta,0)(0,\tilde{y}),(P(\theta,0)^{T})^{-1}(\tilde{\xi},0)\ra+\Theta=0$.
%
%So we obtain~:
%\begin{eqnarray*}
%\boxed{N^*\left(W^u(\Lambda_j)\right)=\{ (P(\theta,0)(0,\tilde{y}),\theta;P(\theta,0)^{T})^{-1}(\tilde{\xi},0),-\la \partial_\theta P(\theta,0)(0,\tilde{y}),(P(\theta,0)^{T})^{-1}(\tilde{\xi},0)\ra | \tilde{y}\in , \tilde{\xi}\in  \}.}
%\end{eqnarray*}
%Near a closed orbit $\Lambda_j$, using the coordinates system from paragraph~\ref{ss:coordinates-orbits}, $N^*(W^u(\Lambda_j))$ can be written as:
%{\small
%\begin{eqnarray}\label{e:conormunstable-orbit}
%\{\left(P(\theta,0)(0,\tilde{y}),\theta;(P(\theta,0)^T)^{-1}(\tilde{\xi},0),-\la \partial_\theta P(\theta,0)(0,\tilde{y}),
%(P(\theta,0)^{T})^{-1}(\tilde{\xi},0)\ra\right)\\
%\nonumber :\tilde{y}\in\IR^{n_u},\ \theta\in \IR/\ml{P}_{\Lambda}\IZ,\ \tilde{\xi}\in\IR^{n_s}\setminus\{0\}
%\}.
%\end{eqnarray}
%}

\subsection{Coordinate representation of $\cup_{x\in\Lambda_j}N^*(W^{uu}(x))$ near $\Lambda_j$}
\label{ss:unstable-conormal-coord-2} It is also useful to represent these objects in local coordinates. In the case of a fixed point, 
the expression is the same 
as in subsection~\ref{ss:unstable-conormal-coord}. Yet, in the case of a closed orbit, there is a small difference. Namely, 
%for every $\theta$ in $\IR/\ml{P}_{\Lambda}\IZ$, 
$\cup_{\theta}N^*(W^{uu}(0,0,\theta))$ can be represented near $\Lambda_j$ as~:
$$\left\{\left(P(\theta,0)(0,\tilde{y}),\theta;(P(\theta,0)^T)^{-1}(\tilde{\xi},0),\Theta\right):
\tilde{y}\in\IR^{n_u},\ \theta\in \IR/\ml{P}_{\Lambda_j}\IZ,\ (\tilde{\xi},\Theta)\in\IR^{n_s+1}\setminus\{0\}\right\}.$$
Again, it only represents the contribution of $\Lambda$ as there may be points associated to the unstable manifold of $\Lambda_i$ 
(with $i\neq j$) which are arbitrarily close to $\Lambda_j$. Finally, observe that one has $\Sigma_{uu}\subset\Sigma_u$ and $\Sigma_{ss}\subset\Sigma_s$.

\subsection{The sets $\Sigma_u,\Sigma_{uu}$ as attractors for the backward flow on the unit cotangent bundle.}

Let us start our description of the dynamics in the cotangent bundle with the following Lemma~:
\begin{lemm}\label{l:limit-attractor} Let $\varphi^t$ be a Morse-Smale flow which is $\ml{C}^1$-linearizable. 
Then, one has~:
\begin{enumerate}
\item $\forall (x;\xi)\in S^*M\setminus\Sigma_{ss},\ \lim_{t\rightarrow-\infty}d_{S^*M}(\tilde{\Phi}^t(x;\xi),\Sigma_u)=0,$
\item $\forall (x;\xi)\in S^*M\setminus\Sigma_{s},\ \lim_{t\rightarrow-\infty}d_{S^*M}(\tilde{\Phi}^t(x;\xi),\Sigma_{uu})=0.$
\end{enumerate}
\end{lemm}
By reversing the time, one can verify that the same properties hold if we intertwine the roles of $s$ and $u$. 
We shall divide the proof of this Lemma in two parts and make use of the coordinate representation of 
paragraphs~\ref{ss:unstable-conormal-coord} and~\ref{ss:unstable-conormal-coord-2}.

\subsubsection{Proof of part $1$ of Lemma~\ref{l:limit-attractor}}

Let $\rho$ be an element in $S^*M\setminus\Sigma_{ss}$. According to Lemma~\ref{l:partition}, there exists a 
unique $1\leq i\leq K$ such that the projection 
of $\rho$ on $M$ belongs to $W^u(\Lambda_i)$. 
Up to applying the flow in backward 
times, we can suppose that $\rho\in S^*M\setminus\Sigma_{ss}$ belongs to the linearizing chart near $\Lambda_i$ which was defined in 
section \ref{s:cotangent}. 
Let us start with 
the case where $\Lambda_i$ is a fixed point of the flow. Using our local coordinates, we can then write $\rho=(0,\tilde{y};\tilde{\xi},\tilde{\eta})$ with 
$\tilde{\eta}\neq 0$ as $\rho\notin N^*(W^u(\Lambda_i))$ by definition. Recall now that the flow $\Phi^t$ reads in these local coordinates:
$$\Phi^{t}(0,\tilde{y};\tilde{\xi},\tilde{\eta})=(0,e^{t\Omega_u}\tilde{y};e^{t\Omega_s^T}\tilde{\xi},e^{-t\Omega_u^T}\tilde{\eta}),$$
where all the eigenvalues of $\Omega^{u/s}$ have positive real parts. In order to conclude, we need to normalize the cotangent vector, i.e. 
consider $(e^{t\Omega_s^T}\tilde{\xi},e^{-t\Omega_u^T}\tilde{\eta})/\|(e^{t\Omega_s^T}\tilde{\xi},e^{-t\Omega_u^T}\tilde{\eta})\|$. All the norms on 
$\IR^n$ being equivalent and by the inequality~\eqref{e:lyapunov-bound-fixed-point} 
and the fact that $\tilde{\eta}\neq 0$, we find that the $\tilde{\xi}$ 
component of this normalized covector tends to $0$ as $t\rightarrow-\infty$. In particular,
this implies that any accumulation point of $\tilde{\Phi}^{t}(\rho)$ (as $t\rightarrow-\infty$) 
is of the form $(0,0;0,\tilde{\eta}')\neq 0$. This implies that any accumulation point of 
$\tilde{\Phi}^{t}(\rho)$ belongs to $N^*(W^s(\Lambda_i))\subset\Sigma_u$ by the results from paragraph~\ref{ss:unstable-conormal-coord}.

Consider now the case where $\Lambda_i$ is a closed orbit. 
Again, up to applying the flow in backward times, 
we may suppose that $\rho$ belongs to the neighborhood 
of $\Lambda_i$ where we have our linearizing coordinates. 
In this chart, $\rho$ can be represented as
$$\rho=(P(\theta,0)(0,\tilde{y}),\theta;(P(\theta,0)^T)^{-1}(\tilde{\xi},\tilde{\eta}),\Theta),$$
with either $\tilde{\eta}\neq 0$, or $\tilde{\eta}=0$ and $\Theta+ \la \partial_\theta P(\theta,0)(0,\tilde{y}),(P(\theta,0)^{T})^{-1}(\tilde{\xi},0)\ra \neq 0$ 
%$(P(\theta,0)^T)^{-1}(\tilde{\xi},\tilde{\eta}),\Theta) \neq (P(\theta,0)^T)^{-1}(\tilde{\xi},0),-\la \partial_\theta P(\theta,0)(0,\tilde{y}),(P(\theta,0)^{T})^{-1}(\tilde{\xi},0)\ra$
as $\rho$ does not belong to $\Sigma_{ss}$ -- see~equation~\eqref{e:conormunstable-orbit}.
%with either $\tilde{\eta}\neq 0$ or $\Theta\neq 0$ as $\rho$ does not belong to $\Sigma_{ss}$ -- see~equation~\eqref{e:conormunstable-orbit}. 

We begin with the case $\tilde{\eta}\neq 0$.
What we have to understand is the asymptotic behaviour of 
$$\left(\left(d_{P(\theta,0)(0,\tilde{y}),\theta}\varphi^t\right)^T\right)^{-1}\left((P(\theta,0)^T)^{-1}(\tilde{\xi},\tilde{\eta}),\Theta\right)$$
as $t\rightarrow -\infty$. This quantity corresponds to the evolution of the cotangent component of $\rho$ under the Hamiltonian flow~$\Phi^t$. Recall that an 
explicit expression for this quantity was given in equation~\eqref{e:tangent-map} which we reproduce here in the case
where $\tilde{x}=0$~:
\begin{eqnarray*}
\left(\left(d_{P(\theta,0)(0,\tilde{y}),\theta}\varphi^t\right)^T\right)^{-1}\left((P(\theta,0)^T)^{-1}(\tilde{\xi},\tilde{\eta}),\Theta\right)\\
=\left((P(\theta+t,0)^T)^{-1}(e^{t\Omega_s^T}\tilde{\xi},e^{-t\Omega_u^T}\tilde{\eta}),\Theta+R(t,0,\tilde{y},\theta,\tilde{\xi},\tilde{\eta})\right).\end{eqnarray*}

Combined with the explicit expression of the remainder $R$ given in~\eqref{e:linearize-hamiltonian-flow-periodic-orbit-2}, 
one finds that there exist some $C,\chi_0>0$ such that 
$$\forall t\leq 0,\ \vert R(t,0,\tilde{y},\theta,\tilde{\xi},\tilde{\eta})\vert \leqslant C\left(1+e^{t\chi_0}\Vert \tilde{y} \Vert 
\Vert e^{-t\Omega_u^T}\tilde{\eta} \Vert\right).$$ This expression tells us that the cotangent component is in fact 
of the form
$$\left((P(\theta+t,0)^T)^{-1}(e^{t\Omega_s^T}\tilde{\xi},e^{-t\Omega_u^T}\tilde{\eta}),\mathcal{O}(1)+o(1)\|e^{-t\Omega_u^T}\tilde{\eta}\|)\right),\quad \text{as}\quad t\rightarrow-\infty.$$
According to~\eqref{e:lyapunov-bound-closed-orbit}, one knows that $\|e^{-t\Omega_u^T}\tilde{\eta}\|$ is exponentially large as $t\rightarrow-\infty$ since $\tilde{\eta}\neq 0$ by assumption. 
From this explicit expression and from the fact that all the matrices have eigenvalues with positive real parts, 
we find that the leading contribution comes from the term $(P(\theta+t,0)^T)^{-1}(0,e^{-t\Omega_u^T}\tilde{\eta})$ as $t$ tends to $-\infty$. Hence, 
the normalized cotangent component of $\tilde{\Phi}^{t}(\rho)$ will approach the set $\cup_{\theta}\{((P(\theta,0)^T)^{-1}(0,\tilde{\eta}),0):\tilde{\eta}\neq 0\}$ 
as $t$ tends to $-\infty$. It implies that any accumulation point of $\tilde{\Phi}^{t}(\rho)$ (as $t\rightarrow-\infty$) 
is a point inside $\cup_{x\in\Lambda_i}N^*(W^{ss}(x))\cap S^*M\subset\Sigma_u$ by paragraph~\eqref{ss:unstable-conormal-coord-2}. 
%From subsection~\ref{ss:unstable-conormal-coord-2}, it exactly means 
%that $\tilde{\Phi}^{t}(\rho)$ tends to $\Sigma_{uu}\subset\Sigma_u$ as $t\rightarrow -\infty$. 

In the case $\tilde{\eta}=0$, we must have $\Theta\neq - \la \partial_\theta P(\theta,0)(0,\tilde{y}),(P(\theta,0)^{T})^{-1}(\tilde{\xi},0)\ra$ and the cotangent 
component of $\Phi^{t}(\rho)$ can be expressed in local coordinates as follows, as $t\rightarrow-\infty$:
$$\left((P(\theta+t,0)^T)^{-1}(e^{t\Omega_s^T}\tilde{\xi},0),\Theta+ \la \partial_\theta P(\theta,0)(0,\tilde{y}),(P(\theta,0)^{T})^{-1}(\tilde{\xi},0)\ra
+o(1)\|e^{t\Omega_s^T}\tilde{\xi}\|\right),$$
where we used the expression given in~\eqref{e:linearize-hamiltonian-flow-periodic-orbit-2} one more time. 
Since all 
the eigenvalues of $\Omega_s^T$ have eigenvalues with positive real parts, 
we find that any accumulation point (as $t\rightarrow-\infty$) of the cotangent component will be of the form 
$(0,\Theta')\neq 0$. Again, if we consider the normalized version, the limit vector will be $(0,1)$. 
Using one more time the coordinate representation of $\Sigma_u$ from subsection~\ref{ss:unstable-conormal-coord-2}, we can conclude that $\tilde{\Phi}^{t}(\rho)$ 
tends to $\Sigma_u$ as $t\rightarrow-\infty$.

\subsubsection{Proof of part $2$ of Lemma~\ref{l:limit-attractor}}

In order to prove the second part of the Lemma, we follow the same strategy. 
We fix a point $\rho\in S^*M \setminus\Sigma_{s}$ whose projection on 
$M$ belongs to some unstable
manifold 
$W^u(\Lambda_i)$. Up to applying the flow in backward times, we can one more time suppose that $\rho$ belongs to the linearizing chart 
near $\Lambda_i$. In the case where $\Lambda_i$ is a fixed point, the proof is basically the same. In the case of a closed orbit, we now write $\rho$ as
$$\rho=(P(\theta,0)(0,\tilde{y}),\theta;
(P(\theta,0)^T)^{-1}(\tilde{\xi},\tilde{\eta}),\Theta),$$
with $\tilde{\eta}\neq 0$ 
and not only $(\tilde{\eta},\Theta+\la \partial_\theta P(\theta,0)(0,\tilde{y}),(P(\theta,0)^{T})^{-1}(\tilde{\xi},0)\ra)\neq 0$. 
Hence, the same argument
as above
allows to conclude one more time.

\subsection{Main results on compactness and stable neighborhoods.}

We now state the main results on the geometric and topological properties of the sets we have just defined. The proofs 
will be given in the next two sections. The first Theorem deals with
the compactness of the sets $\Sigma_*$ 
previously defined~:
\begin{theo}[Compactness Theorem]\label{t:compactness} Let $\varphi^t$ be a Morse-Smale flow which is $\ml{C}^1$-linearizable. 
Then $\Sigma_u$, $\Sigma_{uu}$, $\Sigma_s$ and $\Sigma_{ss}$ are compact subsets of $S^*M$.
\end{theo}
The proof of this result was already given in~\cite[Lemma~3.6]{DaRi16} in the particular case of Morse-Smale gradient flows satisfying 
a certain (generic) linearization property. We shall give in section~\ref{s:proofcompact} a proof which is valid for \emph{any Morse-Smale flow} (including 
of course the case of gradient flows) satisfying also certain (generic) linearization property. The property of being $\ml{C}^{1}$-linearizable seems 
crucial in our proof as it allows us to control the asymptotic behaviour of cotangent vectors under the flow near a critical element $\Lambda_j$. 
In fact, having a $\ml{C}^1$-chart enables us to use the local expressions of paragraphs~\ref{ss:coordinates},~\ref{ss:coordinates-orbits} and~\ref{aaa:floquet} 
for the Hamiltonian flow. Using only
the hyperbolicity\footnote{For instance, hyperbolicity allows us to use the 
Grobman-Hartman Theorem but it only provides a $\ml{C}^0$-chart.} 
at $\Lambda_j$ does not seem to be enough to control the asymptotic behaviour of cotangent vectors 
outside $\Lambda_j$ in our proof of compactness.

Note that the statement of this theorem is in certain cases trivial. 
Take for instance the gradient flow associated 
with the height function on the $2$-sphere endowed with its 
canonical metric. In that case, the flow has two critical points: the north pole 
$\mathbf{n}$ and the south pole $\mathbf{s}$. Then, $\Sigma_{uu}=\Sigma_{u}$ is equal to 
the set $\{(\mathbf{s},\xi)\in S^*M: \|\xi\|_{\mathbf{s}}=1\}$ which is obviously compact.
Coming back to the general case, 
the situation may be subtle as illustrated by the following example on
the plane $\IR^2$.
\begin{ex}
 In the $(x,y)$-plane, consider 
the curve $Y=\{(x,\sin(\frac{1}{x})) ; x\in \mathbb{R}_{\geqslant 0} \}$ which is the graph of the function
$\sin(\frac{1}{x})$ 
and the vertical line $X=\{x=0\}$. 
Then $X\cup Y$ is closed since $\overline{Y}\setminus Y\subset X$, yet
the conormal $N^*X=\{(0,y;\xi,0);y\in \IR,\xi\in \IR \}$ is not contained in 
$\overline{N^*Y}$ since the curve $Y$ will oscillate near $x=0$, hence conormal covectors
to $Y$ will not converge to a fixed codirection $(\xi,0)$. 
Therefore $N^*X\cup N^*Y$ is not a closed, conical subset in $T^*\IR^2$ although
$X\cup Y$ was closed.
% the two submanifolds:
% $$\IS^1:=\{(x,y)\in\IR^2: x^2+y^2=1\},$$
% and
% $$\ml{M}:=\left\{(r\cos (r),r\sin(r)): 0\leq r<1\right\}.$$
% One can verify that $N^*(\IS^1)\cup N^*(\ml{M})$ is indeed a closed subset of $T^*\IR^2$. Yet, we can perturb a little bit $\ml{M}$ by taking a curve $\tilde{\ml{M}}$ 
% that is oscillating around $\ml{M}$ as $r$ goes to $1$. In fact, we can construct $\tilde{\ml{M}}$ such that the closure of $N^*(\tilde{\ml{M}})$ is equal to 
% $$N^*(\tilde{\ml{M}})\cup\{(z,\zeta)\in T^*\IS^1\backslash 0\}.$$
% In particular, $N^*(\IS^1)\cup N^*(\ml{M})$ is not a closed subset. 
This shows that taking the union of the conormals of two submanifolds may not give rise to a closed 
subset even if the union of the two submanifolds is closed itself.
\end{ex}

Theorem~\ref{t:compactness} described the topological properties of the subsets $\Sigma_*$ and we will now turn to more dynamical properties. 
More precisely, our next result is that the sets $\Sigma_*$ are attractors or repellers 
of the flow $\tilde{\Phi}^t$:
\begin{theo}[Stable neighborhood Theorem]\label{t:attractor} Let $\varphi^t$ be a Morse-Smale flow which is $\ml{C}^1$-linearizable. Let $\eps>0$. Then, there exists 
an open neighborhood $V^{ss}$ (resp. $V^s$) of $\Sigma_{ss}$ (resp. $\Sigma_s$) inside $S^*M$ all of whose points are at a distance $\leq\eps$ of $\Sigma_{ss}$ (resp. $\Sigma_s$) 
and such that
$\forall t\geq 0,\ \tilde{\Phi}^t(V^{ss})\subset V^{ss},$ (resp.
$\forall t\geq 0,\ \tilde{\Phi}^t(V^{s})\subset V^{s}$).
\end{theo}
If we replace $s$ by $u$, the same conclusion holds except that we have to replace positive times by negative ones. This Theorem may be thought 
of as a ``monotonic'' version of Lemma~\ref{l:limit-attractor}. Like for the property of compactness, 
this result was already proved in~\cite{DaRi16} in the particular case of Morse-Smale \emph{gradient} flows 
satisfying certain linearization 
properties given for instance by the Sternberg-Chen Theorem. Again, we prove that the extension to more general Morse-Smale flows is still true.

We shall now devote the next two sections 
to the proofs of these two Theorems.

\section{Proof of the compactness Theorem~\ref{t:compactness}}
\label{s:proofcompact}
It is sufficient to prove that $\Sigma_s$ and $\Sigma_{ss}$ are compact. The other 
cases follow by reversing the time.
The proof will proceed by a contradiction argument and it is based on 
an important technical Lemma that we will  
present in the next paragraph.

\subsection{A technical Lemma.}

The following result generalizes to the cotangent framework earlier results of Smale~\cite{Sm60}:
\begin{lemm}[key technical Lemma]
\label{l:keytechnicallemma}
Let $(z_\infty;\zeta_\infty)$ be some element 
of $S^*M$ such that 
$(z_\infty;\zeta_\infty)\notin \Sigma_{ss}$ (resp $\Sigma_s$) and 
$z_\infty\in W^u(\Lambda_j)$ for 
some elementary critical element $\Lambda_j$. 
Let $(z_m;\zeta_m) \rightarrow (z_\infty;\zeta_\infty)$ 
be a  
sequence in $S^*M$ 
such that, for every $m\geq 0$, $z_m\in W^u(\Lambda_i)$ for some fixed $\Lambda_i$. Then, one has:
\begin{itemize}
 \item either $i=j$ and $(z_m;\zeta_m)$ does not belong to $\Sigma_{ss}$ (resp. $\Sigma_s$) for $m$ large enough;
 \item or $i\neq j$ and there exists a convergent subsequence $(z^{(1)}_{\phi(m)};\zeta^{(1)}_{\phi(m)})_{m\geq 0}, \phi:\mathbb{Z}_+\mapsto \mathbb{Z}_+$ injective, with the following properties~:
\begin{enumerate}
\item $(z^{(1)}_{\phi(m)};\zeta^{(1)}_{\phi(m)})$ belongs to the integral curve $\{\tilde{\Phi}^{t}(z_{\phi(m)};\zeta_{\phi(m)})\vert t\in \mathbb{R}\}$,
in particular $z^{(1)}_{\phi(m)}\in W^u(\Lambda_i)$ for every $m\in \mathbb{N}$,
\item $\lim_{m\rightarrow +\infty} (z^{(1)}_{\phi(m)};\zeta^{(1)}_{\phi(m)})=(z_\infty^{(1)};\zeta_\infty^{(1)}) \in \cup_{x\in\Lambda_j}N^*(W^{ss}(x))\cap S^*M$ 
%(resp. $\Sigma_s$),
%\in\bigcup_{x\in\Lambda_j}N^*(W^{ss}(x))\cap S^*M $ 
(resp. $N^*(W^s(\Lambda_j))$),  
\item $z_\infty^{(1)}\in W^s(\Lambda_j)\setminus\Lambda_j$.
\end{enumerate} 
\end{itemize}
\end{lemm}

In particular, in the second case, we can conclude from~\eqref{e:transversality} that the new limit point does not belong to $\Sigma_{ss}$ (resp. $\Sigma_s$).

Let us first observe that the case $i=j$ is easy to deal with. Indeed, we can fix $T>0$ large enough to ensure that 
$\varphi^{-T}(z_{\infty})$ belongs to the linearizing chart near $\Lambda_i$. By continuity of $\varphi^{-T}(.)$ for fixed $T\in \mathbb{R}$, 
we know that for $m\geq 1$ large 
enough, $\varphi^{-T}(z_{m})$ also belongs to this chart. Then, the explicit expressions of the conormals in these local 
coordinates from subsections~\ref{ss:unstable-conormal-coord} and~\ref{ss:unstable-conormal-coord-2} gives the conclusion as the limit point 
$\tilde{\Phi}^{-T}(z_\infty;\zeta_\infty)$ does not belong to $\Sigma_{ss}$ (resp. $\Sigma_s$).

Hence, the main difficulty lies in the case where $i\neq j$ that will be divided in two subcases~: (1) $\Lambda_j$ is a fixed point and (2)
$\Lambda_j$ is a periodic orbit.

\subsubsection{$\Lambda_j$ is a fixed point}
Up to applying the flow $\varphi^t$ in backward times, we can suppose that $z_{\infty}$ belongs to the 
linearizing chart near $\Lambda_j$. Moreover, by continuity of the flow, we can suppose that for $m$ large enough, $z_m$ also belongs to this 
linearizing chart. Fix two small enough $\delta_1,\delta>0$ with $\delta_1\gg\delta$. As $(z_{\infty};\zeta_{\infty})$ 
does not belong to $\Sigma_{ss}$ (hence
$(z_{\infty};\zeta_{\infty})\notin N^*(W^u(\Lambda_j))$),
Lemma~\ref{l:limit-attractor} tells us that 
$\tilde{\Phi}^t(z_{\infty};\zeta_{\infty})$ will be attracted by
$\Sigma_u$ when $t\rightarrow -\infty$. We also have $\lim_{t\rightarrow -\infty}\varphi^{-t}(z_\infty)=\Lambda_j$
since $z_\infty\in W^u(\Lambda_j) $. Therefore,  without loss of generality,
we may apply the flow $\tilde{\Phi}^t$ in backward times to $(z_{\infty};\zeta_{\infty})$ to ensure that $(z_{\infty};\zeta_{\infty})$ is at a 
distance $\leq \delta$ of $N^*(W^s(\Lambda_j))\cap S^*M$.
%Beware this does not mean that $z_\infty\in W^s(\Lambda_j)$!!!

Again, by continuity of $\tilde{\Phi}^{-t}(.)$ acting on $S^*M$, 
we find that, for $m$ large enough, we may assume
that $(z_m;\zeta_m)$ is also at a distance less than $2\delta$ from $N^*(W^s(\Lambda_j))\cap S^*M$. 
Write now the expression of these points in local coordinates:
$$(z_m;\zeta_m)=(\tilde{x}_m,\tilde{y}_m;\tilde{\xi}_m,\tilde{\eta}_m)\ \text{and}\ (z_{\infty};\zeta_{\infty})
=(0,\tilde{y}_{\infty};\tilde{\xi}_{\infty},\tilde{\eta}_{\infty}).$$
Observe that, as $i\neq j$ and $z_m\in W^u(\Lambda_i)$, 
one necessarily has $\tilde{x}_m\neq 0$ for every $m$ large enough. 
Moreover, as $(z_m;\zeta_m)$ is 
within a distance $\delta$ 
from $N^*(W^s(\Lambda_j))\cap S^*M$, 
we know that $\|\tilde{\eta}_m\|$ is uniformly bounded from below by a 
positive constant. 
Apply now the flow $\tilde{\Phi}^{-t}$ for some 
positive $t$ to the sequence $(z_m;\zeta_m)$. In local coordinates, this reads
$$\tilde{\Phi}^{-t}(z_m;\zeta_m)=\left(e^{t\Omega_s}\tilde{x}_m,e^{-t\Omega_u}\tilde{y}_m;
\frac{(e^{-t\Omega_s^T}\tilde{\xi}_m,e^{t\Omega_u^T}\tilde{\eta}_m)}{\|(e^{-t\Omega_s^T}\tilde{\xi}_m,e^{t\Omega_u^T}\tilde{\eta}_m)\|}\right).$$
%Fix some small parameter $\delta_1>0$ (slightly larger than $\delta$).
For all $m\in \mathbb{N}$, we choose some
$T_m\in \IR$ large enough to 
ensure that 
$\delta_1\leq \|e^{T_m\Omega_s}\tilde{x}_m\|\leq 2\delta_1$. 
Precisely, it means that one has to take 
$T_m$ of order $|\log \|\tilde{x}_m\||$. 
%Intuitively, this means that when $z_m$ tends to 
%$W^s(\Lambda_j)$, the component $\tilde{x}_m$ tends to zero 
%and $T_m\rightarrow +\infty$.

We define 
a new sequence (see figure~\ref{f:compactness}):
$$(z_m^{(1)};\zeta_m^{(1)})=\tilde{\Phi}^{-T_m}(z_m;\zeta_m)=\left(e^{T_m\Omega_s}\tilde{x}_m,e^{-T_m\Omega_u}\tilde{y}_m;
\frac{(e^{-T_m\Omega_s^T}\tilde{\xi}_m,e^{T_m\Omega_u^T}\tilde{\eta}_m)}{\|(e^{-T_m\Omega_s^T}\tilde{\xi}_m,e^{T_m\Omega_u^T}\tilde{\eta}_m)\|}\right).$$
\begin{figure}[ht]\label{f:compactness}
\includegraphics[width=10cm, height=7cm]{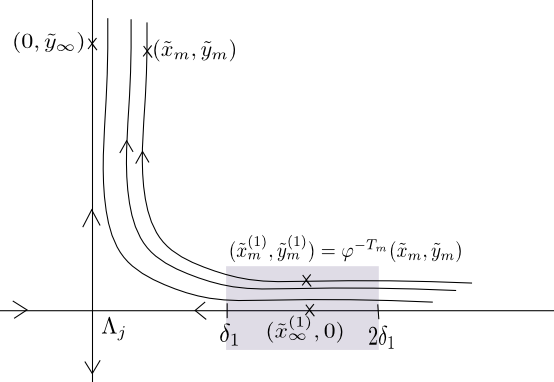}
\centering
\caption{Construction of the new sequence}
\end{figure}
For every $m$, $z_m^{(1)}=\varphi^{-T_m}(z_m)$ belongs to $W^u(\Lambda_i)$ by assumption. 
We know that $\|\tilde{x}_m\|$ goes to $0$ as $m$ tends to $+\infty$. In particular, we will have $e^{-T_m\Omega_u}\tilde{y}_m\rightarrow 0$ as $m$ 
tends to $+\infty$.
Since $\| e^{T_m\Omega_s}\tilde{x}_m\|$ is constrained to be in the interval $[\delta_1,2\delta_1]$, we can extract a subsequence
so that $e^{T_m\Omega_s}\tilde{x}_m$ converges to some value 
$\tilde{x}_{\infty}^{(1)}\neq 0$. Moreover, up to another extraction,  
$\frac{(e^{-T_m\Omega_s^T}\tilde{\xi}_m,e^{T_m\Omega_u^T}\tilde{\eta}_m)}{\|(e^{-T_m\Omega_s^T}\tilde{\xi}_m,e^{T_m\Omega_u^T}\tilde{\eta}_m)\|}$ 
converges to some element $(0,\tilde{\eta}_{\infty}^{(1)})$ since $\|\tilde{\eta}_m\|$ is uniformly bounded from below by a 
positive constant. By construction, the limit $(z_{\infty}^{(1)};\zeta_{\infty}^{(1)})$ reads $(\tilde{x}_{\infty}^{(1)},0;0,\tilde{\eta}_{\infty}^{(1)})$
where $\tilde{\eta}_{\infty}^{(1)}\neq 0, \tilde{x}_{\infty}^{(1)}\neq 0$ hence the limit
$(z_{\infty}^{(1)};\zeta_{\infty}^{(1)})$ belongs to $N^*\left(W^s(\Lambda_j)\setminus\Lambda_j\right)\cap S^*M$ which proves our claim in the case of fixed points. 
%Moreover, as $\zeta^{(1)}_{\infty}=(0,\tilde{\eta}_{\infty}^{(1)})$, we know that $(z^{(1)}_{\infty},\zeta^{(1)}_{\infty})$ belongs to 
%$N^*(W^s(\Lambda_j))\cap S^*M$. 
%Hence, by the transversality assumption~\eqref{e:transversality}, it does not belong to $\Sigma_{ss}$ 
%nor $\Sigma_s$ which is the last assumption we had to verify.  

\subsubsection{$\Lambda_j$ is a closed orbit} 
We proceed to the case of a closed orbit and 
we start with the case of $\Sigma_{ss}$. 
As before, we may apply the flow $\varphi^{t}$ 
in backward times in order to ensure that, 
for every $m$ large enough, $z_m$ belongs to the linearizing 
chart near the closed orbit $\Lambda_j$. 
As $(z_{\infty};\zeta_{\infty})$ does not belong to $\Sigma_{ss}$, 
we know from 
Lemma~\ref{l:limit-attractor} that it will converge 
to a point in $\Sigma_u$ under the action of the lifted flow $\tilde{\Phi}^t$. 
Hence, up to 
applying the lifted flow in backward times, we can suppose again that for every $m$ large enough,
$(z_m;\zeta_m)$ is within a distance $\delta$ 
from the component induced by $\Lambda_j$ of $\Sigma_u$ for some fixed $\delta>0$ small enough. We write these points in local coordinates near $\Lambda_j$:
$$(z_m;\zeta_m)=(P(\theta_m,0)(\tilde{x}_m,\tilde{y}_m),\theta_m;(P(\theta_m,0)^T)^{-1}(\tilde{\xi}_m,\tilde{\eta}_m),\Theta_m)$$
and 
$$(z_{\infty};\zeta_{\infty})=(P(\theta_{\infty},0)(0,\tilde{y}_{\infty}),\theta_{\infty};(P(\theta_{\infty},0)^T)^{-1}(\tilde{\xi}_{\infty},\tilde{\eta}_{\infty}),\Theta_{\infty}).$$
As $z_m$ belongs to $W^u(\Lambda_i)$ with $i\neq j$, 
we have one more time $\tilde{x}_m\neq 0$. 
From the expression\footnote{Observe that the we wrote the expression for $\Sigma_s$ and that the expression for $\Sigma_u$ is analogous.} of 
the $\Lambda_j$ component of $\Sigma_u$ near $\Lambda_j$ 
given in paragraph~\ref{ss:unstable-conormal-coord-2}, 
we also know that $\|(\tilde{\eta}_\infty,\Theta_\infty)\|\neq 0$, hence, for $m$ large enough, 
$\|(\tilde{\eta}_m,\Theta_m)\|$ is uniformly 
bounded from below by some positive constant. 
In these local coordinates, the flow reads
$$\varphi^{-t}(z_m)=\left(P(\theta_m-t,0)(e^{t\Omega_s}\tilde{x}_m,e^{-t\Omega_u}\tilde{y}_m),\theta_m-t\right),$$
while the cotangent component evolves as
\begin{equation}\label{e:cotangentcomponenttechnicallemma}
\frac{((P(\theta_m-t,0)^T)^{-1}(e^{-t\Omega_s^T}\tilde{\xi}_m,e^{t\Omega_u^T}\tilde{\eta}_m),\Theta_m+R(-t,\tilde{x}_m,\tilde{y}_m,\theta_m,\xi_m,\eta_m))}
{\|((P(\theta_m-t,0)^T)^{-1}(e^{-t\Omega_s^T}\tilde{\xi}_m,e^{t\Omega_u^T}\tilde{\eta}_m),\Theta_m+R(-t,\tilde{x}_m,\tilde{y}_m,\theta_m,\xi_m,\eta_m)\|},
\end{equation}
where $R(-t,\tilde{x}_m,\tilde{y}_m,\theta_m,\xi_m,\eta_m)$ was defined 
precisely in section~\ref{s:cotangent} -- see equations~\eqref{e:remainderR} and~\eqref{e:linearize-hamiltonian-flow-periodic-orbit-2}. 
As in the case of a critical 
point, we fix $\delta_1>0$ and, 
for every $m$ large enough, 
we may choose some 
time $T_m$ large enough 
to ensure that $\delta_1
\leq\|e^{T_m\Omega_s}\tilde{x}_m\|\leq2\delta_1$. 
Again, $T_m$ will be of order $|\log\|x_m\||$ thanks to 
property~\eqref{e:lyapunov-bound-closed-orbit} and we 
define our new sequence as 
$(z_m^{(1)};\zeta_m^{(1)})=\tilde{\Phi}^{-T_m}(z_m;\zeta_m)$. 
As above, we can extract some
subsequence which converges 
to a limit 
point $(z_{\infty}^{(1)};\zeta_{\infty}^{(1)})$. 
Let us verify that this new sequence satisfies the claim of the Lemma. 

Again, property 
(1) is directly verified from our construction. 
Moreover, as $\|\tilde{x}_m\|$ tends to $0$, we can verify that 
$z_{\infty}^{(1)}=(P(\theta_{\infty}^{(1)},0)(\tilde{x}_{\infty}^{(1)},0),\theta_{\infty}^{(1)})$ 
for some $\tilde{x}_{\infty}^{(1)}\neq 0$. Hence, 
$z_{\infty}^{(1)}$ belongs to $W^s(\Lambda_j)\setminus\Lambda_j$ which is property (3) we are looking for. 
It remains to show that the limit point $(z_{\infty}^{(1)};\zeta_{\infty}^{(1)})$
belongs to $N^*(W^s(\Lambda_j))\cap S^*M$. For that purpose, we may distinguish two cases: $\|e^{T_m\Omega_u^T}\tilde{\eta}_m\|$ tends to $+\infty$, or 
$\|e^{T_m\Omega_u^T}\tilde{\eta}_m\|$ remains bounded.

 In the first case, using~\eqref{e:linearize-hamiltonian-flow-periodic-orbit-2}, one can find some constant $C>0$ 
 depending only on the flow near the closed orbit such that
$$C_m:= \|((P(\theta_m-T_m,0)^T)^{-1}(e^{-T_m\Omega_s^T}\tilde{\xi}_m,e^{T_m\Omega_u^T}\tilde{\eta}_m),\Theta_m+R(-T_m,
\tilde{x}_m,\tilde{y}_m,\theta_m,\xi_m,\eta_m)\|$$
$$\hspace{5cm}\geq 
C(1-C\delta_1)\|e^{T_m\Omega_u^T}\tilde{\eta}_m\|\geq \frac{C}{2}\|e^{T_m\Omega_u^T}\tilde{\eta}_m\|,$$
if $\delta_1>0$ is small enough. Similarly, one gets an upper bound on $C_m$ which is of order $\|e^{T_m\Omega_u^T}\tilde{\eta}_m\|$. In particular, 
up to an extraction, $e^{T_m\Omega_u^T}\tilde{\eta}_m/C_m\rightarrow\tilde{\eta}_{\infty}^{(1)}\neq 0$. 
Hence, one has, using one more time~\eqref{e:linearize-hamiltonian-flow-periodic-orbit-2}, that the limit covector will be of the form
$$\left((P(\theta_{\infty}^{(1)},0)^T)^{-1}(0,\tilde{\eta}_{\infty}^{(1)}),-\la \partial_\theta P(\theta_{\infty}^{(1)},0)(\tilde{x}_{\infty}^{(1)},0),
(P(\theta,0)^{T})^{-1}(0,\tilde{\eta}_{\infty}^{(1)})\ra\right).$$
Hence, $(z_{\infty}^{(1)},\zeta_{\infty}^{(1)})$ belongs to $N^*(W^s(\Lambda_j)).$ In particular, it belongs to 
$\cup_{x\in\Lambda_j}N^*(W^{ss}(x))\cap S^*M$.
%the $\Theta$ component of the covector $\zeta_{\infty}^{(1)}$ 
%is $\ml{O}(\delta_1)$ while the $\tilde{\xi}$ component vanishes and the $\tilde{\eta}$ component does not vanish. In particular, 
%$(z_{\infty}^{(1)},\zeta_{\infty}^{(1)})$ belongs to $\cup_{x\in\Lambda_j}N^*(W^{ss}(x))\cap S^*M$. Hence, by the transversality 
%assumption~\eqref{e:transversality}, it does not belong to $\Sigma_s$ which is the content of property (2). Moreover, as the $\tilde{\eta}$ component 
%does not vanish, one has still by the transversality assumption that the limit point does not belong to $\Sigma_{ss}$.

Suppose now that $\|e^{T_m\Omega_u^T}\tilde{\eta}_m\|$ remains bounded. In particular, this implies that $\tilde{\eta}_{\infty}=0$ and that 
$\Theta_{\infty}\neq -\la \partial_\theta P(\theta_{\infty},0)(0,\tilde{y}_{\infty}),
(P(\theta,0)^{T})^{-1}(\tilde{\xi}_{\infty},0)\ra$. Observe also that $e^{-T_m\Omega_s^T}\tilde{\xi}_m$ still goes to 
$0$ as $m\rightarrow+\infty$. Hence, in that case, one can verify, using~\eqref{e:linearize-hamiltonian-flow-periodic-orbit-2} that 
the limit covector $\zeta_{\infty}^{(1)}$ 
is of the form $((P(\theta_{\infty}^{(1)},0)^T)^{-1}(0,\tilde{\eta}_{\infty}^{(1)}),\Theta_{\infty}^{(1)})$ (note that 
$\tilde{\eta}_{\infty}^{(1)}=0$ iff $\|e^{T_m\Omega_u^T}\tilde{\eta}_m\|\rightarrow 0^+$), i.e. $(z_{\infty}^{(1)},\zeta_{\infty}^{(1)})$ belongs to $\cup_{x\in\Lambda_j}N^*(W^{ss}(x))\cap S^*M$.

In the case of $\Sigma_s$, the situation is 
slightly simpler as the fact that the limit point does not belong to $\Sigma_s$ implies that the component $\|\tilde{\eta}_m\|$ has to be uniformly bounded from 
below by a positive constant. In particular, $\|e^{T_m\Omega_u^T}\tilde{\eta}_m\|$ tends to $+\infty$, and the $\tilde{\eta}$ component of the limit point does 
not vanish. We already discussed 
that case above and we saw that the new limit point does not belong to $N^*(W^u(\Lambda_j))$.

\subsection{Conclusion of the proof of Theorem \ref{t:compactness}}
\label{ss:algo}
Using the technical Lemma, we can give a proof of Theorem \ref{t:compactness} by
contradiction. Suppose that there exists a sequence $(z_m;\zeta_m)$ in $\Sigma_{ss}$ (resp. $\Sigma_s$) which converges to a point $(z_\infty;\zeta_\infty)$ 
that does not belong to $\Sigma_{ss}$ (resp. $\Sigma_s$). Without loss of generality, we may assume that all elements
$(z_m;\zeta_m)$ belong to $N^*(W^u(\Lambda))$ (resp. $\cup_{x\in\Lambda}N^*(W^{uu}(\Lambda))$) for some $\Lambda\in NW(\varphi^t)$. We have two situations for the sequence 
$(z_m;\zeta_m)$ and its limit $(z_\infty;\zeta_\infty)$:
\begin{itemize}
 \item either $z_\infty\in W^u(\Lambda)$ and we get the contradiction from the first part of Lemma~\ref{l:keytechnicallemma};
 \item or $z_\infty\in W^u(\Lambda_{i_1})$ for some $\Lambda_{i_1}\neq\Lambda$ and Lemma~\ref{l:keytechnicallemma}
gives us a new sequence $(z^{(1)}_m;\zeta^{(1)}_m)$ in $\Sigma_{ss}$ (resp. $\Sigma_s$)
which converges to $(z_\infty^{(1)};\zeta_{\infty}^{(1)})\in \cup_{x\in\Lambda_{i_1}}W^{ss}(x) $ (resp. $N^*(W^s(\Lambda_{i_1}))$)
with $z^{(1)}_m\in W^u(\Lambda)$ for all $m$. Moreover, $(z_\infty^{(1)};\zeta_{\infty}^{(1)})\notin \Sigma_{ss}$ (resp. $\Sigma_s$) with
$z_\infty^{(1)}\in W^s(\Lambda_{i_1})\setminus \Lambda_{i_1}$.
\end{itemize}
In the second case, $z_\infty^{(1)}\in W^u(\Lambda_{i_2})$ with $i_2\neq i_1$ and we can repeat the same argument. In fact, we can reproduce 
this procedure as long as
$\Lambda_{i_p}\neq\Lambda$ and we obtain some sequences $(z^{(1)}_m;\zeta^{(1)}_m),(z^{(2)}_m;\zeta^{(2)}_m),\dots$, 
together with their respective limits $(z^{(1)}_\infty;\zeta^{(1)}_\infty)\notin \Sigma_{ss},(z^{(2)}_\infty;\zeta^{(2)}_\infty)\notin \Sigma_{ss},\dots$ (resp. $\notin\Sigma_{s}$) and
some sequence of critical elements $(\Lambda_{i_1},\Lambda_{i_2},\ldots)$ which are \textbf{two by two distinct}
such that
$z_\infty^{(1)}\in W^s(\Lambda_{i_1})\cap W^u(\Lambda_{i_2}), z_\infty^{(2)}\in W^s(\Lambda_{i_2})\cap W^u(\Lambda_{i_3}), \dots $. Indeed, 
if we had $\Lambda_{i_p} = \Lambda_{i_q}  $ for $q>p$, then, by Lemma \ref{l:partialorder}, we would 
necessarily have  the full set of equalities $\Lambda_{i_p}=\dots=\Lambda_{i_q}$. But we would also have 
that $z_\infty^{(k)}\in W^u(\Lambda_{i_k+1}=\Lambda_{i_k})\cap W^s(\Lambda_{i_k})$ 
and $z_\infty^{(k)}\notin \Lambda_{i_k}$ by construction of the sequence $(z_m^{(k)})_m$. 
This would contradict the property that
the intersection $W^u(\Lambda_i)\cap W^s(\Lambda_i)$ is reduced to $\Lambda_i$ 
from Lemma \ref{l:nocycle}. 

Since the number of critical elements is finite and the 
critical elements $\Lambda_{i_1},\Lambda_{i_2},\dots$ produced by our procedure are two by two distinct, 
this algorithm must terminate at some $\Lambda_{i_p}=\Lambda$. 
This leads us to situation 1 for the sequence
$(z^{(p)}_m;\zeta^{(p)}_m)$ and its limit $(z^{(p)}_\infty;\zeta^{(p)}_\infty)$ and we
would get the result that infinitely many terms in the sequence $(z_m^{(p)};\zeta_m^{(p)})$ do not belong to $\Sigma_{ss}$ (resp. $\Sigma_s$) contradicting the initial assumption.  
This concludes the proof of compactness.

\subsection{Proof of Lemma \ref{l:boundary2} about the closure of unstable manifolds}

Note that the proof we just gave was independent of Smale's Theorem~\ref{t:smale} as it only used Lemma~\ref{l:partialorder} from the Appendix. In fact, 
Lemma~\ref{l:keytechnicallemma} can also be used to recover a result due to Smale~\cite[Lemmas 3.6 and 3.7]{Sm60} (see also~\cite{Web06}): 
\begin{lemm}\label{l:boundary2} 
Suppose that $W^u(\Lambda_j)\cap \overline{W^u(\Lambda_i)}\neq\emptyset$. Then, there exists a sequence $j=i_1,\ldots ,i_q=i$ such that 
$W^s(\Lambda_{i_p})\cap W^u(\Lambda_{i_{p+1}})\neq\emptyset$ for every $1\leq p\leq q-1$. In particular, from Lemmas~\ref{l:boundary} and~\ref{l:partialorder}, $W^u(\Lambda_j)\subset\overline{W^u(\Lambda_i)}.$
\end{lemm}
This Lemma is part of Smale's proof of Theorem~\ref{t:smale} and it could in fact be derived without the $\ml{C}^1$-linearization property. 
We briefly recall how this result could be deduced from Lemma~\ref{l:keytechnicallemma}.
\begin{proof} The proof follows a similar algorithm as in paragraph \ref{ss:algo} but working only on the base $M$. 
Consider some sequence $(z_m)_{m\in \mathbb{N}}$ in $W^u(\Lambda_i)$ such that $z_m\underset{m\rightarrow \infty}{\longrightarrow} z_\infty\in W^u(\Lambda_j)$. 
Then we have two situations~: 
\begin{enumerate}
\item either $\Lambda_i=\Lambda_j$ and we are done with $i_1=i=j$.
\item or $\Lambda_i\neq\Lambda_j$ and flowing backwards by the flow, we can find, up to extraction,
some subsequence $(z^{(1)}_m)_m$ in $W^u(\Lambda_{i})$
which converges to $z^{(1)}_\infty\in W^s(\Lambda_j)\setminus\Lambda_j$. Hence, there exists $i_2$ such that $z^{(1)}_\infty\in W^u(\Lambda_{i_2})$.
\end{enumerate} 
Then, either $i_2=i$ and we are done or $i_2\neq i$. In the latter case, we can apply one more time Lemma~\ref{l:keytechnicallemma} to produce a new sequence. 
Iterating this argument, we will be given a sequence of elementary critical elements $(\Lambda_{i_1},\Lambda_{i_2},\ldots)$ which are two by two distinct by the arguments 
of paragraph~\ref{ss:algo}. Hence, the procedure will end at some step $i_q$ where $i_q=i$. 
\end{proof}

\subsection{Further comments}

Actually, reproducing the same argument as for the proof of Theorem~\ref{t:compactness} 
allows to prove something slightly stronger :  
 \begin{theo}  
 Fix $J\subset\{1,\ldots, K\}$. Then,
 $$\Sigma_{ss}^J:=\bigcup_{j\in J}\bigcup_{W^u(\Lambda_i)\preceq W^u(\Lambda_j)} N^*(W^u(\Lambda_i))\cap S^*M,$$
  and
  $$\Sigma_{s}^J:=\bigcup_{j\in J}\bigcup_{W^u(\Lambda_i)\preceq W^u(\Lambda_j)} \bigcup_{x\in\Lambda_j}N^*(W^{uu}(x))\cap S^*M$$
  are compact subsets of $S^*M$. By considering negative times of the flow, the same of course holds for the stable manifolds with the 
  associated partial order relation.
 \end{theo}

\section{Proof of the stable neighborhood Theorem~\ref{t:attractor}.}\label{s:proof-attractor}

We now turn to the proof of Theorem~\ref{t:attractor} in the case of $\Sigma_{ss}$. Here, it will somehow be more convenient to work with conical 
neighborhoods rather than neighborhoods in the unit cotangent bundle. More precisely we define
$$\mathbf{\Sigma}_{ss}:=\bigcup_{j=1}^KN^*(W^u(\Lambda_j)).$$
Let us state a precised version of Theorem \ref{t:attractor}.
\begin{theo}[Conical stable neighborhood]
\label{t:preciseattractor}
For every $\eps>0$, we will construct an \emph{open conical neighborhood} $\mathbf{V}_{ss}$ of $\mathbf{\Sigma}_{ss}$ 
in $T^*M\backslash 0$ such that
\begin{enumerate}
 \item $\forall t\geq 0$, $\Phi^t(\mathbf{V}_{ss})\subset \mathbf{V}_{ss}$,
 \item for every $(z;\zeta)$ in $\mathbf{V}_{ss}$, $(z;\zeta/\|\zeta\|_z)$ is within a distance $\eps$ 
 from some element in $\mathbf{\Sigma}_{ss}$.
\end{enumerate}
\end{theo}
We will focus on the case of $\mathbf{\Sigma}_{ss}$ and we will explain at each step how the proof has to be adapted for 
$\mathbf{\Sigma}_s$. Once this 
conical neighborhood is constructed, one can conclude the proof of Theorem~\ref{t:attractor} by relating the flow $\tilde{\Phi}^t$ 
to $\Phi^t$.

 We note that we used so far an auxiliary metric $g$ to define the distance. 
In the upcoming proofs, 
near every elementary critical element $\Lambda_j$, 
we shall use a norm denoted by $\mathbf{N}_j$
to define our 
neighborhood in every linearizing 
chart near the fixed $\Lambda_j$. 
This norm has a priori nothing to do with the norm induced by $g$ on 
the local chart. 
For the sake of simplicity, 
we shall start with the Euclidean metric $\|.\|$ 
in the chart and then show how to 
adapt it to the dynamics.
In order to construct the neighborhood, we will also introduce
three small parameters~:
\begin{enumerate}
\item $\epsilon_i>0$ which controls the distance of base points
to $W^u(\Lambda_i)$, 
\item $\epsilon^\prime_i>0$ which controls the distance of base points to $W^s(\Lambda_i)$, 
\item $\epsilon^{\prime\prime}_i>0$ which controls the aperture of some cone in the cotangent fiber.
\end{enumerate}
This triple of parameters will be adjusted in the inductive construction of the stable neighborhood. 
In order to clarify the upcoming statements, $\eps_i, \eps_i^{\prime},\eps_i^{\prime\prime}$ will be adjusted in terms of $\eps$ and of 
the $\eps_j, \eps_j^{\prime},\eps_j^{\prime\prime}$ with $W^u(\Lambda_j)\preceq W^u(\Lambda_i)$. Moreover, $\eps_i^{\prime\prime}$ will 
be adjusted in terms of $\eps_i^{\prime}$ and $\eps_i$ in terms of $\eps_i^{\prime\prime}$.

We now fix $\eps>0$ (small enough) and proceed to the construction of $\mathbf{V}_{ss}$ by induction on Smale's partial order relation.

\subsection{Construction near minimal elements of Smale's partial ordering}
\label{ss:initialisation}
 We start by setting
$$J_0:=\left\{1\leq j\leq K: W^u(\Lambda_j)\succeq W^u(\Lambda_{j'})\Longrightarrow j=j'\right\}.$$
These are the minimal elements for Smale's partial order relation. Recall that such points are attracting for the flow. Fix $j\in J_0$ and some small parameter 
$\eps_j>0$ that we will adjust with respect to the value of $\eps$. Suppose first that $\Lambda_j$ is a fixed 
point and let us explain how to construct $\mathbf{V}_{ss}$ near this point. Without loss of generality, we can assume 
that we are in a neighborhood of
$\Lambda_j$ where we can use the linearizing chart of paragraph~\ref{ss:coordinates}. Following 
classical ideas for hyperbolic dynamical systems~\cite[Prop.~5.2.2]{BrSt02}, we introduce the conical 
neighborhood:
\begin{equation}
V_j:= T^*B_j\setminus \{0\}  \text{ for } B_j=\left\{\tilde{x}:\int_0^{+\infty}\|e^{-t\Omega_s}\tilde{x}\|dt<\eps_j\ \right\},
\end{equation}
 where the integral is convergent thanks to~\eqref{e:lyapunov-bound-fixed-point}. Recall that, for any fixed point which is minimal, 
 $N^*(W^u(\Lambda_j))=T_{\Lambda_j}^*M$.
 By construction, this set is 
 invariant under the forward flow $\Phi^t$. Obviously, if $(\tilde{x};\tilde{\xi})\in\mathbf{\Sigma}_{ss}$ is such that
 $\int_0^{+\infty}\|e^{-t\Omega_s}\tilde{x}\|dt<\epsilon_j$, then $(\tilde{x};\tilde{\xi})$ belongs to $V_j$. Moreover, recalling the expression of 
$N^*(W^u(\Lambda_j))$ given in 
paragraph~\ref{ss:unstable-conormal-coord}, 
by choosing $\eps_j>0$ small enough, and
every point in $V_j$ is $\eps$-close to $\mathbf{\Sigma}_{ss}$ 
in the sense of the second claim of Theorem \ref{t:preciseattractor}.

Let us discuss the case where $\Lambda_j$ is a closed orbit. In that case, we fix two small parameters $\eps_j>0$ and $\eps_j''>0$ 
that will both depend on $\eps$. Moreover,
 we will fix $\eps_j$ in terms of $\eps_j''$. First of all, we introduce a new norm on $\IR^{n-1}$:
$$\mathbf{N}_j\left(\tilde{\xi}\right):=\int_0^{+\infty}\|e^{-t\Omega_s^T}\tilde{\xi}\|e^{t\lambda_j}dt,$$
which is well--defined for $\lambda_j>0$ 
small enough\footnote{We just choose $\lambda_j<\chi_+$.} thanks to the inequality~\eqref{e:lyapunov-bound-closed-orbit}. With this norm, one 
has $\mathbf{N}_j\left(e^{-t_0\Omega_s^T}\tilde{\xi} \right)\leq e^{-t_0\lambda_j}\mathbf{N}_j\left(\tilde{\xi}\right)$ for every $t_0\geq 0$ and every $\tilde{\xi}\in\IR^{n-1}$. Then, 
using the notational conventions of paragraph~\ref{aaa:floquet}, we set 
 $$B_{j}:=\left\{(P(\theta,0)\tilde{x},\theta) 
 :\int_0^{+\infty}\|e^{-t\Omega_s}\tilde{x}\|dt<\eps_j\right\},$$
 and
 $$V_{j}:=\left\{(P(\theta,0)\tilde{x},\theta,(P(\theta,0)^{T})^{-1}\tilde{\xi}),\Theta)\in T^*B_j\backslash 0 
 : (*)\ \text{holds}\right\},$$
where $(*)$ means that
\begin{equation}\label{e:starcondition}
\eps_j''\mathbf{N}_j\left(\tilde{\xi}\right)>\left|\Theta+\langle\partial_{\theta}P(\theta,0)\tilde{x},(P(\theta,0)^T)^{-1}\tilde{\xi}\rangle\right|,
\end{equation}
%Again, the integrals converge thanks to the bound~\eqref{e:lyapunov-bound-closed-orbit} and to
%the fact that $R$ has polynomial growth in $t$ 
%by Lemma~\ref{r:remainder}. 
which is a conical set in $T^*M\backslash 0$. This condition on the cotangent component follows from the exact expression given in~\eqref{e:conormunstable-orbit}.

\begin{rema}
 In the case of $\mathbf{\Sigma}_s$, the situation is slightly simpler at this step as we just need to impose $(\tilde{\xi},\Theta)\neq 0$.
\end{rema}

Let us verify that $V_j$ is invariant under the flow 
in positive time. We fix a point in $V_j$. For the variable on $M$, this follows from the definition. 
For the cotangent component, we write using~\eqref{e:linearize-hamiltonian-flow-periodic-orbit-2} for $t_0\geq 0$~:
$$\left|\Theta+R(t_0,\tilde{x},\theta,\tilde{\xi})+\langle\partial_{\theta}P(\theta+t_0,0)e^{-t_0\Omega_s}\tilde{x},(P(\theta+t_0,0)^T)^{-1}e^{t_0\Omega_s^T}\tilde{\xi}\rangle\right|$$
$$=\left|\Theta+\langle\partial_{\theta}P(\theta,0)\tilde{x},(P(\theta,0)^T)^{-1}\tilde{\xi}\rangle\right|< \eps_j''\mathbf{N}_j\left(\tilde{\xi}\right)\leq  \eps_j''\mathbf{N}_j\left(e^{t_0\Omega_s^T}\tilde{\xi}\right)$$
%\begin{eqnarray*}\left|\Theta+\langle\partial_{\theta}P(\theta+t_0,0)e^{-t_0\Omega_s}\tilde{x},(P(\theta+t_0,0)^T)^{-1}e^{t_0\Omega_s^T}\tilde{\xi}\rangle\right| &\leq &
%e^{t_0\lambda_j}\int_0^{+\infty}|\Theta+R(t,\tilde{x},\theta,\tilde{\xi})|e^{-t\lambda_j}dt\\
%  < e^{t_0\lambda_j}\eps_j''\mathbf{N}_j\left(\tilde{\xi}\right)&\leq & \eps_j''\mathbf{N}_j\left(e^{t_0\Omega_s^T}\tilde{\xi}\right).
% \end{eqnarray*}
Hence, the set we have just defined is invariant under the action of the flow $\Phi^{t_0}, t_0\geqslant 0$. 
Let us now verify that it is $\eps$-close to $\mathbf{\Sigma}_{ss}$ 
which is the second claim of Theorem~\ref{t:preciseattractor}. 
For $\tilde{x}=0$, this is immediate as the condition on $\tilde{\xi}$ 
reads $\eps_j''\mathbf{N}_j\left(\tilde{\xi}\right)>|\Theta|$. 
Hence, if we choose $\eps_j,\eps_j''$ small enough (with respect to $\eps$), we are done. 
In fact, as our conditions are continuous with respect to the different variables, 
we can verify that this remains true for 
$\tilde{x}$ small enough. 

\subsection{Adjusting the neighborhood property to $\Sigma_{ss}$}
\label{ss:neighborhoodproperty}
Finally, we would like to show that, for small enough $\eps_j$ and $\eps_j''$, any element
$(P(\theta,0)\tilde{x},\theta;(P(\theta,0)^{T})^{-1}\tilde{\xi}),\Theta)\in\mathbf{\Sigma}_{ss}$ 
such that 
$(P(\theta,0)\tilde{x},\theta)$ belongs to $B_j$ is in fact inside $V_j$.
This is not a priori obvious for the following reason. Since
$\Sigma_{ss}=\bigcup_{\Lambda\in NW(\varphi^t)} N^*\left(W^u(\Lambda)\right)\cap S^*M$, there are two kinds of elements 
in $\Sigma_{ss}\cap S^*B_j$~ near the periodic orbit $\Lambda_j$:
\begin{itemize}
\item the elements in $N^*\left(W^u(\Lambda_j)\right)\cap S^*M$ contained in $V_j$ by construction,
\item the points coming from the conormals $N^*\left(W^u(\Lambda)\right)\cap S^*M$
for critical elements $\Lambda$ such that $\overline{W^u(\Lambda)} \cap W^u(\Lambda_j)\neq \emptyset$.
\end{itemize}
Up to this point, what is obvious is the fact that $V_j$ forms a neighborhood of $N^*(W^u(\Lambda_j))$
by construction and we would like to show $V_j$ also contains the other points.
For that purpose, we fix $\eps_j''>0$ and we argue by contradiction.
Precisely, we suppose that for every $\eps_j>0$, we can find some element in $S^*B_j\cap\Sigma_{ss}$ not belonging to $V_j$.
Consider some sequence $(\epsilon_{j,p})_{p\in \mathbb{N}}$ such that $\lim_{p\rightarrow \infty}\epsilon_{j,p}= 0$, set $B_{j,p}=\{(P(\theta,0)\tilde{x},\theta) 
 :\int_0^{+\infty}\|e^{-t\Omega_s}\tilde{x}\|dt<\eps_{j,p}  \}$ and $V_{j,p}=\{(P(\theta,0)\tilde{x},\theta;(P(\theta,0)^{T})^{-1}\tilde{\xi}),\Theta)\in T^*B_{j,p}\backslash 0  : (*)\ \text{holds} \}$ where $(*)$ is given by equation 
 (\ref{e:starcondition}) for the \textbf{fixed} parameter $\epsilon_j^{\prime\prime}$. 
We thus obtain some sequence 
$(z_p;\zeta_p)\in S^*B_{j,p}\cap\mathbf{\Sigma}_{ss}$ such that $\text{dist}(z_{p},\Lambda_j)\rightarrow 0$ and $\zeta_p$ fails to satisfy the
inequality \eqref{e:starcondition} for all $p$~:
\begin{equation}\label{e:contradictionstar}
\eps_j''\mathbf{N}_j\left(\tilde{\xi}_p\right)\leqslant \left|\Theta_p+\langle\partial_{\theta}P(\theta_p,0)\tilde{x}_p,(P(\theta_p,0)^T)^{-1}\tilde{\xi}_p\rangle\right| .
\end{equation}
By Theorem \ref{t:compactness} the set $S^*B_j\cap\Sigma_{ss}$ is \textbf{compact}.
Hence, up to extraction, we can find a subsequence $(z_p;\eta_p)$ which converges to $(z;\zeta)\in S^*B_j\cap\Sigma_{ss}$ with $z\in \Lambda_j$.
But, since $\Sigma_{ss}$ is partitioned as $\bigcup_{\Lambda\in NW(\varphi^t)} N^*\left(W^u(\Lambda)\right)\cap S^*M$,
the element $(z;\zeta)$ actually belongs to $N^*\left(W^u(\Lambda_j)\right)\cap S^*M$.
Therefore, by the coordinate representation of $N^*\left(W^u(\Lambda_j)\right)$ from paragraph~\ref{ss:unstable-conormal-coord},
$(z;\zeta)$ is of the form $\left(0,\theta;\left(P(\theta,0)^T\right)^{-1} \tilde{\xi},0\right),$
and $(z_p;\zeta_p)=\left(\tilde{x}_p,\theta_p;\left(P(\theta_p,0)^T\right)^{-1} \tilde{\xi}_p,\Theta_p\right)$ 
where $\tilde{x}_p\rightarrow 0,\Theta_p\rightarrow 0, \tilde{\xi}_p\rightarrow \tilde{\xi}\neq 0$. But this means that, for $p\rightarrow+\infty$,
 $$ \eps_j''\mathbf{N}_j\left(\tilde{\xi}_p\right)\rightarrow\eps_j''\mathbf{N}_j\left(\tilde{\xi}\right)\neq 0,$$
 and
 $$\left|\Theta_p+\langle\partial_{\theta}P(\theta_p,0)\tilde{x}_p,(P(\theta_p,0)^T)^{-1}\tilde{\xi}_p\rangle\right|\rightarrow 0.$$
%where the secon term tends to $0$ thanks to inequality
%(\ref{e:small-remainder}) and the fact that $\tilde{x}_p\rightarrow 0$.  
This contradicts the inequality (\ref{e:contradictionstar}). 
% 
%Then, up to extraction, this point would converge to a point satisfying 
%$\tilde{x}=0$ and $\Theta\neq 0$. Hence, this limit point would 
%not belong to $N^*(W^u(\Lambda_j))$ -- see Remark~\ref{r:unstable-conormal-coord}. In particular, 
%it does not belong to $\Sigma_{ss}$ and this contradict the fact that $\Sigma_{ss}$ is compact. In other words, 
%we have shown that $V_j$ is an open conical neighborhood of $T^*V_j^{(1)}\cap\mathbf{\Sigma}_{ss}$.

We note that the importance of being able to adjust the parameter $\epsilon_j$ will play an important role at each step of our construction. This concludes 
the construction of the neighborhood near minimal elements for Smale's partial order relation.

\subsection{Induction on Smale's partial ordering: the case of fixed points} 
 
We now proceed to the second step of the induction and set
$$J_1:=\left\{ j\notin J_0: W^u(\Lambda_j)\succeq W^u(\Lambda_{j'})\Longrightarrow j=j'\ \text{or}\ j'\in J_0\right\}.$$
We will now construct an adapted neighborhood near every $\Lambda_j$ such that $j\in J_1$. Again, we start with 
the case of a fixed point $\Lambda_j$ and we fix several small parameters $\eps_j,\eps_j',\eps_j''>0$ that will be determined. 

%\begin{rema}
% In order to facilitate the reading, let us explain what these parameters will constrain and how they depend on each other. First of all, 
% $\eps_j'>0$ will represent the size of the neighborhood inside $M$ along the $\tilde{y}$-variable (meaning along $W^u(\Lambda_j$). It will be 
% chosen small enough to work inside the local chart and it is fixed once and for all. The parameter $\eps_j''>0$ represents the aperture of 
% the cone we will define in the cotangent variable $(\tilde{\xi},\tilde{\eta})$: it will be adjusted in a manner that depends on $\eps_j'$ and on the neighborhoods that have been 
% constructed for $j'\in J_0$. Finally, $\eps_j>0$ represents the size along the $\tilde{x}$-variable inside $M$ (meaning along $W^s(\Lambda_j)$) 
% and it will be adjusted in terms of $\eps_j'$, $\eps_j''$ and the neighborhoods previously constructed. 
%\end{rema}

Recall that in the linearizing chart, 
$N^*(W^u(\Lambda_j))$ can be written as
$$\left\{(0,\tilde{y},\tilde{\xi},0):\ \tilde{y}\in\IR^{n_u},\ \tilde{\xi}\in\IR^{n_s}\setminus\{0\}\right\}.$$
First of all, we define a small neighborhood of $\Lambda_j$ inside $M$ as follows (see figure~\ref{f:neighborhood}):
$$B_j:=\left\{(\tilde{x},\tilde{y}):\int_0^{+\infty}\|e^{-t\Omega_s}\tilde{x}\|dt<\eps_j\ \text{and}\ \|\tilde{y}\|<\eps_j'\right\}.$$
Recall that the integrals converge thanks to~\eqref{e:lyapunov-bound-fixed-point}. 
\begin{figure}[ht]\label{f:neighborhood}
\includegraphics[width=14cm, height=8cm]{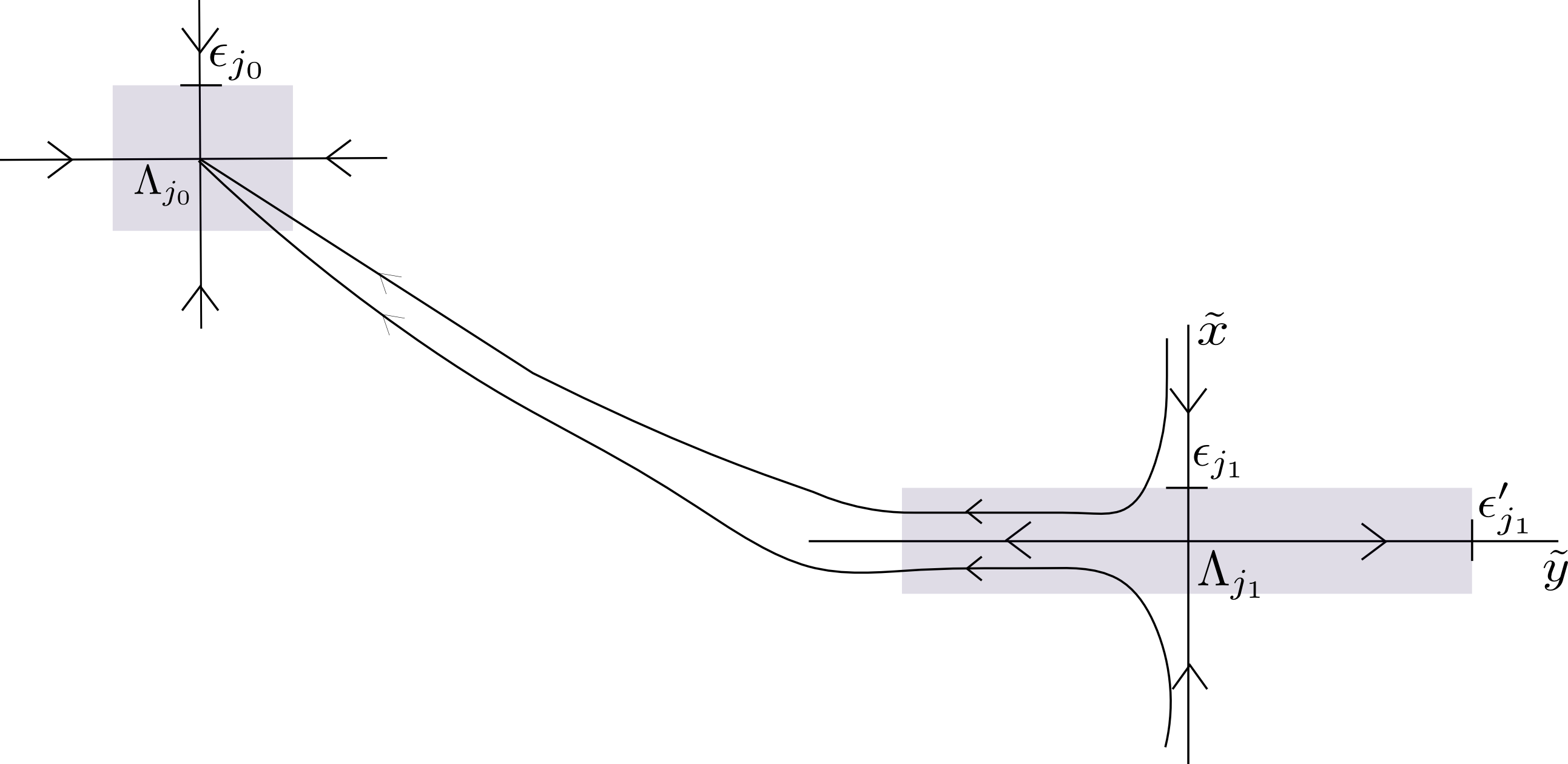}
\centering
\caption{Construction of the neighborhood inside $M$}
\end{figure}

\subsubsection{Reaching previous boxes in finite time.}
Let us first verify that there exists a uniform $T_j>0$ such that, for every point of the form 
$(0,\tilde{y})$ with $\|\tilde{y}\|=\eps_j'$ which 
are elements on the boundary of the box $B_j$, $\varphi^{T_j}(0,\tilde{y})$ belongs to one of the  previously 
constructed box $B_i$ with $i\in J_0$. Indeed, suppose by contradiction that it is not true. 
Then, one can construct a point $(0,\tilde{y}')$ with $\|\tilde{y}'\|=\eps_j'$ such that $\varphi^{t}(0,\tilde{y})$ reaches these neighborhoods in infinite time and this would 
contradict the fact 
that $\varphi^{t}(0,\tilde{y}')$ must converge 
to some $\Lambda_{i}$ with $i\in J_0$. 
By uniform continuity of the flow and up to decreasing the value of $\eps_j$ a 
little bit (in a way that depends on $\eps_j'$), we can thus assume that every point $(\tilde{x},\tilde{y})\in \overline{B_j}$ 
such that $\|\tilde{y}\|=\eps_j'$ will reach one of 
the neighborhoods $B_i$ (with $i\in J_0$) in an uniform time $T_j$.

\subsubsection{Cotangent vectors} 
 
Now we take care of the cotangent part of the elements in $B_j$. 
Precisely, as before, we first define a new norm on $\IR^{n_s}$: 
$$\mathbf{N}_j\left(\tilde{\xi}\right):=\int_0^{+\infty}\|e^{-t\Omega_s^T}\tilde{\xi}\|e^{t\lambda_j}dt,$$
which converges for small enough $\lambda_j>0$ thanks to~\eqref{e:lyapunov-bound-fixed-point}. Then, one has 
$\mathbf{N}_j\left(e^{-t_0\Omega_s^T}\tilde{\xi}\right)\leq e^{-t_0\lambda_j}\mathbf{N}_j\left(\tilde{\xi}\right)$ 
for every $t_0\geq 0$ and every $\tilde{\xi}\in\IR^{n_s}$. We now set
 $$V_{j}:=\left\{(\tilde{x},\tilde{y};\tilde{\xi},\tilde{\eta})\in T^*B_j\backslash 0 
 : (*)\ \text{holds}\right\},$$
where $(*)$ means that
$$\eps_j''\mathbf{N}_j\left(\tilde{\xi}\right)>\int_0^{+\infty}\|e^{-t\Omega_u^T}\tilde{\eta}\|dt.$$
Again this integral converges thanks to~\eqref{e:lyapunov-bound-fixed-point}. 
This defines clearly an open conical 
subset of $T^*B_j$, and we have, for every $t_0\geq 0$,
\begin{eqnarray*}
\int_0^{+\infty}\|e^{-t\Omega_u^T}e^{-t_0\Omega_u^T}\tilde{\eta}\|dt\leq \int_0^{+\infty}\|e^{-t\Omega_u^T}\tilde{\eta}\|dt
<\eps_j'' \mathbf{N}_j\left(  \tilde{\xi}\right)\leq e^{-t_0\lambda_j}\eps_j''\mathbf{N}_j\left(e^{t_0\Omega_s^T}\tilde{\xi}\right)\leq\eps_j''\mathbf{N}_j\left(e^{t_0\Omega_s^T}\tilde{\xi}\right).
\end{eqnarray*} 
In other words, it means that the 
conical condition on the cotangent vectors is 
preserved under the action of the forward flow $\Phi^{t_0}$ on $T^*M$. All the norms on 
$\IR^{n_s}$ being equivalent, 
we can also verify that every point in $V_j$ is $\eps$-close to $N^*(W^u(\Lambda_j))$ in the sense of
condition (2) of Theorem \ref{t:preciseattractor} (at least if the $\eps_j^{*}$ 
are chosen small enough). 
Up to decreasing the value of $\eps_j$ and $\eps_j''$ (in a way that depends only on $\eps_j'$ and $\eps$), 
we may also suppose that for every $(z;\zeta)\in V_j$ and for 
every $0\leq t\leq T_j$, the point $\Phi^t(z;\zeta)$ 
remains $\eps$-close to $N^*(W^u(\Lambda_j))$ in the sense of (2).

We now introduce the subsets
$$\mathcal{B}_j:=\bigcup_{t\geq 0}\varphi^t\left(\cup_{i\in J_0}B_{i}\cup B_j\right)\subset M,$$
and 
$$\mathcal{V}_j:=\bigcup_{t\geq 0}\Phi^t\left(\cup_{i\in J_0}V_{i}\cup V_j\right)\subset T^*M\backslash 0.$$ 
These are invariant subsets under the forward flow by definition. 
From our construction, all the points inside $\mathcal{V}_j$ are also $\eps$-close to $\mathbf{\Sigma}_{ss}$. 
\subsubsection{Adjusting the constants to $\Sigma_{ss}$}
\label{r:crucial}
 Let us come back to the important observation made in paragraph~\ref{ss:neighborhoodproperty}. Recall that $\mathbf{\Sigma}_{ss}$ is the union of all the conormals $N^*(W^u(\Lambda_i))$ 
 where $1\leq i\leq K$ and we aim at constructing a neighborhood of $\mathbf{\Sigma}_{ss}$. In this second step of the induction, we have indeed
 constructed a neighborhood $\mathcal{V}_j$ of $N^*(W^u(\Lambda_j))$ for $j\in J_1$ because we know from the first step of the induction 
 that every element 
 $(z;\zeta)\in \mathbf{\Sigma}_{ss}$ with 
 $z\in\cup_{i\in J_0}B_{i}$ belongs to $\cup_{i\in J_0}V_{i}$. In order to continue the procedure, we need to 
 ensure the same property near $\Lambda_j$ for $j\in J_1$.

Hence, the last thing we want to impose is that, if $(z;\zeta)$ belongs to $T^*\mathcal{B}_j\cap\Sigma_{ss}$, then necessarily $(z;\zeta)$ belongs to 
$\mathcal{V}_j$ which means we have really constructed some neighborhood of $\Sigma_{ss}$. 
Again, we fix $\eps_j''>0$ and we argue by contradiction as in 
subsection \ref{ss:neighborhoodproperty}. Again, we suppose that, 
for every $\eps_j>0$, we can find point of $\Sigma_{ss}$ near $\Lambda_j$ not lying in this neighborhood. Then, up to an 
extraction, we can find a sequence of points 
inside\footnote{Note, from the first step of the induction, 
that such a sequence cannot be contained inside $\cup_{i\in J_0}V_{i}$.}
$\Sigma_{ss}$ that would converge to a point not belonging to 
$N^*(W^u(\Lambda_j))\cup \left(\cup_{i\in J_0}V_{i}\right)$. As in paragraph~\ref{ss:neighborhoodproperty}, this would 
contradict the fact that $\Sigma_{ss}$ is compact.
 
\subsection{Induction on Smale's partial ordering: the case of closed orbits} We now treat the case where $\Lambda_j$ is a closed orbit. 
The procedure is more or less the same but we repeat it to take into account 
the effects on the $\theta$ variable and its dual variable $\Theta$.
Recall that, in this case, 
$N^*(W^u(\Lambda_j))$ can be written in the linearizing chart as~:
$$\left\{\left(P(\theta,0)(0,\tilde{y}),\theta;(P(\theta,0)^T)^{-1}(\tilde{\xi},0),\Theta(\theta,\tilde{y},\tilde{\xi})\right)
:\tilde{y}\in\IR^{n_u},\ 
\theta\in \IR/\ml{P}_{\Lambda}\IZ,\ 
\tilde{\xi}\in\IR^{n_s}\setminus
\{0\}
\right\},$$
where $\Theta(\theta,\tilde{y},\tilde{\xi})=-\la\partial_{\theta}P(\theta,0)\tilde{x},(P(\theta,0)^T)^{-1}\tilde{\xi}\rangle$.
As above, we fix three small parameters $\eps_j,\eps_j',\eps_j''>0$ that will be adjusted in terms of the dynamics and we define first a neighborhood inside $M$:
 $$B_j:=\left\{(P(\theta,0)(\tilde{x},\tilde{y}),\theta)\in\IR^{n-1}\times(\IR/\ml{P}_{\Lambda_j}\IZ):\int_0^{+\infty}\|e^{-t\Omega_s}\tilde{x}\|dt<\eps_j\ \text{and}\ \|\tilde{y}\|<\eps_j'\right\}.$$
Arguing as with 
fixed points, we can find a uniform $T_j>0$ such that every point inside $\overline{B_j}$ (with $\|\tilde{y}\|=\eps_j'$) will belong to one of the 
neighborhoods $B_{i}$ with $i\in J_0$ at time $T_j$.

We now need to have a look at the way we lift this neighborhood in the cotangent space. As before, we define a new norm $\mathbf{N}_j$ 
on $\IR^{n_s}$ such that 
$\mathbf{N}_j\left(e^{-t_0\Omega_s^T}\tilde{\xi}\right)
\leq e^{-t_0\lambda_j}\mathbf{N}_j\left(\tilde{\xi}\right)$ 
for every $t_0\geq 0$ and every $\tilde{\xi}\in\IR^{n_s}$. 
Here, $\lambda_j>0$ 
is one more time a small enough parameter. Mimicking what we have done before, we then define
 $$V_{j}:=\left\{\left(P(\theta,0)(\tilde{x},\tilde{y}),\theta;
 (P(\theta,0)^T)^{-1}(\tilde{\xi},\tilde{\eta}),\Theta\right)\in T^*B_j\backslash 0 
 : (*)\ \text{holds}\right\},$$
where $(*)$ now means that
$$\eps_j''\mathbf{N}_j\left(\tilde{\xi}\right)>\int_0^{+\infty}\|e^{-t\Omega_u^T}\tilde{\eta}\|dt
+\left|\Theta+\langle\partial_{\theta}P(\theta,0)(\tilde{x},\tilde{y}),(P(\theta,0)^T)^{-1}(\tilde{\xi},\tilde{\eta})\rangle\right|.$$
%Again, the integrals on the r.h.s converge thanks to Lemma~\ref{r:remainder} and to~\eqref{e:lyapunov-bound-closed-orbit}. 
This defines an open conical set inside $T^*V_j$, and we want to check that property~$(*)$ is 
preserved under the action of the forward flow $\Phi^{t_0}$ for $t_0\geq 0$~:
$$
 \int_0^{+\infty}\|e^{-(t+t_0)\Omega_u^T}\tilde{\eta}\|dt\leq\int_0^{+\infty}\|e^{-t\Omega_u^T}\tilde{\eta}\|dt,\ \mathbf{N}_j\left(\tilde{\xi}\right)\leq
\eps_j''\mathbf{N}_j\left(e^{t_0(\Omega_s)^T}\tilde{\xi}\right)$$
 %\leq e^{\lambda_j t_0}\int_0^{+\infty}\|e^{-t\Omega_u^T}\tilde{\eta}\|dt,$$
and, thanks to~\eqref{e:linearize-hamiltonian-flow-periodic-orbit-2},
 $$\left|\Theta+\langle\partial_{\theta}P(\theta,0)(\tilde{x},\tilde{y}),(P(\theta,0)^T)^{-1}(\tilde{\xi},\tilde{\eta})\rangle\right|\hspace{3cm}$$
 $$=
 \left|\Theta+R(t_0,\tilde{x},\tilde{y},\theta,\tilde{\xi},\tilde{\eta})+\langle\partial_{\theta}P(\theta+t_0,0)
 (e^{-t_0\Omega_s}\tilde{x},e^{t_0\Omega_u}\tilde{y}),(P(\theta+t_0,0)^T)^{-1}(e^{t_0\Omega_s^T}\tilde{\xi},e^{-t_0\Omega_u^T}\tilde{\eta})\rangle\right|.$$
%Hence, using the property of the norm $\mathbf{N}_j$, we find
%$$\int_0^{+\infty}\|e^{-(t+t_0)\Omega_u^T}\tilde{\eta}\|dt+\int_0^{+\infty}|\Theta+R(t+t_0),\tilde{x},\tilde{y},\theta,\tilde{\xi},\tilde{\eta})|
%e^{-t\lambda_j}dt< \eps_j''e^{\lambda_j t_0}\mathbf{N}_j\left(\tilde{\xi}\right)\leq
%\eps_j''\mathbf{N}_j\left(e^{t(\Omega_s)^T}\tilde{\xi}\right).$$
Combining these three inequalities, we find that the property of the cotangent component is preserved under the action of the forward flow.
Arguing as for fixed points, we can ensure that up to decreasing the value of $\eps_j$ and $\eps_j''$ (in a way that depends only on $\eps_j'$ and on $\eps$), 
we can verify that, for every $(z;\zeta)\in V_j$ and for 
every $0\leq t\leq T_j$, the point $\Phi^t(z;\zeta)$ remains $\eps$-close to $N^*(W^u(\Lambda_j))$ in the sense of~\ref{t:preciseattractor}.

We now connect these two neighborhoods with the ones constructed for $j'\in J_0$. Again, we define the (forward) invariant sets: 
$$\mathcal{B}_j:=\bigcup_{t\geq 0}\varphi^t\left((\cup_{i\in J_0}B_{i})\cup B_j\right)\subset M,$$
and 
$$\mathcal{V}_j:=\bigcup_{t\geq 0}\Phi^t\left((\cup_{i\in J_0}V_{i})\cup V_j\right)\subset T^*M\backslash 0.$$ 
By construction, any point $(z;\zeta)$ in $\mathcal{V}_j$ 
is $\eps$-close to $\mathbf{\Sigma}_{ss}$. Once again, we can make use of the compactness of $\Sigma_{ss}$ to verify that, 
for $\eps_j>0$ small enough, the following holds. If $(z;\zeta)$ 
belongs to $T^*\mathcal{B}_j\cap\mathbf{\Sigma}_{ss}$, then necessarily $(z;\zeta)$ belongs to 
$\mathcal{V}_j$ -- see paragraph~\ref{ss:neighborhoodproperty}.

\subsection{The case of $\mathbf{\Sigma}_s$.}
  We now have to discuss what has to be 
  modified in the case of $\mathbf{\Sigma}_s$. 
  Recall that, in this case, $\cup_{x\in\Lambda_j}N^*(W^{uu}(x))$ can be written in the linearizing chart as
 $$\left\{\left(P(\theta,0)(0,\tilde{y}),\theta;(P(\theta,0)^T)^{-1}(\tilde{\xi},0),\Theta\right)
 :\tilde{y}\in\IR^{n_u},\ \theta\in \IR/\ml{P}_{\Lambda}\IZ,\ (\tilde{\xi},\Theta)\in\IR^{n_s+1}\setminus\{0\}
\right\}.$$
The argument is exactly the same and the only point that 
needs to be modified is the condition $(*)$ appearing in the definition of $V_j$. Precisely, we set $(*)$ to be the 
condition:
$$\eps_j''\left(\mathbf{N}_j\left(\tilde{\xi}\right)+
\left|\Theta+\langle\partial_{\theta}P(\theta,0)(\tilde{x},\tilde{y}),(P(\theta,0)^T)^{-1}(\tilde{\xi},\tilde{\eta})\rangle\right|\right)>
\int_0^{+\infty}\|e^{-t\Omega_u^T}\tilde{\eta}\|dt.$$
%Here, $\lambda_j$ is chosen small enough to ensure that both $\mathbf{N}_j\left(\tilde{\xi}\right)$ and the integral on $\tilde{\eta}$ are well defined. 
Again, this is an open conical condition. The same calculation as before shows that the condition $(*)$ is preserved under the action of the forward flow (as 
soon as we remain close to $\Lambda_j$).

%We verify that this condition is preserved under the action of the forward flow $\Phi^{t_0}$, and the rest of the proof works in the same manner. 
%For that purpose, we fix $t_0\geq 0$ and we write:
%$$\int_0^{+\infty}\|e^{-t\Omega_u^T}e^{-t_0\Omega_u^T}\tilde{\eta}\|e^{t\lambda_j}dt\leq e^{-t_0\lambda_j} \int_0^{+\infty}\|e^{-t\Omega_u^T}\tilde{\eta}\|e^{t\lambda_j}dt,$$
%and
%$$\int_0^{+\infty}|\Theta+R(-t+t_0,\tilde{x},\tilde{y},\theta,\tilde{\xi},\tilde{\eta})|e^{-t\lambda_j}dt\geq e^{-t_0\lambda_j}
%\int_0^{+\infty}|\Theta+R(-t,\tilde{x},\tilde{y},\theta,\tilde{\xi},\tilde{\eta})|e^{-t\lambda_j}dt.$$
%Combining both inequalities with the properties of $\mathbf{N}_j$, we get
%$$\int_0^{+\infty}\|e^{-t\Omega_u^T}e^{-t_0\Omega_u^T}\tilde{\eta}\|e^{t\lambda_j}dt\leq \eps_{j}''
%\left(\int_0^{+\infty}|\Theta+R(-t+t_0,\tilde{x},\tilde{y},\theta,\tilde{\xi},\tilde{\eta})|e^{-t\lambda_j}dt+\mathbf{N}_j\left(e^{t_0\Omega_s^T}\tilde{\xi}\right)\right),$$
%which yields the expected invariance property.

\subsubsection{Conclusion of the proof.} 
 
In order to conclude the construction, we continue this induction procedure up to the point where we have exhausted all the critical elements $(\Lambda_{j})_{j=1,\ldots, K}$. In the case of 
$\mathbf{\Sigma}_{ss}$, this means that we exhaust all the critical points such that $\text{dim}(W^u(\Lambda_j))>0$ and all closed orbits
such that $\text{dim}(W^u(\Lambda_j))>1$. In the case of $\mathbf{\Sigma}_s$, 
we have to go one step further and include closed orbits such that $\text{dim}(W^u(\Lambda_j))>0$.

Actually, our procedure allows more than just 
constructing neighborhoods of $\mathbf{\Sigma}_s$ and $\mathbf{\Sigma}_{ss}$. 
Indeed, for any subfamily $(\Lambda_j)_{j\in J}$ of critical elements, we could repeat the 
same construction inductively for the Smale partial order relation, and
we would prove the following Theorem that will be used in~\cite{DaRi17b}:
\begin{theo} Suppose that $\varphi^t$ is a $\ml{C}^1$-linearizable Morse-Smale flow.
Fix a subfamily $(\Lambda_j)_{j\in J}$ of critical elements.
Then for every $\eps>0$, there exists 
an $\eps$-neighborhood of
$$\bigcup_{j\in J}\bigcup_{W^u(\Lambda_i)\preceq W^u(\Lambda_j)}W^u(\Lambda_i)$$
which is invariant under application of the flow in forward times.
Similarly, there exist forward invariant small neighborhoods of the following 
subsets of $S^*M$:
$$\bigcup_{j\in J}\bigcup_{W^u(\Lambda_i)\preceq W^u(\Lambda_j)} N^*(W^u(\Lambda_i))\cap S^*M,$$
and
$$\bigcup_{j\in J}\bigcup_{W^u(\Lambda_i)\preceq W^u(\Lambda_j)} \bigcup_{x\in\Lambda_i}N^*(W^{uu}(x))\cap S^*M.$$

\end{theo}

\section{Construction of an escape function}
\label{s:escape}

We now proceed to the proof of our last dynamical statement before explaining the spectral construction that follows from 
such analysis. 
Precisely, one of the key ingredient in the spectral construction of 
Faure-Sj\"ostrand is the construction of a nice enough ``escape function''. In order to state the existence of an escape function in our framework, 
we introduce the following subset inside $T^*M$:
$$\Gamma_0:=\left\{(x;\xi)\in T^*M\backslash 0:\ \exists 1\leq j\leq K,\ x\in\Lambda_j\ \text{and}\ (x;\xi)\in N^*\left(W^{uu}(x)\right)\cap 
N^*\left(W^{ss}(x)\right)\right\}.$$
This set corresponds to the neutral direction of the flow above closed orbits of $\varphi^t$. In terms of local coordinates near a closed orbit $\Lambda$, it reads
$$\{(0,0,\theta;0,0,\Theta):(\theta,\Theta) \in(\IR/\ml{P}_{\Lambda}\IZ)\times(\IR-\{0\})\}.$$
The main dynamical ingredient to start our spectral construction 
is the following Lemma: 
\begin{lemm}[Escape function]\label{l:escape-function} Let $V$ be a $\ml{C}^{\infty}$ vector field inducing a $\ml{C}^1$ linearizable Morse-Smale flow. 
Let $(u,n_0,s)$ be elements in $\IR$ with $u<-2\|E\|_{\infty}<2\|E\|_{\infty}<n_0<s$
where $E$ is the energy function of Theorem~\ref{t:meyer}. 
Let $N_0$ be an arbitrarily small conic neighborhood of $\Gamma_0$ inside 
$T^*M\backslash 0$.

Then, there exists a smooth metric $\|.\|_x$, a smooth function $m(x;\xi)\in\ml{C}^{\infty}(T^*M)$ called an order function, taking values in $[u,s]$, and an escape function 
on $T^*M$ defined by
$$G_m(x;\xi):=\frac{1}{2}m(x;\xi)\log(1+f(x;\xi)^2),$$
where $f(x;\xi)\in\ml{C}^{\infty}(T^*M)$. Moreover, $f(x;\xi)$ is positive and positively homogeneous of degree $1$ for $\|\xi\|_x\geq 1$, 
$f(x;\xi)=\|\xi\|_x$ outside a small neighborhood of~$N_0$ and $f(x;\xi)=|\xi(V(x))|=|H(x;\xi)|$ in a 
small neighborhood of~$N_0$. Finally, one has
\begin{enumerate}
 \item  For $\|\xi\|_x\geq 1$, $m(x;\xi)$ depends only on $\xi/\|\xi\|_x$ and it takes values $\leq \frac{u}{4}$ (resp. $\geq\frac{n_0}{4}$ and $\geq\frac{s}{4}$) in a small neighborhood of 
 $\cup_{j=1}^KN^*(W^u(\Lambda_j))$ (resp.~$\Gamma_0$ and~$\cup_{j=1}^KN^*(W^s(\Lambda_j))$). It also takes value $\geq \frac{n_0}{10}-\frac{u}{2}$ outside a slightly larger neigborhood of $\cup_{j=1}^KN^*(W^u(\Lambda_j))$
 \item There exists $R>0$ such that, for every $(x;\xi)$ in $T^*M$ satisfying $\|\xi\|_x\geq R$, one has 
 $$X_{H}G_m(x;\xi)\leq0.$$
 \item If in addition $(x;\xi)\notin N_0$, then
 $$X_{H}G_m(x;\xi)\leq -C_m<0,$$
 with 
 $$C_m:=c\min(|u|,s),$$
 for some constant $c>0$ independent of $u$, $n_0$ and $s$.
\end{enumerate}
\end{lemm}

In fact, the metric in this Lemma is rather arbitrary. As we shall see, we only need to fix its value on the $\Lambda_j$ in a manner that depends 
on the hyperbolic dynamics. Note that, as $\xi(V(x))$ does not vanish in a neighborhood of 
$\Gamma_0$, the function $f(x;\xi)$ does not vanish on $S^*M$. The same result was proved in the Anosov framework in~\cite[Lemma~1.2]{FaSj11}. The main input compared with that case are Meyer's Theorem~\ref{t:meyer} 
on the one hand and Theorems~\ref{t:compactness} and~\ref{t:attractor} on the other. 
We now proceed to the construction of 
the escape function.

\subsection{A preliminary lemma}

The main lines of the argument are very close to the one given for Anosov vector fields by Faure and Sj\"ostrand in~\cite{FaSj11}. Yet, several steps need to be 
adapted in order to fit into our dynamical framework. To begin with, we recall a dynamical Lemma from~\cite[Lemma~2.1]{FaSj11}:
\begin{lemm}\label{l:faure-sjostrand} Let $V^{uu}$ and $V^s$ be small open neighborhoods of $\Sigma_{uu}$ and $\Sigma_s$ respectively, and let $\eps>0$. Then, 
there exist smaller open neighborhoods $\ml{O}^{uu}\subset V^{uu}$ and $\ml{O}^s\subset V^s$ of $\Sigma_{uu}$ 
and $\Sigma_s$ respectively, $\tilde{m}_1$ in $\ml{C}^{\infty}(S^*M,[0,1])$, $\eta_1>0$ such that $\tilde{X}_{H}\tilde{m}_1\geq 0$ on $S^*M$, 
$\tilde{X}_{H}\tilde{m}_1\geq\eta_1>0$ on $S^*M-(\ml{O}^{uu}\cup\ml{O}^s)$, $\tilde{m}_1(x;\xi)>1-\epsilon$ for $(x;\xi)\in \ml{O}^s$ and $\tilde{m}_1(x;\xi)<\eps$ for $(x;\xi)\in \ml{O}^{uu}$. 
\end{lemm}

Recall that $\tilde{X}_H$ is the vector field induced by the Hamiltonian $H(x;\xi)=\xi(V(x))$ on $S^*M$. 
This Lemma was proved in~\cite[Lemma 8]{FaSj11} for Anosov flows, and we can follow the same strategy now that we have properly settled the dynamical properties of the symplectic lift of Morse-Smale flows. 
 For the sake of completeness, we briefly recall the proof as this is where we will crucially use Theorems~\ref{t:compactness} and~\ref{t:attractor}. 
\begin{proof} As $\Sigma_{uu}$ and $\Sigma_s$ are compact from Theorem~\ref{t:compactness} and as their intersection is empty thanks to~\eqref{e:transversality}, we may  
suppose
without loss of 
generality that $V^{uu}$ and $V^s$ have 
empty intersection. Thanks to Theorem~\ref{t:attractor}, we may also assume that
$$\forall t\geq 0,\ \tilde{\Phi}^t(V^s)\subset V^s\quad \text{and}\quad \forall t\leq 0,\ \tilde{\Phi}^t(V^{uu})\subset V^{uu}.$$
We are now in position to (briefly) repeat the argument from~\cite[Lemma 8]{FaSj11}. Before that, note that, while these properties were easy to check in the Anosov case, they constitute  
the core 
of our construction in the Morse-Smale case.

The argument is as follows. Fix $T>0$ and set $\ml{O}^s=\tilde{\Phi}^T(S^*M\setminus V^{uu})$ and 
$\ml{O}^{uu}=\tilde{\Phi}^{-T}(S^*M\setminus V^{s})$. Using Lemma~\ref{l:limit-attractor} and Theorem~\ref{t:attractor}, 
we know that for $T>0$ large enough
$\ml{O}^{uu}\subset V^{uu}$ and $\ml{O}^s\subset V^s$. 
We now fix a smooth function $m_0$ in $\ml{C}^{\infty}(S^*M,[0,1])$ such that $m_0=1$ on 
$V^s$ and $m_0=0$ on $V^{uu}$. We then set~: 
$$\tilde{m}_1(x;\xi):=\frac{1}{2T}\int_{-T}^Tm_0\circ\tilde{\Phi}^t(x;\xi)dt,$$
and we have~:
$$\tilde{X}_H\tilde{m}_1(x;\xi)=\frac{1}{2T}\left(m_0\circ\tilde{\Phi}^T(x;\xi)-m_0\circ\tilde{\Phi}^{-T}(x;\xi)\right)\geq 0.$$
In particular, for $(x;\xi)\notin \ml{O}^{uu}\cup\ml{O}^s$, one has $\tilde{X}_H\tilde{m}_1(x;\xi)=\frac{1}{2T}>0$. It now remains to prove the statements on the 
value of $\tilde{m}_1$ in $\ml{O}^{uu}\cup\ml{O}^s$. For that purpose, recall from our assumptions on $V^{uu}$ and $V^s$ that
$$\forall (x;\xi)\in S^*M,\ \ml{I}(x;\xi):=\left\{t\in\IR:\tilde{\Phi}^t(x;\xi)\notin V^{uu}\cup V^s\right\}$$
 is a closed connected interval whose length is bounded by some constant $\tau>0$ (that does not depend on $T$). Then, one can verify that
 $$(x;\xi)\in\ml{O}^s\Rightarrow \tilde{m}_1(x;\xi)\geq\frac{2T-\tau}{2T}\quad\text{and}\quad(x;\xi)\in\ml{O}^{uu}\Rightarrow \tilde{m}_1(x;\xi)<\frac{\tau}{2T},$$
 which concludes the proof of the preliminary Lemma. Note that, similarly, $\tilde{m}_1(x;\xi)\geq\frac{1}{2}-\frac{\tau}{2T}$ if $(x,\xi)\in V^s$ and 
 $\tilde{m}_1(x;\xi)\leq\frac{1}{2}+\frac{\tau}{2T}$ if $(x,\xi)\in V^{uu}$.
\end{proof}

\begin{rema}\label{r:Lyapunov} We remark that by inverting the sense of time, we can obtain a similar result for $\Sigma_{ss}$ and $\Sigma_u$. More precisely, we let  
$V^{ss}$ and $V^u$ be small open neighborhoods of $\Sigma_{ss}$ and $\Sigma_u$ respectively, and we let $\eps>0$. Then, 
there exist open neighborhoods $\ml{O}^{ss}\subset V^{ss}$ and $\ml{O}^u\subset V^u$, $\tilde{m}_2$ in $\ml{C}^{\infty}(S^*M,[0,1])$, $\eta_2>0$ such that $\tilde{X}_{H}\tilde{m}_2\geq 0$ on $S^*M$, 
$\tilde{X}_{H}\tilde{m}_2\geq\eta_2>0$ on $S^*M-(\ml{O}^{ss}\cup\ml{O}^u)$, $\tilde{m}_2(x;\xi)>1-\epsilon$ for $(x;\xi)\in \ml{O}^{ss}$ and $\tilde{m}_2(x;\xi)<\eps$ for $(x;\xi)\in \ml{O}^{u}$. 
\end{rema}

\subsection{Proof of Lemma~\ref{l:escape-function}}\label{ss:proof-escape}
We can now start the construction of the escape function. Again, the proof follows closely what is done in the Anosov case except that a couple of steps need to be revisited. By homogeneity, 
it is sufficient to perform the construction of $m$ inside $S^*M$ and then extend 
the definition by homogeneity for $\|\xi\|_x\geq 1$.

We keep the constants $u<0<n_0<s$ and  
the conic neighborhood $N_0$ of $\Gamma_0$ exactly 
as in the statement of Lemma~\ref{l:escape-function}. 
Let us also
consider the constant $\eps>0$ and the collection $\ml{O}^{uu}\subset V^u, \ml{O}^s\subset V^s, \ml{O}^{ss}\subset V^{ss}, \ml{O}^u\subset V^u$ of neighborhoods 
from the statements of Lemma~\ref{l:faure-sjostrand} and Remark~\ref{r:Lyapunov}. Following~\cite[paragraph 3.3.1]{FaSj11} and~\cite[appendix A]{DaRi16}, we set
\begin{equation}\label{e:order-unit}
 \tilde{m}(x;\xi):=-E(x)+s+(n_0-s)\tilde{m}_1+(u-n_0)\tilde{m}_2,
\end{equation}
where $E$ is the energy function of Theorem~\ref{t:meyer}. We already observe that
\begin{equation}\label{e:m-decrease}\forall(x;\xi)\in S^*M,\ \tilde{X}_H\tilde{m}(x;\xi)\leq 0
\end{equation}
since the function $-E$ decreases along the flow, $u-n_0<0, n_0-s<0$ and
$\tilde{X}_H\tilde{m}_i\geqslant 0$, for $i\in\{1,2\}$.

\subsubsection{First properties}

We start by collecting a few properties of the function
$\tilde{m}$ on $S^*M$.
First of all, one knows from Meyer's Theorem~\ref{t:meyer} that, if we fix $\delta>0$ and if $x$ does not belong to a $\delta$-neighborhood $\ml{V}_{\delta}$ of $NW(\varphi^t)$, then there exists some 
constant $c(\delta)>0$ such that
\begin{equation}\label{e:decay-energy-function}\forall (x;\xi)\notin S^*\ml{V}_{\delta},\ \tilde{X}_H\tilde{m}(x;\xi)\leq -c(\delta)<0.\end{equation}
This shows that away from the critical elements of the flow, $\tilde{m}$ is strictly decaying with a control in terms of the distance to the critical elements.

We now analyze more precisely the properties of $\tilde{m}$ near the critical elements. For that purpose, we set $\eta:=\min(\eta_1,\eta_2)$. Using the same convention as in~\cite{FaSj11} 
(in order to facilitate the comparison), we now define several open sets as follows~:
$$\tilde{N}_s:=\ml{O}^{uu}\cap\ml{O}^u,\ \tilde{N}_0:=\ml{O}^u\cap\ml{O}^{s},\ \tilde{N}_u:=\ml{O}^{ss}\cap\ml{O}^s.$$
We already observe that $\left(\tilde{N}_s,\tilde{N}_0,\tilde{N}_u\right)$ define disjoint open neighborhoods inside $S^*M$ of
$\left(\Sigma_{uu},\Gamma_0\cap S^*M,\Sigma_{ss}\right)$ respectively. 
Then, as in~\cite[p.~338]{FaSj11} we can verify the following properties:
\begin{itemize}
 \item On $S^*M-(\tilde{N}_s\cup\tilde{N}_0\cup\tilde{N}_u)=(S^*M-(\ml{O}^{uu}\cup \ml{O}^s))\cup(S^*M-(\ml{O}^{ss}\cup \ml{O}^u))$, one has
 \begin{equation}\label{e:decay-far}
  \tilde{X}_H\tilde{m}(x;\xi)\leq -\eta\min(|n_0-s|,|u-n_0|),
 \end{equation}
 with $\eta=\min\{\eta_1,\eta_2\}$ since, for 
$(x;\xi)$ in this region, either
$\tilde{X}_H\tilde{m}_1(x;\xi)\geqslant\eta_1\geqslant \eta$ or 
$\tilde{X}_H\tilde{m}_2(x;\xi)\geqslant\eta_2\geqslant \eta$.

 \item On $\tilde{N}_s=\ml{O}^{uu}\cap\ml{O}^u$, one has both $\tilde{m}_1\leqslant \epsilon, \tilde{m}_2\leqslant \epsilon$ therefore
$$\tilde{m}(x;\xi)=\underset{\geq -\Vert E\Vert_\infty}{\underbrace{-E(x)}}+s+
\underset{\geqslant \epsilon(u-s)}{\underbrace{(n_0-s)\tilde{m}_1+(u-n_0)\tilde{m}_2}}  $$ hence
by $\Vert E\Vert_\infty <\frac{s}{2}$, we find that
 \begin{equation}\label{e:bound-stable}
  \tilde{m}(x;\xi) \geq -\Vert E\Vert_\infty+u\eps+(1-\eps) s\geq \left(\frac{1}{2}-\eps\right)s+u\eps\geq\frac{s}{4},
 \end{equation}
 where the last inequality holds provided we take $\eps<\frac{s}{4(s-u)}$.
 \item Similarly, on $\tilde{N}_u=\ml{O}^{ss}\cap\ml{O}^s$, one has both $\tilde{m}_1\geqslant 1- \epsilon, \tilde{m}_2\geqslant 1- \epsilon$ therefore
$$\tilde{m}(x;\xi)=\underset{\leq -\text{min}E}{\underbrace{-E(x)}}+s+
\underset{\leqslant (1-\epsilon)(u-s)}{\underbrace{(n_0-s)\tilde{m}_1+(u-n_0)\tilde{m}_2}}  $$ hence
by $ \frac{u}{2}<-\Vert E\Vert_\infty$, we find that
  \begin{equation}\label{e:bound-unstable}
  \tilde{m}(x;\xi)\leq -\text{min}\ E+u(1-\eps)+s\eps \leq \left(\frac{1}{2}-\eps\right)u+s\eps\leq\frac{u}{4},
 \end{equation}
 where the last inequality holds provided we take $\eps<\frac{u}{4(u-s)}$.
  \item On $\tilde{N}_0=\ml{O}^u\cap\ml{O}^{s}$, 
  one has both $1-\varepsilon < \tilde{m}_1\leqslant 1 $ and $\tilde{m}_2\leqslant \varepsilon$ 
  hence
  $$ \tilde{m}(x;\xi)=\underset{\geq -\Vert E\Vert_\infty}{\underbrace{-E(x)}}+s+
\underset{\geqslant (n_0-s)+(u-n_0)\varepsilon}{\underbrace{(n_0-s)\tilde{m}_1+(u-n_0)\tilde{m}_2}}   $$
  \begin{equation}\label{e:bound-neutrla}
  \tilde{m}(x;\xi)\geq -\text{max}\ E+(1-\eps)n_0+(u-n_0)\eps\geq \left(\frac{1}{2}-2\eps\right)n_0+u\eps \geq\frac{n_0}{4},
 \end{equation}
 where the last inequality holds provided $\eps<\frac{n_0}{4(2n_0-u)}.$
 \item Finally, by similar arguments and using the final remark in the proof of Lemma~\ref{l:faure-sjostrand}, we can verify that 
 $\tilde{m}(x;\xi)\geq\frac{s}{10}-\frac{u}{2}$ outside a slightly bigger set than $\tilde{N}_u$.
\end{itemize}

We now extend $\tilde{m}$ into a smooth function $m$ defined on $T^*M$ by setting $m(x;\xi)=0$ for $\|\xi\|_x\leq 1/2$ and
$$\forall(x;\xi)\in T^*M\ \text{s.t. } \|\xi\|_x\geq 1,\ m(x;\xi)=\tilde{m}\left(x;\frac{\xi}{\|\xi\|_x}\right).$$
From~\eqref{e:bound-stable},~\eqref{e:bound-unstable} and~\eqref{e:bound-neutrla}, we remark that 
for $\|\xi\|_x\geq 1$, $m(x;\xi)$ depends only on $\xi/\|\xi\|_x$ and it takes values $\leq \frac{u}{4}$ (resp. $\geq\frac{n_0}{4}$ and $\geq\frac{s}{4}$) in a small neighborhood of 
$\cup_{j=1}^KN^*(W^u(\Lambda_j))$ (resp.~$\Gamma_0$ and~$\cup_{j=1}^KN^*(W^s(\Lambda_j))$).
Hence point (1) of Lemma~\ref{l:escape-function} is proved. 

\subsubsection{Decay of $G_m$ along the flow}

We start by defining $f(x;\xi)$ according to the neighborhoods that have already been introduced.
Recall that $\ml{V}_\delta$ was some $\delta$--neighborhood 
of the nonwandering set. 
We set $f$ to be  
a smooth function which is positive and positively homogeneous of degree $1$ for $\|\xi\|_x\geq 1$. Moreover, we suppose that for
 $\xi\in N_0\cap T^*\ml{V}_{\delta}$, one has 
$f(x;\xi)=|H(x;\xi)|$ and $f(x;\xi)=\|\xi\|_x$ when $\xi$ does not belong to a small open neighborhood of $N_0$. 
Without loss of generality, we can 
assume that $f(x;\xi)\geq C_0\|\xi\|$ for $\|\xi\|\geq 1$ and for some positive constant $C_0>0$ depending on $N_0$. 

\begin{rema}\label{r:neighbohoods} We observe that, up to choosing a smaller value of $\delta$ and smaller neighborhoods in the statements of Lemma~\ref{l:faure-sjostrand} and Remark~\ref{r:Lyapunov}, 
we can always suppose that $T^*\ml{V}_{\delta}\cap\tilde{N}_0\subset S^*M\cap N_0$. 
\end{rema}
Let us now discuss the decay 
properties of $$G_m:=\frac{1}{2}m(x;\xi)\log\left(1+f^2(x;\xi) \right)$$ 
and prove properties (2) and (3) of Lemma \ref{l:escape-function}. 
It suffices to show that, in some conic neighborhood of $N_0$, we have the bound $X_HG_m\leqslant 0$ and outside some conic neighborhood of
$N_0$, we have a uniform negative bound on $X_HG_m$. We decompose $X_HG_m$ as the sum of two terms~: 
\begin{equation}\label{e:decompfuite}
2X_HG_{m}(x;\xi)= \left(X_Hm\right)(\log(1+f(x;\xi)^2))+m(x;\xi)X_H(\log(1+f(x;\xi)^2)) .
\end{equation}
We will split the discussion between points $x\in M$ which are close to the nonwandering set and those which are far from it.

\begin{rema}\label{r:crude-bound}
Before discussing this, we observe that we have a crude upper bound on the term
$m(x;\xi) X_H\log(1+f(x;\xi)^2)$ which is uniform for $(x;\xi)\in T^*M$ satisfying
$\|\xi\|_x$ larger than some $R>0$. For that purpose, we first remark that
\begin{equation}\label{e:derivative-log}
 X_H\log(1+f(x;\xi)^2)= \frac{X_Hf(x;\xi)^2}{1+f(x;\xi)^2},
\end{equation}
defines a bounded function on $T^*M$ thanks to the homogeneity properties of $f$. Then, we just note that $\tilde{m}$ is uniformly 
bounded by $\ml{O}(|u|+s+\|E\|_{\infty})$.
\end{rema}

\textbf{ The element $x$ does not belong to the $\delta$ neighborhood $\ml{V}_\delta$ of the nonwandering set.} Fix an element $(x;\xi)\notin T^*\ml{V}_{\delta}$ 
satisfying $\|\xi\|_x\geq 1$. Recall that we have the estimates $$X_Hm(x;\xi)\leqslant -c(\delta)\quad \text{and}\quad
\log(1+f(x;\xi)^2)\geqslant \log(1+C_0 \vert\xi\vert^2).$$ Combined with~\eqref{e:decompfuite} and with Remark~\ref{r:crude-bound}, this yields the bound
\begin{equation}\label{e:decay-far-escape}
 2X_HG_{m}(x;\xi)\leq -c(\delta)\log(1+C_0^2\|\xi\|^2)+C_1(\|E\|_{\infty}+s+|u|),
\end{equation}
for some positive constant $C_1>0$ depending only on $N_0$. Hence, provided $R$ is large enough (in terms of $(C_i)_{i=0,1}$, $\delta$, $u$ and $s$) in the statement of Lemma~\ref{l:escape-function}, 
parts (2) and (3) of the Lemma are satisfied for these points of phase space. 

\textbf{The element $x\in \ml{V}_\delta$ and $(x;\xi)\in N_0$.} It now remains to analyze the situation near the nonwandering set (and thus fix a small enough value of $\delta$). 
Here, the situation follows closely what was done in~\cite[p.~339]{FaSj11}. We note that as $X_HH=\{H,H\}=0$, one can show that for $(x;\xi)\in N_0$ satisfying $\|\xi\|_x\geq 1$, one has
\begin{equation}\label{e:decay-neutral-escape}
 X_HG_m(x;\xi)=X_Hm(x;\xi)\log(1+H(x;\xi)^2)^{\frac{1}{2}}+m(x;\xi)X_H\log(1+H(x;\xi)^2)\leq 0.
\end{equation}
This proves part (2) of the Lemma in this region of phase space.

\textbf{The element $x\in \ml{V}_\delta$ and $(x;\xi)\notin N_0$.} Thanks to remark~\ref{r:neighbohoods}, we can
cover $S^*\ml{V}_\delta\setminus N_0$ by three regions where we will prove uniform decay of $G_m$ along the Hamiltonian flow~: 
\begin{eqnarray*}
S^*\ml{V}_{\delta}\setminus N_0\subset\left(\tilde{N}_u\cap S^*\ml{V}_{\delta} \right)\bigcup\left( \tilde{N}_s\cap S^*\ml{V}_{\delta}\right) \bigcup\left( S^*\ml{V}_{\delta}\setminus(\tilde{N}_0\cup\tilde{N}_u\cup\tilde{N}_s)\right).
\end{eqnarray*}
\begin{itemize}
\item We begin with the case where $(x;\tilde{\xi})\in S^*M\setminus(\tilde{N}_0\cup\tilde{N}_u\cup\tilde{N}_s)$.
In that case, $f(x;\xi)\geq C_0\|\xi\|_x$ by construction (at least for $\|\xi\|_x$ large enough). Hence, it follows from~\eqref{e:decay-far} that
\begin{eqnarray*}\label{e:decay-far-3}
 2X_HG_m(x;\xi) & = &\left( X_Hm\right)(x;\xi)\log(1+f(x;\xi)^2)+m(x;\xi)X_H\log(1+f(x;\xi)^2)\\
  &\leq & -\eta\min(|n_0-s|,|u-n_0|)\log(1+C_0^2\|\xi\|_x^2)+C_1(\|E\|_{\ml{C}^0}+s+|u|).
\end{eqnarray*}
Hence, as for the case $x\notin\ml{V}_{\delta}$, 
we can ensure that parts (2) and (3) of the Lemma are satisfied by picking $R>0$ large enough 
(in a way that depends on $\eta$, $u$, $s$ and $N_0$) so that the negative
term $-\eta\min(|n_0-s|,|u-n_0|)\log(1+C_0^2\|\xi\|_x^2)$ predominates
over the positive term $C_1(\|E\|_{\ml{C}^0}+s+|u|)$ and makes sure the r.h.s. is uniformly negative as long as $\|\xi\|_x \geqslant R$. 
\item Assume now that $(x;\xi/\|\xi\|_x)\in \tilde{N}_u\cap S^*\ml{V}_{\delta}$.
We  
show that we can pick $\delta>0$ small enough to ensure that the bound holds for $(x;\xi/\|\xi\|_x)\in \tilde{N}_u\cap S^*\ml{V}_{\delta}$.
In this case $f(x;\xi)=\|\xi\|_x$,
hence
\begin{eqnarray*}
2X_HG_m(x;\xi)& = &\underset{\leqslant 0}{\underbrace{\left(X_Hm\right)}} \underset{\geqslant 0}{\underbrace{\log\left(1+\|\xi\|_x^2\right)}}+ 
m(x;\xi)X_H\log\left(1+\|\xi\|_x^2\right) \\
&\leq & m(x;\xi)\frac{X_H(\|\xi\|_x^2)}{2(1+\|\xi\|_x^2)}.
\end{eqnarray*}
We will now make use of the hyperbolicity of the flow one more time
in order to control the term $X_H(\|\xi\|_x^2)$
\footnote{Observe that hyperbolicity was already used in Meyer's Theorem and in 
the proof of Lemma~\ref{l:faure-sjostrand} which relied on Theorems~\ref{t:compactness} and~\ref{t:attractor}.}.
We assume without loss of generality (periodic orbits are treated similarly)
that we are near a critical point $\Lambda_j$. For $(x;\xi)$ in $T_{\Lambda_j}^*M\cap N^*(W^u(\Lambda_j))$, one can deduce from the 
hyperbolicity bound~\eqref{e:lyapunov-bound-fixed-point} that
$$X_{H}\|\xi\|_x^2\geq c_0\|\xi\|_x^2,$$
%We now use the linearized expression of the Hamiltonian flow from paragraph~\ref{ss:coordinates}. In these coordinates, a point 
%$(z;\zeta)\in T_{\Lambda_j}^*M\cap N^*(W^u(\Lambda_j))$ is of the form $(0;\tilde{\xi},0)$. Hence, one has 
%\begin{eqnarray*}
%X_H\|(\tilde{\xi},0)\|_z^2&=&  \frac{d}{dt}|_{t=0} \left(\| (e^{t\Omega_s^T}\tilde{\xi},0) \|_z^2\right) \\
%&=& \langle (\Omega_s^T+\Omega_s)\tilde{\xi}, \tilde{\xi} \rangle
%  \geq c_0\|\zeta\|_z^2,
%\end{eqnarray*}
for some positive constant $c_0$ depending only on $(M,g)$ and $\varphi^t$. This can be achieved by choosing the value of $\|.\|_x$ for 
$x\in\Lambda_j$ in a way to have the constant equal to $1$ in~\eqref{e:lyapunov-bound-fixed-point} -- see~\cite[Prop.~5.2.2]{BrSt02}.
By compactness of $\Sigma_{ss}$, we can ensure that this inequality remains true in 
$\tilde{N}_u\cap S^*\ml{V}_{\delta}$, where $\tilde{N}_u$ is a conical neighborhood of $N^*(W^u(\Lambda_j))$, 
provided that we replace $c_0$ by $c_0/2$ and that $\delta>0$ and all neighborhoods $(V^{ss},V^{uu},V^s,V^u)$
from Lemma~\ref{l:faure-sjostrand} are chosen small enough. Thus, thanks to inequality~\eqref{e:bound-unstable} that yields
$\tilde{m}(x;\xi)\leqslant \frac{u}{4}<0$, one finds that
$$\left(x,\frac{\xi}{\|\xi\|_x}\right)\in\tilde{N}_u\cap S^*\ml{V}_{\delta}\Longrightarrow X_HG_m(x;\xi)\leq\frac{c_0 u}{8},$$
where $c_0>0$ is a geometric constant. 

\item The case where $(x;\tilde{\xi})\in \tilde{N}_s\cap S^*\ml{V}_{\delta}$ is treated  similarly thanks to~\eqref{e:bound-stable}. 
Then we can derive similar bounds of the form
$$\left(x,\frac{\xi}{\|\xi\|_x}\right)\in\tilde{N}_s\cap S^*\ml{V}_{\delta}\Longrightarrow X_HG_m(x;\xi)\leq -\frac{c_0 s}{8},$$
which concludes the construction of the escape function from Lemma~\ref{l:escape-function}.
\end{itemize}

\section{Anisotropic Sobolev spaces}\label{s:sobolev}

In this section, we aim at proving 
Theorems~\ref{t:maintheo1} and~\ref{t:maintheo2} or more precisely 
their generalization to vector bundles 
$(\ml{E},\nabla)$ equipped with a connection $\nabla$. 
More precisely, we will consider the Lie derivative $-\ml{L}_{V,\nabla}$ acting
on currents on $M$ valued in the bundle $\ml{E}$ and the connexion $\nabla$ allows to lift
the action of the flow from the base $M$ to the total space of $\ml{E}$. 
The reason why we need to introduce such algebra
is that we plan to discuss in~\cite{DaRi17c} the relationship between the Pollicott-Ruelle spectrum and
the Reidemeister torsion. To define 
the Reidemeister torsion, we need to fix a
unitary representation of the fundamental
group $\pi_1(M)$ and then 
consider geometric complexes on $M$ twisted 
by the representation. Geometrically, the twisting operation corresponds  
to considering
some flat bundle $(\ml{E},\nabla)$ over $M$, such that
the monodromy of the flat connection 
$\nabla$ gives the corresponding representation
of $\pi_1(M)$. In that framework, the De Rham complex of differential forms $(\Omega^\bullet(M),d)$ is replaced by
the twisted De Rham complex of differential forms valued in $\ml{E}$, $\Omega^\bullet(M,\ml{E})$ with
twisted differential $d^\nabla$.
 
The purpose of this final section is to recall how one can construct anisotropic Sobolev spaces 
adapted to the operator $-\ml{L}_{V,\nabla}^{(k)}$ for every $0\leq k\leq n$ starting from the escape function of Lemma~\ref{l:escape-function}. 
More precisely, for every $\sigma>0$, we let $m_\sigma\in C^\infty(T^*M)$ be the corresponding order function
from Lemma \ref{l:escape-function} with $|u|,s\gg\sigma$, 
and for every $0\leq k\leq n$, we will 
construct Hilbert spaces of currents $\ml{H}_k^{m_{\sigma}}(M,\ml{E})$ such that
$$-\ml{L}_{V,\nabla}^{(k)}:\ml{H}_k^{m_{\sigma}}(M,\ml{E})\longrightarrow\ml{H}_k^{m_{\sigma}}(M,\ml{E})$$
has discrete spectrum on $\text{Re}(z)>-\sigma$ -- see Proposition~\ref{p:eigenvalues}. Once this spectral framework is properly 
settled, in paragraph~\ref{ss:proof-theo} we will see how 
one can easily deduce our main Theorems on the existence of a discrete set of Pollicott-Ruelle resonances. Hence, everything boils down to constructing appropriate Hilbert spaces, and 
now that we are given a nice escape function 
$m$ adapted to the Morse-Smale dynamics, we just need to follow step by step the (microlocal) strategy initiated by Faure and Sj\"ostrand in~\cite{FaSj11}. 
  Recall that this consists in showing that $-\ml{L}_{V,\nabla}^{(k)}$ has nice spectral properties on that space via microlocal arguments and analytic Fredholm theory. This 
kind of strategy is natural in the context of the study of semiclassical resonances of Schr\"odinger operators~\cite{HeSj86} and we refer 
to~\cite{FaSj11} for a more detailed discussion on that aspect.

The proof in~\cite{FaSj11} only deals with the scalar case for $k=0$ and $\ml{E}=M\times\IC$. 
Yet, the argument can be adapted to any $0\leq k\leq n$ 
as it only relies on the
construction of a proper escape function 
for the flow as in the scalar case. 
We emphasize that the case of vector bundles 
for Anosov flows was considered via microlocal methods 
by Dyatlov and Zworski in order to prove the 
meromorphic continuation of the zeta function~\cite{DyZw13}. 
We chose to use the microlocal techniques from 
these references but it is most likely that 
the methods developed in~\cite{BuLi07, GiLiPo13} 
could also be developed to define a proper spectral framework.

\begin{rema}
Keeping in mind the applications to differential topology~\cite{DaRi17c}, it may be convenient to choose a different $m_{\sigma}$ for every 
$0\leq k\leq n$. Typically, we may pick in degree $k$ an order function of the form $m(x;\xi)+n-k$ where $m(x;\xi)$ is given by 
Lemma~\ref{l:escape-function}. The following analysis would not be affected by this choice of (shifted) order function for large enough choices of $|u|$ and 
$s$ in this Lemma. 
\end{rema}

\subsection{Vector bundles and connection}\label{ss:connection} Before defining the Sobolev spaces, we collect a few definitions and properties and 
we refer to~\cite[Ch.6-8-12]{Lee09} for more details. 
For every integer $k\in \{0,\dots,n\}$, recall the bundle $\Lambda^k\left(T^*M\right)$ of $k$-forms on $M$ 
has all its elements
of the form
$\sum_{0\leqslant i_1<\dots<i_k \leqslant n} \xi_{i_1\dots i_k}dx^{i_1}\wedge \dots \wedge dx^{i_k}$ 
in local coordinates where
$(dx^1,\dots,dx^n)$ forms a basis of $T^*M$ in local coordinates. Smooth sections of 
$\Lambda^k\left(T^*M\right)$ are denoted by $\Omega^k(M)$.
If we are 
given a metric $g$ on $M$, this induces some 
inner product $\la .,. \ra_{g^*}^{(k)}$ on $\Lambda^k(T^*M)$. 
Then,
the Hodge star operator is the unique isomorphism 
$\star_k :\Lambda^k(T^*M)\rightarrow \Lambda^{n-k}(T^*M)$ such that 
for every $\psi_1\in\Omega^k(M)=\Omega^k(M,\IC)$ and $\psi_2\in \Omega^{n-k}(M)=\Omega^{n-k}(M,\IC)$, 
$$\int_M\psi_1\wedge\psi_2=\int_M\la \psi_1,\star_k^{-1}\psi_2\ra_{g^*}^{(k)}\omega_g,$$
where $\omega_g$ is the volume form induced by the Riemannian metric
$g$ on $M$.
Given a smooth Hermitian vector bundle $\pi:\ml{E}\rightarrow M$~\cite[Def.~6.21]{Lee09} 
of rank $N$ with inner product $\la.,.\ra_{\ml{E}}$ on the fibers, the
space $\Omega^k(M,\ml{E}), k=1,\dots,n$ of $k$-forms on $M$ valued in $\ml{E}$
is defined as the tensor product $\Omega^k(M)\otimes_{C^\infty(M)}\Omega^0(M,\ml{E})$ where
$\Omega^0(M,\ml{E})$ denotes the smooth sections of $\ml{E}$.
%We suppose without loss of generality that this vector bundle is equipped with a smooth Hermitian structure $\la.,.\ra_{\ml{E}}$. 
%We now consider the space $\Omega^k(M,\ml{E})$ of smooth differential $k$-forms with coefficients in $\ml{E}$.  
Using the hermitian structure on $\ml{E}$, 
we can identify $\ml{E}$ with its dual bundle and 
we can also introduce the following fiberwise pairing, 
for every $0\leq k,l\leq n$,
$$\la.\wedge.\ra_{\ml{E}}:\Omega^k(M,\ml{E})\times\Omega^l(M,\ml{E})\rightarrow \Omega^{k+l}(M).$$
This induces a map $\star_k:\Omega^k(M,\ml{E})\rightarrow \Omega^{n-k}(M,\ml{E})$ which acts trivially on the 
$\ml{E}$-coefficients. 
Then, we can define the 
positive definite Hodge inner product on $\Omega^{k}(M,\ml{E})$ as
$$(\psi_1,\psi_2)\in\Omega^k(M,\ml{E})\times\Omega^k(M,\ml{E})\mapsto\int_M\la \psi_1\wedge\star_k(\psi_2)\ra_{\ml{E}}.$$
In particular, we can define $L^2(M,\Lambda^k(T^*M)\otimes\ml{E})$ as the completion of $\Omega^k(M,\ml{E})$ for this scalar product.

We would now like to define an analogue of the Lie 
derivative along a vector field $V$ in the context of differential forms with values in a vector bundle $\ml{E}$. For that purpose, we observe that the contraction operator $\iota_{V}$ 
is well-defined and we have to introduce a proper substitute 
for the De Rham differential $d$. 
For this, we fix 
a smooth (Koszul) connection ~\cite[Def.~12.1 and p.~505]{Lee09}, i.e. a linear map 
$$\nabla:\Omega^0(M,\ml{E})
\rightarrow \Omega^1(M,\ml{E}),$$
satisfying, for every $\psi\in\ml{C}^{\infty}(M)$ and for every $u$ in $\Omega^0(M,\ml{E})$,
$$\nabla(\psi u)=\psi\nabla u+\left(d\psi\right)u ,$$
where $d$ is the usual De Rham differential on $M$. According to~\cite[Th.~12.57]{Lee09}, this map uniquely extends to a map $d^{\nabla}$ such that, for every $0\leq k,l\leq n$, 
\begin{itemize}
 \item $d^{\nabla}:\Omega^k(M,\ml{E})\rightarrow\Omega^{k+1}(M,\ml{E})$,
 \item for every $(\psi,u)$ in $\Omega^k(M)\times\Omega^l(M,\ml{E})$,
 $$d^{\nabla}(\psi\wedge u)=d\psi\wedge u+(-1)^k\psi\wedge d^{\nabla}u\ \text{and}\ d^{\nabla}(u\wedge\psi)=d^{\nabla}u\wedge\psi+(-1)^lu\wedge d\psi,$$
 \item for every $u$ in $\Omega^0(M,\ml{E})$, $d^{\nabla}u=\nabla u$.
\end{itemize}
The operator $d^{\nabla}$ is the coboundary operator associated with $(\ml{E},\nabla)$, and we define the corresponding \emph{Lie derivative} along the vector field $V$ as
\begin{equation}\label{e:lie-derivative}
\boxed{ \ml{L}_{V,\nabla}=d^{\nabla}\circ \iota_V+\iota_V\circ d^{\nabla}.}
\end{equation}

\begin{rema}\label{r:curvature} We define the \emph{curvature} of the vector bundle $(\ml{E},\nabla)$ as the unique map
$F_{\nabla}$ in $\Omega^2(M,\text{End}(\ml{E}))$ such that, for every $0\leq k\leq n$ and for every $u$ in $\Omega^k(M,\ml{E})$,
$$d^{\nabla}\circ d^{\nabla} u=F_{\nabla}\wedge u.$$
 We say $(\ml{E},\nabla)$ is a \emph{flat vector bundle} if $F_{\nabla}=0$, equivalently, if $d^{\nabla}\circ d^{\nabla}=0$ which means that
 $(\Omega^\bullet(M,\mathcal{E}),d^\nabla)$ is a cochain complex.
 \end{rema}
%
%Finally, one can define a connection 
%$$\nabla^{\dagger}:\Omega^0(M,\ml{E}')\rightarrow\Omega^1(M,\ml{E}'),$$
%which is completely determined by requiring the following property, for any $(\psi_1,\psi_2)$ in $\Omega^{0}(M,\ml{E})\times\Omega^{0}(M,\ml{E}')$,
%$$d(\psi_2(\psi_1))= \psi_2(\nabla \psi_1)+(\nabla^{\dagger}\psi_2)(\psi_1).$$
%This connection induces a coboundary operator $d^{\nabla^{\dagger}}:\Omega^k(M,\ml{E})\rightarrow\Omega^{k+1}(M,\ml{E})$. By construction, one has, 
%for any $\psi_1$ in $\Omega^{k-1}(M,\ml{E})$ and for any $\psi_2$ in $\Omega^{n-k}(M,\ml{E}')$,
%\begin{equation}\label{e:duality-coboundary}\int_M\psi_2(d^{\nabla}\psi_1)=(-1)^k\int_{M}(d^{\nabla^{\dagger}}\psi_2)(\psi_1).\end{equation}

In the following, $\ml{E}\rightarrow M$ is always a smooth complex vector bundle endowed with a Hermitian structure $\la,\ra_{\ml{E}}$ and a connection $\nabla$ (not necessarily flat).

\subsection{Definition of the anisotropic Sobolev spaces}\label{ss:anisotropic}
%In the present paper, following~\cite[appendix A 1.3]{FaRoSj08}, pseudodifferential operators
%on the Riemannian manifold $(M,g)$ are defined
%using the very convenient quantization
%scheme of~\cite{Pflaum}. There is an open neighborhood
%$W$ of the zero section $\underline{0}$ in $TM$ such that
%$\exp:W\subset TM\mapsto M\times M$ is diffeomorphic on its image.
%Choose some cut--off function $\chi\in C^\infty_c(W)$ s.t. $\chi=1$ near $\underline{0}\subset TM$.
%%Let $\nabla$ denote the unique torsion free connection on the bundle $\Lambda^kT^*M\otimes \mathcal{E}$, then
%%for every $(x,v)\in T_xM$ denote by
%%$\tau_{\exp_x(v)} $ the parallel transport along
%%$t\in [0,1]\mapsto \exp_x(tv)$.
%For any section $A\in C^\infty(T^*M)$,
%the corresponding operator $\widehat{A}$ is defined as~:
%$$u\in C^\infty(M)\mapsto
%\widehat{A}u(x)=\int_{T^*_xM}A(x;\xi)\tilde{u}(x;\xi) d^n\xi $$
%where $\int_{\mathbb{R}^n} e^{iv.\xi}\chi(v)u(\exp_x(v))d^nv=\tilde{u}(x;\xi)$.
Now, we set
\begin{equation}\label{A_mweightsobo}
A_m(x;\xi):=\exp\left(G_m(x;\xi)\right)\in C^\infty(T^*M),
\end{equation}
where $G_m(x;\xi)$ is given by Lemma~\ref{l:escape-function}. Let now $0\leq k\leq n$ and $\ml{E}\rightarrow M$ be a smooth vector bundle equipped with a Hermitian structure. 
We consider the vector bundle $\Lambda^k(T^*M)\otimes\ml{E}$.
Consider some finite cover $(U_i)_{i\in I}$ of $M$ by contractible open 
subsets. 
%For every $U_i$, choose some local coordinates $(x^1,\dots,x^n):U_i\mapsto \mathbb{R}^n$, 
%a local basis $(\mathbf{e}_{j})_{j=1}^N$ of
%the bundle $\ml{E}|_{U_i}$. We define
%some operator $\mathbf{A}_i: \Omega_c^k(U_i)\otimes \mapsto \Omega^k(M)$ 
%defined explicitely in local coordinates
%as follows~: 
%$$ \mathbf{A}_i(f\mathbf{e}_jdx^{i_1}\wedge\dots\wedge dx^{i_k})=\left(\int_{U_i\times \mathbb{R}^n} A_m(.;\xi)e^{i\la\xi,.-y\ra } f(y)d^nyd^n\xi\right) \mathbf{e}_jdx^{i_1}\wedge\dots\wedge dx^{i_k}$$
%for every $f\in C^\infty_c(U_i)$, $j\in \{1,\dots,r\}$ and $1\leqslant i_1<\dots<i_k\leqslant n$.
For every $U_i$, choose some local trivialization
$\kappa: 
\mathcal{E}\otimes 
\Lambda^k\left(T^*M\right)|_{U_i}\mapsto 
\tilde{U}\times \mathbb{R}^N\times \mathbb{R}^{\frac{n!}{k!(n-k)!}}$  
of $\mathcal{E}\otimes \Lambda^k\left(T^*M\right)$ 
where $\tilde{U}$ is some open subset of $\mathbb{R}^n$. 
Set $(\mathbf{e}_j)_{j=1}^N$ to be
the canonical basis of $\mathbb{R}^N$. 
Then 
we define
some operator $\mathbf{A}_i: \Omega_c^k(U_i,\ml{E}) \mapsto \Omega^k(M,\ml{E})$ 
explicitely in local coordinates
as follows~: 
$$ \kappa \circ \mathbf{A}_i\circ \kappa^{-1}(\mathbf{e}_j fdx^{i_1}\wedge\dots\wedge dx^{i_k})|_{\tilde{U}}=\mathbf{Op}^W(A_m)(f) \mathbf{e}_jdx^{i_1}\wedge\dots\wedge dx^{i_k}$$
for every $f\in C^\infty_c(U_i)$, $j\in \{1,\dots,r\}$, $1\leqslant i_1<\dots<i_k\leqslant n$ and where
$\mathbf{Op}^W(A_m)$ is the usual Weyl quantization defined on $\mathbb{R}^n$~\cite[equation 14.5 p.~68]{Taylor2}.
Choose some partition of unity $(\chi_i)_{i\in I}, \sum_{i\in I}\chi_i=1$ subordinated
to the cover $(U_i)_{i\in I}$ and for every $i\in I$, choose $\Psi_i\in C^\infty_c(U_i)$ such that
$\Psi_i=1$ on $\text{supp}(\chi_i)$.
Then define
a global linear operator
$\Omega^k(M,\mathcal{E}) \mapsto \Omega^k(M,\mathcal{E})$
as~:
\begin{equation}
\mathbf{A}_m^{(k)}=\frac{1}{2}\sum_{i\in I} \Psi_i \mathbf{A}_i \chi_i+\frac{1}{2} 
\left(\sum_{i\in I} \Psi_i \mathbf{A}_i \chi_i\right)^*
\end{equation}
where $\left(\sum_{i\in I} \Psi_i \mathbf{A}_i \chi_i\right)^*$ is the
formal adjoint of $\sum_{i\in I} \Psi_i \mathbf{A}_i \chi_i$ for 
the Hilbert space structure on $L^2(M,\Lambda^k(T^*M)\otimes \ml{E})$ defined above.
The operator $\mathbf{A}_m^{(k)}$ is elliptic in the sense of~\cite[definition 8]{FaRoSj08}
since its principal symbol reads $A_m(x;\xi)\otimes \mathbf{Id}$ where $\mathbf{Id}$ 
means the identity map $\Lambda^k(T^*M_x)\otimes \ml{E}_x\mapsto \Lambda^k(T^*M_x)\otimes \ml{E}_x$ and $\mathbf{A}_m^{(k)}$ is formally self--adjoint. Hence
by~\cite[Lemma 11]{FaRoSj08}, 
it has 
a self--adjoint extension on the Hilbert space $L^2(M,\Lambda^k(T^*M)\otimes \ml{E})$.
 Without loss of generality,
we may assume 
$\mathbf{A}_m^{(k)}$
to be invertible in
$\Omega^k(M,\ml{E})$ since 
by~\cite[Lemma 12]{FaRoSj08}, there is
a smoothing, self--adjoint operator $\widehat{r}$ 
such that
$\mathbf{A}_m^{(k)}+\widehat{r}$ is self--adjoint, elliptic and invertible in $\Omega^k(M,\mathcal{E})$. 
%We define 
%a linear operator
%$\mathbf{A}_m^{(k)}:\Omega^k(M,\mathcal{E}) \mapsto \Omega^k(M,\mathcal{E})$
%such that
%there is subordinated
%to some 
%$$\mathbf{A}_m^{(k)}  $$
%for every $\mathbf{s}=\sum_{i\in I,1\leqslant j\leqslant r} \chi_is_{ij}\mathbf{e}_{ij}$.  
%We define $\mathbf{A}_m^{(k)}(x;\xi):=A_m(x;\xi)\textbf{Id}$ belonging to $\text{Hom}(\Lambda^k(T^*M)\otimes\ml{E})$. 
We define an 
anisotropic Sobolev space of currents by setting
$$\boxed{\ml{H}^m_k(M,\ml{E})=(\mathbf{A}_m^{(k)})^{-1}L^2(M,\Lambda^k(T^*M)\otimes\ml{E}).}$$

%\begin{rema} Note that this requires to deal with symbols of variable orders 
%whose symbolic calculus was described in Appendix~A of~\cite{FaRoSj08}. This can be done as the symbol $m(x;\xi)$ 
%belongs to the standard class of symbols $S^0(T^*M)$. We also refer to~\cite[App.~C.1]{DyZw13} or to~\cite[Part I]{demailly1996intro} 
%for a brief reminder of pseudodifferential operators with values in vector bundles. In particular, adapting the proof 
%of~\cite[Cor.~4]{FaRoSj08} to the vector bundle valued framework, one can verify that $\mathbf{A}_m^{(k)}$ is an elliptic symbol, and 
%thus $\Op(\mathbf{A}_m^{(k)})$ can be chosen to be invertible. 
%\end{rema}

Mimicking the proofs of~\cite{FaRoSj08}, 
we can deduce some properties of these spaces of currents. 
First of all, they are endowed with a Hilbert 
structure inherited from the $L^2$-structure on $M$. 
The topological dual
$\ml{H}^{m}_k(M,\ml{E})^{\prime}$ of $\ml{H}^m_k(M,\ml{E})$ can be identified
with the space $$\ml{H}^{m}_k(M,\ml{E})^{\prime}\simeq \mathbf{A}_m^{(k)}L^2(M,\Lambda^k(T^*M)\otimes\ml{E})$$
by the pairing
$$ \psi_1\in  \mathbf{A}_m^{(k)}L^2(M,\Lambda^k(T^*M)\otimes\ml{E}),\psi_2 \in \ml{H}^m_k(M,\ml{E})\mapsto \int_M \left\langle \psi_1\wedge\star_k\psi_2 \right\rangle_\mathcal{E} $$
and $\ml{H}^{m}_k(M,\ml{E})$ is in fact reflexive. We also note that the space $\ml{H}_k^{m}(M,\ml{E})$ can be identified with $\ml{H}^m_0(M,\IR)\otimes_{\ml{C}^{\infty}(M)}\Omega^k(M,\ml{E})$, and one has
$$\Omega^k(M,\ml{E}) \subset \ml{H}_k^{m}(M,\ml{E}) \subset \mathcal{D}^{\prime,k}(M,\ml{E}),$$ 
where the injections are continuous. Using the Hodge star, we finally find that
$$\ml{H}^{m}_k(M,\ml{E})^{\prime}\simeq \star_k \mathbf{A}_m^{(k)}L^2(M,\Lambda^k(T^*M)\otimes\ml{E})=\mathbf{A}_m^{(n-k)}L^2(M,\Lambda^{n-k}(T^*M)\otimes\ml{E})$$
which is also equipped with a natural Hilbert space structure.
The Hilbert space $$\mathbf{A}_m^{(n-k)}L^2(M,\Lambda^{n-k}(T^*M)\otimes\ml{E})$$ will also be denoted by
$\mathcal{H}^{-m}_{n-k}(M,\mathcal{E})$ by some little abuse of notation.

%It is in some sense more convenient to view the dual as a subset of $\ml{D}^{\prime k}(M,\ml{E}')$ and this can be done via a few natural identifications. 
%The Hermitian structure on $\ml{E}$ allows to define a canonical isomorphism $\tau_{\ml{E}}:\ml{E}\rightarrow\ml{E}'$ by setting, for every $\psi$ in $\ml{E}$, $\tau_{\ml{E}}(\psi)=\la\psi,.\ra_{\ml{E}}$. 
%Then, combined with the Hodge star map, it induces an isomorphism from 
%$\ml{H}^{m}_k(M,\ml{E})^{\prime}$ to $\ml{H}^{-m}_{n-k}(M,\ml{E}')$, 
%whose Hilbert structure is given by the scalar product
%$$(\psi_1,\psi_2)\in\ml{H}^{-m}_{n-k}(M,\ml{E}')^2\mapsto \la \star_k^{-1}\tau_{\ml{E}}^{-1}(\psi_1),\star_k^{-1}\tau_{\ml{E}}^{-1}(\psi_2)\ra_{\ml{H}_k^m(M,\ml{E})'}.$$
%Thus, the topological dual of $\ml{H}_k^m(M,\ml{E})$ can be identified with $\ml{H}^{-m}_{n-k}(M,\ml{E}')$, where, for every $\psi_1$ in $\Omega^k(M,\ml{E})$ 
%and $\psi_2$ in $\Omega^{n-k}(M,\ml{E}')$, 
%one has the following duality relation~:
%\begin{eqnarray*}
%\la\psi_2,\psi_1\ra_{\ml{H}_{n-k}^{-m}(M,\ml{E}')\times\ml{H}_k^m(M,\ml{E})} & = & \int_{M}\psi_2\wedge\psi_1\\
%& = &\la \Op(\mathbf{A}_m^{(k)})^{-1}\star_k^{-1}\tau_{\ml{E}}^{-1}(\overline{\psi_2}),\Op(\mathbf{A}_m^{(k)})\psi_1\ra_{L^2(M,\Lambda^k(T^*M)\otimes\ml{E})}\\
%& = &\la\star_k^{-1}\tau_{\ml{E}}^{-1}(\psi_2),\psi_1\ra_{\ml{H}_k^m(M,\ml{E})\times\ml{H}_{k}^m(M,\ml{E})^{\prime}}.
%\end{eqnarray*}

\subsection{Pollicott--Ruelle resonances and their resonant states}

Now that we have defined our Sobolev in a similar fashion as for the Anosov setting of~\cite{FaSj11}, 
we can follow almost verbatim the argument of Faure and Sj\"ostrand in~\cite{FaSj11} 
in order to show the existence of a discrete dynamical spectrum on these spaces. Indeed, this part of their arguments only made use of the dynamical properties of $m$ combined with microlocal tools and analytic Fredholm theory. 
More precisely, the main result on the spectral properties of $-\ml{L}_{V,\nabla}^{(k)}$ acting on these anisotropic spaces is the following Proposition:
\\
\\
\fbox{
\begin{minipage}{0,94\textwidth} 
\begin{prop}[Discrete spectrum]\label{p:eigenvalues} The operator
$-\ml{L}_{V,\nabla}^{(k)}$ defines a maximal closed unbounded operator on $\ml{H}^m_k(M,\ml{E})$, 
$$-\ml{L}_{V,\nabla}^{(k)}:\ml{H}^m_k(M,\ml{E})\rightarrow\ml{H}^m_k(M,\ml{E}),$$
with domain given by $\ml{D}(-\ml{L}_{V,\nabla}^{(k)}):=\{ \psi\in\ml{H}^m_k(M,\ml{E}):-\ml{L}_{V,\nabla}^{(k)} \psi\in\ml{H}^m_k(M,\ml{E})\}.$ It coincides with the 
closure of $-\ml{L}_{V,\nabla}^{(k)}:\Omega^k(M,\ml{E})\rightarrow\Omega^k(M,\ml{E})$ in the graph norm for operators. Moreover, there exists a constant 
$C_0$ in $\IR$ (that depends on the choice of the order function 
$m(x;\xi)$) such that $-\ml{L}_{V,\nabla}^{(k)}$ has empty spectrum for $\operatorname{Re}(z)>C_0$. Finally,
the operator
 $$-\ml{L}_{V,\nabla}^{(k)}:\ml{H}_k^m(M,\ml{E})\rightarrow \ml{H}_k^m(M,\ml{E}),$$
 has a discrete spectrum with finite multiplicity in the domain
 $$\operatorname{Re} (z)>-C_m+C_{\ml{E}},$$
 where $C_{\ml{E}}>0$ depends only 
 on the choice of the metric $g$, $\la,\ra_{\ml{E}}$ and $C_m>0$ is the constant from Lemma~\ref{l:escape-function}.
\end{prop}
\end{minipage}
}
\\
\\
The second part on the discrete spectrum is obtained by showing that the operator $(-\ml{L}_{V,\nabla}^{(k)}-z)$ is a 
Fredholm operator of index $0$ depending analytically on $z$ in the corresponding half plane~\cite{HeSj86, Zw12}. 
In the case of Anosov flows, the proof of this result was given by Faure-Sj\"ostrand in~\cite[Sect.~3]{FaSj11} 
for $k=0$ while the extension to the case of currents was done by Dyatlov-Zworski in~\cite[Sect.~3]{DyZw13}. Note that the proofs in both 
references are of slightly different nature but they both crucially rely on the properties of the escape function used to define the 
anisotropic space $\ml{H}_k^m(M,\ml{E})$. The proof of this Proposition was given in great details in~\cite[Th.~1.4]{FaSj11} in the 
case $k=0$ and $\ml{E}=M\times\IC$. As was already mentioned, the extension to the case where $0\leq k\leq n$ and where $\ml{E}$ is 
an arbitrary vector bundle can be adapted almost verbatim except that we have to deal with pseudodifferential operators 
with symbols in $\text{Hom}(\Lambda^k(T^*M)\otimes\ml{E})$. As was already observed in~\cite{DyZw13, DaRi16}, the main point to adapt to the vector bundle 
framework is that the (pseudodifferential) operators under consideration have a \emph{scalar symbol}. In fact, given any local basis $(e_j)_{j=1,\ldots J_k}$ 
of $\Lambda^k(T^*M)\otimes\ml{E}$ and any family $(u_j)_{j=1,\ldots J_k}$ of smooth functions $\ml{C}^{\infty}(M)$, one has
$$\ml{L}_{V,\nabla}^{(k)}\left(\sum_{j=1}^{J_k}u_je_j\right)=\sum_{j=1}^{J_k}\ml{L}_{V}(u_j)e_j+\sum_{j=1}^{J_k}\ml{L}_{V,\nabla}^{(k)}(e_j)u_j,$$
where each term in the second part of the sum on the right-hand side is of order $0$ as a differential operator acting on the $u_j$. In other words, the principal symbol of 
$\ml{L}_{V,\nabla}^{(k)}$ is $\xi(V(x))\mathbf{Id}_{\Lambda^k(T^*M)_x\otimes\ml{E}_x}$. 
This scalar form allows to adapt the 
proofs of~\cite{FaSj11} to this vector bundle framework -- see~\cite{FaSj11} for a detailed proof.

For the sake of completeness, we refer to paragraph~\ref{ss:FS-proof} where we give a brief picture of the strategy developed by 
Faure and Sj\"ostrand to prove Proposition~\ref{p:eigenvalues} starting from Lemma~\ref{l:escape-function}. Note that, in this reference, the authors 
use the convention $-i\ml{L}_V$ instead of $-\ml{L}_V$.

\begin{rema}\label{r:resolvent-estimate}
 We also note that they implicitely show~\cite[Lemma~3.3]{FaSj11} that, 
for every $z$ in $\mathbb{C}$ satisfying $\text{Im} z>C_0$, one has
\begin{equation}\label{e:norm-resolvent}
 \left\|\left(\ml{L}_{V,\nabla}^{(k)}+z\right)^{-1}\right\|_{\ml{H}^m_k(M,\ml{E})\rightarrow\ml{H}_k^m(M,\ml{E})}\leq\frac{1}{\text{Re}(z)-C_0}.
\end{equation}
In particular, combining Proposition~\ref{p:eigenvalues} to the Hille-Yosida 
Theorem~\cite[Cor.~3.6, p.~76]{EnNa00}, one knows that
\begin{equation}\label{e:semigroup}(\varphi^{-t})^*:\ml{H}^m_k(M,\ml{E})\rightarrow \ml{H}_k^m(M,\ml{E}),\end{equation}
generates a strongly continuous semigroup which is defined for every $t\geq 0$ and whose norm is bounded by $e^{tC_0}.$
\end{rema}

We now list some properties of this spectrum:
\begin{itemize}
 \item As in~\cite[Th.~1.5]{FaSj11}, we can show that the eigenvalues (counted with their algebraic multiplicity) 
and the eigenspaces 
of $-\ml{L}_{V,\nabla}^{(k)}:\ml{H}_k^m(M,\ml{E})\rightarrow \ml{H}_k^m(M,\ml{E})$ are in fact independent of the choice of escape function. 
For every $0\leq k\leq n$, we call the eigenvalues the \textbf{Pollicott-Ruelle resonances of index $k$} and we denote by $\ml{R}_k(V,\nabla)$ this set.
 \item  By duality, the same spectral properties holds for the dual operator
\begin{equation}\label{e:dual}(-\ml{L}_{V,\nabla}^{(k)})^{*}=-\ml{L}_{-V,\nabla}^{(n-k)}:\ml{H}^{-m}_{n-k}(M,\ml{E})\rightarrow\ml{H}^{-m}_{n-k}(M,\ml{E}).\end{equation}
 \item Given any $z_0$ in $\ml{R}_k(V,\nabla)$, the corresponding spectral projector $\pi_{z_0}^{(k)}$ is given by~\cite[Appendix]{HeSj86}:
\begin{equation}\label{e:spectral-proj}\pi_{z_0}^{(k)}:=\frac{1}{2i\pi}\int_{\gamma_{z_0}}(z+\ml{L}_{V,\nabla}^{(k)})^{-1}dz:\ml{H}_k^m(M,\ml{E})\rightarrow \ml{H}_k^m(M,\ml{E}),\end{equation}
 where $\gamma_{z_0}$ is a small contour around $z_0$ which only contains the eigenvalue $z_0$ in its interior.
 \item Given any $z_0$ in $\IC$ with $\operatorname{Re} (z_0)>-C_m+C_{\ml{E}}$, there exists $m_k(z_0)\geq 1$ such that, in a small neighborhood of $z_0$, one has
 \begin{equation}\label{e:resolvent}
  (z+\ml{L}_{V,\nabla}^{(k)})^{-1}=\sum_{l=1}^{m_k(z_0)}(-1)^{l-1}\frac{(\ml{L}_{V,\nabla}^{(k)}+z_0)^{l-1}\pi_{z_0}^{(k)}}{(z-z_0)^l}+R_{z_0,k}(z):\ml{H}_k^m(M,\ml{E})\rightarrow \ml{H}_k^m(M,\ml{E}),
 \end{equation}
with $R_{z_0,k}(z)$ a holomorphic function near $z_0$ and with $\pi_{z_0}^{(k)}=0$ whenever $z_0\notin\ml{R}_k(V,\nabla)$.
\end{itemize}

\subsection{Proof of Theorems~\ref{t:maintheo1} and~\ref{t:maintheo2}}\label{ss:proof-theo} We will now briefly deduce the proofs of our main Theorems in the case where $\ml{E}=M\times\IC$. 
We will in fact prove something slightly stronger 
as we will verify that these statements hold for any $\psi_1\in\ml{H}_k^m(M)$ and any $\psi_2\in\ml{H}_{n-k}^{-m}(M)$. Let $\psi_1$ be an element in $\ml{H}^m_k(M)$. Thanks to Remark~\ref{r:resolvent-estimate}, 
we know that
$$\int_0^{+\infty}e^{-t z}\varphi^{-t*}(\psi_1)dt=(\ml{L}_V^{(k)}+z)^{-1}(\psi_1)$$
holds in $\ml{H}_k^m(M)$ for $\text{Re}(z)>C_0$. From Proposition~\ref{p:eigenvalues}, we have a meromorphic extension of the right-hand side on 
$\text{Re}(z)>-C_m+C$ whose poles are included in the set of Pollicott-Ruelle resonances $\ml{R}_k(V)$. Moreover, from~\eqref{e:resolvent}, 
we also know that the following holds true in $\ml{H}_k^m(M)$:
$$(z+\ml{L}_V^{(k)})^{-1}(\psi_1)=\sum_{l=1}^{m_k(z_0)}(-1)^{l-1}\frac{(\ml{L}_V^{(k)}+z_0)^{l-1}\pi_{z_0}^{(k)}(\psi_1)}{(z-z_0)^l}+R_{z_0,k}(z)(\psi_1)$$
in a neighborhood of any $z_0$ in $\ml{R}_k(V)$ which satisfies $\text{Re}(z_0)>-C_m+C$. Finally, as all 
the results hold in $\ml{H}_k^m(M)$, we can always pair these equalities with some $\psi_2$ in $\ml{H}_{n-k}^{-m}(M)$.

\subsection{A few words on Faure-Sj\"ostrand's construction}\label{ss:FS-proof}

As was already explained, once we are given an escape function satisfying the properties of Lemma~\ref{l:escape-function}, we are in 
position to apply the strategy of reference~\cite{FaSj11}. This Lemma was indeed the only Lemma from this reference that used the dynamical 
properties of the flow. After that, the authors proceeded to a detailed analytic work based on microlocal techniques and Fredholm theory. For 
the sake of completeness, let us outline the strategy of their proof and we refer to this reference for more details. Their first observation 
is that studying the spectrum on the above anisotropic Sobolev spaces is equivalent to studying the conjugated operator
$$\mathbf{A}_m^{(k)}\circ\left(-\ml{L}_{V,\nabla}^{(k)}\right)\circ(\mathbf{A}_m^{(k)})^{-1}$$
on the more standard space $L^2(M,\Lambda^k(T^*M)\otimes\ml{E})$. This is a pseudodifferential operator in 
$\Psi^1(M,\Lambda^k(T^*M)\otimes\ml{E})$, and we can write
$$\mathbf{A}_m^{(k)}\circ\left(-\ml{L}_{V,\nabla}^{(k)}\right)\circ(\mathbf{A}_m^{(k)})^{-1}\approx
\Op\left(\left(-i\xi(V(x))+X_HG_m(x;\xi)\right)\textbf{Id}_{\Lambda^k(T^*M)_x\otimes\ml{E}_x}\right),$$
where $\approx$ should be understood as an equality up to an error belonging to $\ml{O}(\Psi^0)+\ml{O}_m(\Psi^{-1+0}).$ In order to explain the 
proof of~\cite{FaSj11}, we will focus ourselves on the ``principal term'' $\left(-i\xi(V(x))+X_HG_m(x;\xi)\right)\textbf{Id}_{\Lambda^k(T^*M)_x\otimes\ml{E}_x}$ even if some attention has to 
be paid to these remainder terms. The next step is that Faure and Sj\"ostrand verify for which $z$ in $\IC$, the operator
\begin{equation}\label{e:inverse}\mathbf{A}_m^{(k)}\circ\left(-\ml{L}_{V,\nabla}^{(k)}\right)\circ(\mathbf{A}_m^{(k)})^{-1}-z\end{equation}
is invertible in $L^2$. There are three regions that we may distinguish using the 
convention of Lemma~\ref{l:escape-function}. First, when $(x;\xi)\notin N_0$ with $\|\xi\|_x\geq R$, we 
can use the fact that $X_HG_m(x;\xi)\leq-c\min(|u|,s)$ in order to invert the symbol $X_HG_m$ in this region of phase space, at least for some 
$z$ verifying $\text{Re}(z)\geq C_{\ml{E}}-c\min(|u|,s)$. On the other hand, if $(x;\xi)\in N_0$ and $\|\xi\|_x\geq R$, then $X_HG_m$ is 
not anymore uniformly stricly negative. Yet, we may subtract to the operator defined by~\eqref{e:inverse} a (selfadjoint) pseudodifferential operator 
$\hat{\chi}_0$ in $\Psi^0(M,\Lambda^k(T^*M)\otimes\ml{E})$ whose principal symbol is proportional to 
$c\min(|u|,s)\textbf{Id}_{\Lambda^k(T^*M)\otimes\ml{E}}$ near $N_0$ and identically vanishes away from it. 
It remains to deal with the compact part of phase space $\|\xi\|_x\leq R$. Here, we can substract 
a compact operator $\hat{\chi}_1$ in order to make the real part invertible in this region too. In the end, the operator
$$P(z):=\mathbf{A}_m^{(k)}\circ\left(-\ml{L}_{V,\nabla}^{(k)}\right)\circ(\mathbf{A}_m^{(k)})^{-1}-\hat{\chi}_1-\hat{\chi}_0-z$$
is invertible in $L^2$ for $\text{Re}(z)\geq C_{\ml{E}}-c\min(|u|,s)$. 
We shall denote its inverse by $r(z)$. This rough picture can be made rigorous and we refer to~\cite[p.340--345]{FaSj11} 
for a detailed proof. Then Faure and Sj\"ostrand conclude using arguments from analytic Fredholm theory. Precisely, they write that for 
$\text{Re}(z)\geq C_{\ml{E}}-c\min(|u|,s)$,
$$\mathbf{A}_m^{(k)}\circ\left(-\ml{L}_{V,\nabla}^{(k)}\right)\circ(\mathbf{A}_m^{(k)})^{-1}-z=(\text{Id}+(\hat{\chi}_0+\hat{\chi}_1)r(z))P(z).$$
Using the ellipticity of $\xi(V(x))$ near $N_0$, they show that $(\text{Id}+(\hat{\chi}_0+\hat{\chi}_1)r(z))$ is a Fredholm operator of 
index $0$~\cite[Lemma 3.4]{FaSj11} as $P(z)$ is. Hence, on $\text{Re}(z)\geq C_{\ml{E}}-c\min(|u|,s)$, we have a holomorphic family of Fredholm operators of index $0$ which are invertible for $\text{Re}(z)$ 
large enough. Then, the conclusion follows from classical theorems of analytic Fredholm theory -- see e.g.~\cite[Th.~D.4]{Zw12}.

\appendix

\section{Hyperbolic critical elements}\label{a:hyperbolic}

In the definition of Morse-Smale flows, we implicitely assumed some results on hyperbolic fixed points and hyperbolic 
closed orbits that we will briefly review in this appendix. For more details, 
we invite the reader to look at the classical textbook of Palis and de Melo~\cite[Ch.~2,3]{PaDeMe82} -- 
see also~\cite{HiPuSh77} for general results on (partially) hyperbolic invariant subsets.

\subsection{Limit sets}

We start with some terminology from the theory of dynamical systems. We say that a point $x$ in $M$ is wandering if there exist some open neighborhood $U$ of $x$ and some $t_0>0$ such that
$$U\cap\left(\cup_{|t|>t_0}\varphi^t(U)\right)=\emptyset.$$
The \emph{nonwandering set} of the flow is given by the points which are not wandering. The set of nonwandering points is denoted by $\operatorname{NW}(\varphi^t)$. Given any 
$x\in M$, we define 
$$\alpha(x):=\cap_{T\leq 0}\overline{\{\varphi^t(x):t\leq T\}},$$
and
$$\omega(x):=\cap_{T\geq 0}\overline{\{\varphi^t(x):t\geq T\}}.$$
We note that for every $x$ in $M$, $\alpha(x)$ and $\omega(x)$ are contained in $\operatorname{NW}(\varphi^t)$. For any invariant closed subset $\Lambda$ 
of $M$, we define the \emph{unstable and stable manifolds} of $\Lambda$:
$$W^u(\Lambda):=\{x\in M:\alpha(x)\subset\Lambda\},$$
and 
$$W^s(\Lambda):=\{x\in M:\omega(x)\subset\Lambda\}.$$

\subsection{Hyperbolic fixed points} We say that a point $x_0$ in $M$ is a 
\emph{hyperbolic fixed point} of $(\varphi^t)_{t\in\IR}$ if 
$V(x_0)=0$ and $d_{x_0}V:T_{x_0}M\rightarrow T_{x_0}M$ has no eigenvalue on the imaginary axis. Equivalently, 
it means that for every $t\neq 0$, $x_0$ is a fixed point of the smooth diffeomorphism $\varphi^{t}$ and that $d_{x_0}\varphi^{t}:T_{x_0}M\rightarrow T_{x_0}M$ has no eigenvalue of modulus one.

Consider now a hyperbolic fixed point $x_0$ and some small enough $\delta>0$. We can define the local unstable and stable manifolds as follows:
$$W^u_{\delta}(x_0):=\{x\in B(x_0,\delta):\ \forall t\leq 0,\ \varphi^t(x)\in B(x_0,\delta)\},$$
and
$$W^s_{\delta}(x_0):=\{x\in B(x_0,\delta):\ \forall t\geq 0,\ \varphi^t(x)\in B(x_0,\delta)\}.$$
Then, one has~\cite[Ch.~2, Prop.~6.1 and Th.~6.2]{PaDeMe82}
$$W^u(x_0)=\bigcup_{t\geq 0}\varphi^t(W^u_{\delta}(x_0)),\ \text{and}\ W^s(x_0)=\bigcup_{t\geq 0}\varphi^{-t}(W^s_{\delta}(x_0)).$$
Moreover, $W^u_{\delta}(x_0)$ (resp. $W^{s}_{\delta}(x_0)$) is a smooth embedded disk whose dimension is that of the unstable (resp. stable) space of $d_{x_0}V$ while $W^u(x_0)$ (resp. $W^s(x_0)$) is a smooth injectively immersed manifold in $M$ whose tangent space at $x_0$ is the unstable (resp. stable) space of $d_{x_0}V:T_{x_0}M\rightarrow T_{x_0}M$.

\subsection{Hyperbolic closed orbits} We say that a 
point $x_0$ 
in $M$ is a hyperbolic periodic point if $V(x_0)\neq 0$, 
there exists $T_0>0$ such that $\varphi^{T_0}(x_0)=x_0$ and 
$d_{x_0}\varphi^{T_0}:T_{x_0}M\rightarrow T_{x_0}M$ has $1$ 
as a simple eigenvalue and no other eigenvalue of modulus $1$. 
Equivalently, we will say that $\{\varphi^t(x_0):0\leq t\leq T_0\}$ is 
a hyperbolic closed orbit.

This can also be defined in terms of Poincar\'e sections 
which allows to make the connection with the case of hyperbolic fixed 
points. Let $\Sigma$ be a smooth hypersurface 
containing $x_0$ which is transversal to the vector field $V$. We denote by 
$P_{\Sigma}:O\subset\Sigma\rightarrow\Sigma$ the 
corresponding Poincar\'e map. Then, the point $x_0$ 
is said to be a hyperbolic periodic point of the 
flow if $x_0$ is a hyperbolic point for 
the Poincar\'e map. Note that this definition 
does not depend on the choice of  
Poincar\'e 
section $\Sigma$. Fix now a neighborhood $\tilde{O}$ of the closed orbit $\Lambda$ generated by the point 
$x_0$. We can define the (local) unstable and stable manifolds:
$$W^u_{\tilde{O}}(x_0):=\{x\in\tilde{O}:\ \forall t\leq 0,\ \varphi^{t}(x)\in\tilde{O}\},$$
and
$$W^s_{\tilde{O}}(x_0):=\{x\in\tilde{O}:\ \forall t\geq 0,\ \varphi^{t}(x)\in\tilde{O}\}.$$
As in the case of fixed points, the following holds:
$$W^u(\Lambda)=\bigcup_{t\geq 0}\varphi^{t}(W^u_{\tilde{O}}(x_0)),\ \text{and}\ W^s(\Lambda)=\bigcup_{t\leq 0}\varphi^{t}(W^s_{\tilde{O}}(x_0)).$$
Moreover, $W^u_{\tilde{O}}(x_0)$ and $W^s_{\tilde{O}}(x_0)$ are smooth submanifolds of $M$ which 
are transverse~\cite[Ch.~3, Prop.~1.5]{PaDeMe82} and $W^u_{\tilde{O}}(x_0)\cap W^s_{\tilde{O}}(x_0)=\Lambda$. 
In fact, for a given Poincar\'e section $\Sigma$, one can define $W^u_{\delta}(x_0)$ and $W^s_{\delta}(x_0)$ for the induced 
Poincar\'e map $P_{\Sigma}$ and $W^u_{\tilde{O}}(x_0)$ (resp. $W^s_{\tilde{O}}(x_0)$) is an open neighborhood of $\Lambda$ 
inside $\cup_{t\in(0,2T_0)}\varphi^{-t}(W^u_{\delta}(x_0))$ (resp. $\cup_{t\in(0,2T_0)}\varphi^{t}(W^u_{\delta}(x_0))$). 
Also, $W^u(\Lambda)$ and $W^s(\Lambda)$ are smooth immersed submanifolds of $M$~\cite[Ch.~3, Coro.~1.6]{PaDeMe82}. 

Finally, $W^u(\Lambda)$ and $W^s(\Lambda)$ are invariantly fibered by smooth submanifolds 
$(W^{uu}(x_0))_{x_0\in\Lambda}$ (resp. $(W^{ss}(x_0))_{x_0\in\Lambda}$) tangent to the unstable (resp. stable) space at $\varphi^{T_0}(x_0)$~\cite[Th.~4.1]{HiPuSh77}. Points of these submanifolds are characterized as follows, for every $x_0$ in $\Lambda$,
$$W^{uu}(x_0):=\left\{x\in M:\ \lim_{n\rightarrow+\infty}\varphi^{-nT_0}(x)=x_0\right\},$$
and
$$W^{ss}(x_0):=\left\{x\in M:\ \lim_{n\rightarrow+\infty}\varphi^{nT_0}(x)=x_0\right\}.$$

\subsection{$\lambda$-Lemma}\label{aa:invariant-cone} In the previous paragraphs, we saw that understanding the dynamics 
near a critical element is related to understanding the dynamics of a diffeomorphism $f$ near a hyperbolic point $\Lambda:=\{x_0\}$. This reduction can be 
done either by considering the time one map of the flow in the case of fixed points, or by looking at the Poincar\'e map associated with a certain transversal to the orbit. 
In this paragraph, we would like to give some quantitative features of the dynamics for a diffeomorphism $f:M\rightarrow M$ near a hyperbolic point $x_0$. 
In particular, we would like to recall the $\lambda$-Lemma~\cite{Sm60, Pa68} (sometimes called the inclination lemma). We follow closely the 
presentation of~\cite[Ch.~2]{PaDeMe82} -- see also~\cite[Ch.~5]{BrSt02} in the case of more general hyperbolic subsets. Before that, recall that two 
smooth subamnifolds $S$ and $S'$ of $M$ are $\epsilon$-$\ml{C}^1$ close if there exists a $\ml{C}^1$ diffeomorphism 
$h:S\rightarrow S'$ such that $i'\circ h$ is $\epsilon$-close to $i$ in the $\ml{C}^1$ topology\footnote{Here $i:S\rightarrow M$ and $i':S'\rightarrow M$ 
denote the inclusion maps.}.

We fix some local coordinates $(x,y)\in \IR^{n'}=E_s\oplus E_u$ around the hyperbolic point $x_0$ (with $n'=n$ or $n-1$). 
The local stable (resp. unstable) manifold is then the graph of a smooth function $\kappa_s:B_s(0,r_1)\rightarrow E_u$ (resp. $\kappa_u:B_u(0,r_1)\rightarrow E_s$) 
where $B_s(0,r_1)$ (resp. $B_u(0,r_1)$) is the ball of radius $r_1$ centered at $0$ inside $E_s$ (resp. $E_u$). 
Moreover, we have that $\kappa_u(0)=\kappa_s(0)=0$ and $d_0\kappa_s=d_0\kappa_u=0$. We then introduce the following change of coordinates:
$$\kappa: B_s(0,r_1)\oplus B_u(0,r_1)\rightarrow E_s\oplus E_u,\ (x,y)\mapsto(x-\kappa_u(y),y-\kappa_s(x)).$$
From our construction, one can verify that for $r_1>0$ small enough, $\kappa$ induces a diffeomorphism near the origin. Hence, 
one can assume without loss of generality that we are working in a local chart where the local stable (resp. unstable) manifold 
is represented by the stable (resp. unstable) linear space. The $\lambda$-Lemma can then be formulated as follows~\cite[Ch.~2, Lemma~7.1]{PaDeMe82} (see also~\cite[Th.~5.7.2]{BrSt02}):

\begin{theo}[$\lambda$-Lemma]\label{t:palis} We use the above conventions. Let $O=B_s(x_0,r)\times B_u(x_0,r)$, let $x$ be an element in 
$W^s(x_0)$ and let $D^u$ be a small disk of dimension $\operatorname{dim} (E^u(x_0))$ which is transversal to $W^s(x_0)$ at $x$.

If we denote by $D^u_N$ the connected component of $f^N(D^u)\cap O$ containing $f^N(x)$, 
then for every $\eps>0$, 
there exists $N_0$ such that, for every $N\geq N_0$, $D^u_N$ is $\eps$ $\ml{C}^{1}$-close to $W^u(x_0)$.
 
\end{theo}

 \subsection{Sternberg-Chen's Theorem}\label{aa:Sternberg} In this paragraph, we collect a few results on the linearization of 
 vector fields near hyperbolic critical elements in order to illustrate that 
 our assumption of being $\ml{C}^1$ linearizable is in some sense generic.

 \subsubsection{The case of a fixed point} Recall Sternberg-Chen's Theorem on the linearization of vector fields near hyperbolic critical points~\cite{Ch63} (see also~\cite[Th.~9, p.50]{Ne69}):
\begin{theo}[Sternberg-Chen]\label{t:Sternberg-Chen} Let $V(x)=\sum_ja_j(x)\partial_{x_j}$ be a smooth vector field defined in a neighborhood of $0$ in 
$\IR^n$. Suppose that $V(0)=0$ and that $0$ is a hyperbolic 
fixed point. Denote by $(\mu_j)$ the eigenvalues of $A:=(\partial_{x_k}a_j(0))_{k,j}$. Suppose that the 
eigenvalues satisfy the \textbf{non resonant assumption},
$$\forall\ \alpha_1,\ldots,\alpha_n\in\IN\ \text{s.t}\ \alpha_1+\ldots +\alpha_n\geq 2,\ \forall\ 1\leq j\leq n,\ \mu_j\neq \sum_{i=1}^d\alpha_i\mu_i.$$
Then, there exists a smooth diffeomorphism $h$ which is defined in a neighborhood of $0$ such that $h(0)=0$ and such that
$$h^*(V)(x)=D_xh(L(x))$$
where $L(x)=Ax.\partial_x$.
\end{theo}
Here $\IN$ denotes the set of \emph{nonnegative} integers. The classical Grobman-Hartman Theorem ensures the existence of a conjugating homeomorphism. The crucial observation for us 
is that the conjugating map is smooth provided some non resonance assumption is made. 
 
\begin{rema}
 We stated here a version of the Theorem which gives conditions to have a smooth diffeomorphism. Yet, if we only search a $\ml{C}^k$ conjugating 
 diffeomorphism (with $k\geq 1$), we can restrict our assumptions by imposing only a \emph{finite} number of conditions on the 
 eigenvalues -- see~\cite[Th.~10, p.~52]{Ne69} for the precise statement.
\end{rema}

 \subsubsection{The case of a closed orbit}
 Sternberg-Chen's Theorem can be generalized in the case of a closed hyperbolic orbit~\cite{WWL08}. Consider a 
 smooth vector field $V(x,\theta)$ defined on $B_{n-1}(0,r)\times(\IR/\ml{P}_{\Lambda}\IZ)$ where $B_{n-1}(0,r)$ is a small ball of radius $r>0$ centered at 
 $0$ in $\IR^{n-1}$. We make the assumption that
 \begin{equation}\label{e:linearized-vector-field}V(x,\theta)=(1+g( x,\theta))\partial_{\theta}+(A(\theta)x+f(x,\theta)).\partial_x,\end{equation}
 with $f(x,\theta)=\ml{O}(\|x\|^2)$, $A(\theta)$ smooth and $g(x,\theta)=\ml{O}(\|x\|)$.

\textbf{Simplifying coordinates near the periodic orbit.} 
 Let us first explain how the vector field can be put under the form~\eqref{e:linearized-vector-field}. Recall that the 
tubular neighborhood Theorem states that some neighborhood of $\Lambda\subset M$ is $C^\infty$ diffeomorphic to some neighborhood of the 
zero section of the normal bundle $N(\Lambda\subset M)$ induced by some Riemannian metric on $M$. But $\Lambda$ is a circle (with global coordinates
$\theta$ in $\IS^1:=\IR/(2\pi\IZ)$) hence $N(\Lambda\subset M)$ is an oriented real vector bundle over $\mathbb{S}^1$. Hence, $N(\Lambda)$ 
is trivial and diffeomorphic to a cartesian product $\mathbb{S}^1\times \mathbb{R}^{d-1}$. Equivalently, triviality and the tubular neighborhood 
Theorem guarantee that we have some germ of coordinate system $(x,\theta)$ near $\Lambda$ where $\Lambda$ is defined by the global equation
$\{x=0\}$. In these coordinates, the vector field $V$ associated with $\varphi^t$ can be written
\begin{equation}
V(x,\theta)=\tilde{g}(x,\theta)\partial_{\theta}+\tilde{f}(x,\theta)\partial_x,
\end{equation}
where $(x,\theta)$ belongs to $O_{\Lambda}\times\IS^1$ for some small neighborhood $O_{\Lambda}\subset\IR^{n-1}$ of $0$. 
As $\Lambda$ is a closed orbit, we can suppose that $\tilde{f}(0,\theta)=0$ for all $\theta\in\mathbb{S}^1$ and $\tilde{g}(0,\theta)>0$ for all $\theta\in \IS^1$. Before explaining the analogue of the Sternberg-Chen's Theorem for 
closed hyperbolic orbit, let us first reparametrize the $\theta$ variable. Set
$\tilde{\theta}(\theta)=\int_0^\theta \tilde{g}^{-1}(0,s) ds$ and $\ml{P}_{\Lambda}:=\int_0^{2\pi} \tilde{g}^{-1}(0,s) ds$. In these new coordinates, one has
$$d\tilde{\theta}=d \int_0^\theta \tilde{g}^{-1}(0,s) ds=\tilde{g}^{-1}(0,\theta)d \theta \implies d\tilde{\theta}(\tilde{g}(0,\theta)\partial_\theta)=1.$$
The fact that this is a diffeomorphism relies on the fact that $\tilde{g}(0,s)>0$ for all $s\in[0,2\pi]$. Therefore up to doing this reparametrization, we may assume that
\begin{equation}
V(x,\theta)=\tilde{f}(x,\theta)\partial_x+\tilde{g}(x,\theta)\partial_\theta
\end{equation}
with $\theta\in \IR/(\ml{P}_{\Lambda}\IZ)$, 
$$\tilde{f}(x,\theta)=A(\theta)x+f(x,\theta),\ \text{and}\ \tilde{g}(x,\theta)=1+g(x,\theta),$$
Moreover, $f(x,\theta)=\ml{O}(\|x\|^2)$ and $g(x,\theta)=\ml{O}(\|x\|)$ uniformly in a small neighborhood of $\{x=0\}$.

Following~\cite{WWL08} and using the conventions of paragraph~\ref{ss:floquet}, we introduce the following nonresonance conditions:
\begin{itemize}
 \item $(\mu_j(\Lambda))_{j=1,\ldots,n-1}$ are \textbf{nonresonant in space} if, for every $\alpha\in \IN^{n-1}$ with $\sum_{k=1}^{n-1}\alpha_k\geq 2$ 
 and for every $1\leq j\leq n-1$,
 $$\mu_j(\Lambda)-\sum_{k=1}^{n-1}\alpha_k\mu_k(\Lambda)\notin \frac{2i\pi}{\ml{P}_{\Lambda}}\IZ,$$
 \item $(\mu_j(\Lambda))_{j=1,\ldots,n-1}$ are \textbf{nonresonant in time} if, for every $\alpha\in \IN^{n-1}$ with $\sum_{k=1}^{n-1}\alpha_k\geq 1$,
 $$\sum_{k=1}^{n-1}\alpha_k\mu_k(\Lambda)\notin \frac{2i\pi}{\ml{P}_{\Lambda}}\IZ.$$
\end{itemize}
We can now state the analogue of Sternberg-Chen's Theorem in the case of a hyperbolic closed orbit~\cite[Th.~3]{WWL08}:
 \begin{theo}\label{t:sternberg-closed-orbit} Let $V(x,\theta)$ be a smooth vector field of the form~\eqref{e:linearized-vector-field}. 
 Suppose that $\{0\}\times\IR/(\ml{P}_{\Lambda}\IZ)$ is a hyperbolic 
 closed orbit for the flow generated by $V$ and that $(\mu_j(\Lambda))_{j=1,\ldots,n-1}$ are both nonresonant in time and space. 
 Then, there exists a smooth diffeomorphism $h$ which is defined in a neighborhood of
   $\{0\}\times(\IR/\ml{P}_{\Lambda}\IZ)$ such that
   $$Vh^*(V)(x,\theta)=D_xh(L(x,\theta)),$$
with 
$$L(x,\theta)=\partial_{\theta}+A(\theta)x.\partial_x.$$
\end{theo}

\begin{rema} Fix $k\geq 1$. As in~\cite{Ne69}, the proof of~\cite{WWL08} could be adapted to ensure that
under a \emph{finite} number of non resonant conditions in space and times, $h$ can be chosen to be $\ml{C}^k$. However, we are not 
aware of a place in the literature where this condition is explicitely written.
 
\end{rema}

\section{Proof of the dynamical statements from Section~\ref{s:smale}}\label{a:proof-smale} 

In this appendix, for the sake of completeness, we briefly review the instructive proofs of some dynamical results due to Smale~\cite{Sm60}.

\subsection{Proof of Lemma~\ref{l:index-increase}}

Let $x$ be an element in $W^u(\Lambda_i)\cap W^s(\Lambda_j)$ which does not belong to $\Lambda_i$ (otherwise $i=j$ and the conclusion is trivial).
Then, the flow line $t\mapsto \varphi^t(x)$ is contained in the intersection $W^u(\Lambda_i)\cap W^s(\Lambda_j)$ which should have 
dimension at least $1$. Also, one knows
$$\text{dim}(T_xW^u(\Lambda_i))+\text{dim}(T_xW^s(\Lambda_j))=\text{dim}(T_xW^u(\Lambda_i)+T_xW^s(\Lambda_j))+\text{dim}(T_xW^u(\Lambda_i)\cap T_xW^s(\Lambda_j)).$$
From the transversality assumption $T_xW^u(\Lambda_i)+T_xW^s(\Lambda_j)=T_xM$ which, combined to the previous observation, implies
$$\text{dim}(W^u(\Lambda_i))+\text{dim}(W^s(\Lambda_j))\geq n+ 1.$$
We now distinguish two cases. On the one hand, if $\Lambda_j$ is a critical point, 
then $\text{dim}(W^s(\Lambda_j))=n-\text{dim}(W^u(\Lambda_j))$ (by transversality at $\Lambda_j$). 
Hence, $\text{dim}(W^u(\Lambda_i))\geq 
1+ \text{dim}(W^u(\Lambda_j))$ as expected. On the other hand, if $\Lambda_j$ is a closed orbit, one has, by 
transversality at $\Lambda_j$, $\text{dim}(W^s(\Lambda_j))=n-\text{dim}(W^u(\Lambda_j))+1$ 
from which one can conclude 
that $\operatorname{dim}(W^u(\Lambda_i))\geq \operatorname{dim}(W^u(\Lambda_j))$.

\subsection{Proof of Lemma~\ref{l:nocycle}}

From Lemma~\ref{l:index-increase}, we note that, 
if $W^u(\Lambda_j)\cap W^s(\Lambda_j)\neq\Lambda_j$, then $\Lambda_j$ 
is necessarily a closed orbit. In fact, it means that there exists $x$ that belongs to $W^u(\Lambda_j)\cap W^s(\Lambda_j)$ but not to $\Lambda_j$. From the case of equality in Lemma~\ref{l:index-increase}, it follows that $\Lambda_j$ is a closed orbit.

Suppose now that there exists such a point $x_0$, i.e. $x_0$ belongs to $W^u(\Lambda_j)\cap W^s(\Lambda_j)$ but not to $\Lambda_j$. Note that $x_0$ does not belong to $\Lambda_i$ for every $i\neq j$ -- see Lemma~\ref{l:partition}. Let $U$ be a small open set containing $x_0$. We would like to prove that, 
for every $t_0>0$, there exists $t\geq t_0$ such that $\varphi^{t}(U)\cap U\neq\emptyset$, which would contradict the fact that $x_0\notin NW(\varphi^t)$.

For that purpose, fix $\Sigma$ a small Poincar\'e section associated with the closed orbit $\Lambda_j$ and centered at the point $y_0\in\Lambda_j$ -- see appendix~\ref{a:hyperbolic}. 
 We also set $D^u$ to be a small open disk 
containing $x_0$ inside $W^u(\Lambda_j)\cap U$ and $D^s$ to be a small open disk containing $x_0$ 
inside $W^s(\Lambda_j)\cap U$. Fix now $t_0>0$ and let us show the expected contradicition. As $x_0$ belongs to $W^s(\Lambda_j)$, 
we know that there exists $t\geq t_0$ such that $\varphi^t(x_0)$ belongs to $\Sigma$. We denote by $D^u(t,\Sigma)$ the 
connected component of $\varphi^t(D^u)\cap\Sigma$ containing $\varphi^t(x_0)$. From the $\lambda$-Lemma~\ref{t:palis}, 
we know that, for $t$ large enough, $D^u(t,\Sigma)$ is $\eps$ $\ml{C}^1$-close to the unstable manifold of the Poincar\'e map 
near $y_0$ (for some small $\eps$). Similarly, we can work in negative times and construct $D^s(-t',\Sigma)$ which 
is $\eps$ $\ml{C}^1$-close to the stable manifold of the 
Poincar\'e map near $y_0$. As $\text{dim} W^u(\Lambda_j)+\text{dim} W^s(\Lambda_j)=  n+1$, 
$D^u(t,\Sigma)$ (resp. $D^s(-t',\Sigma)$) have the same dimension as the unstable (resp. stable) manifolds of the induced Poincar\'e map on $\Sigma$. 
Therefore, there exists a point $y_1$ which lies in the intersection of $D^u(t,\Sigma)$ and $D^s(-t',\Sigma)$. Then, we set $y_2=\varphi^{-t}(y_1)=\varphi^{-(t+t')}(y_3)$ which belongs to $U\cap \varphi^{-(t+t')}(U)$. This gives the contradiction

\subsection{Proof of Theorem~\ref{t:smale}}

The proof of Theorem~\ref{t:smale} was given by Smale in~\cite{Sm60}. The proof of this classical result contains 
important ideas that enlightens our general strategy. Hence, it seems useful to recall Smale's argument.
The proof starts with the following Lemma~\cite[Lemma~3.3]{Sm60}:
\begin{lemm}\label{l:boundary} Suppose $W^u(\Lambda_i)\cap W^s(\Lambda_j)\neq\emptyset$. Then, $W^u(\Lambda_j)\subset\overline{W^u(\Lambda_i)}$.
 
\end{lemm}
This Lemma means the following. If $\varphi^t(x)$ is repelled by $\Lambda_i$ and attracted by $\Lambda_j$, then
the unstable manifold $W^u(\Lambda_j)$ of the set $\Lambda_j$ at the arrival is contained in the closure of the bigger 
stratum $W^u(\Lambda_i)$.

\begin{proof} Let $x_0$ be a point in $W^u(\Lambda_j)$ and let $\eps>0$. We aim at constructing a point $y\in W^u(\Lambda_i)$ which is at a distance $\leq\eps$ 
of $x_0$. 

We start with the case where $\Lambda_j$ is a hyperbolic fixed point. 
We first note that there exists $T_0>0$ such that $\varphi^{-T_0}(x_0)$ belongs to an $\eps$-neighborhood $O$ of $\Lambda_j$ inside 
$W^u(\Lambda_j)$. We also fix a small neighborhood $O'$ of $\Lambda_j$ inside $M$ (containing $x_0$). In order to construct $y$, we will make use of a 
point lying inside $W^u(\Lambda_i)\cap W^s(\Lambda_j)$. Such a point $y_0$ exists from our assumption. Let $D^u(y_0)$ 
be a small disk contained in $W^u(\Lambda_i)$ which is of dimension $\text{dim}(W^u(\Lambda_j))$, which contains $y_0$ in its interior 
and which is transversal to $W^s(\Lambda_j)$. The existence of such a disk is provided by the dimension bound 
$\text{dim}(W^u(\Lambda_j))\leqslant \text{dim}(W^u(\Lambda_i))$ from Lemma~\ref{l:index-increase}. Recall that, for a fixed point, 
$\dim(W^u(\Lambda_j))+\dim(W^s(\Lambda_j))=n$.

Using the map $f=\varphi^1$ in the $\lambda$-Lemma, we can find $m$ large enough such 
that the connected component of $\varphi^m(D^u(y_0))\cap O'$ containing $\varphi^m(y_0)$ is $\eps e^{-C T_0}$ $\ml{C}^1$-close to $W^u(\Lambda_j)$ near 
$\Lambda_j$ (for some $C>0$ larger than the maximal expansion rate of $\varphi^t$).  In particular, we 
can find a point $y_1\in W^u(\Lambda_i)$ which is at a distance $\eps e^{-C T_0}$ of $\varphi^{-T_0}(x_0)$ -- see Figure~\ref{f:Smaletheo1}. 
\begin{figure}[ht]\label{f:Smaletheo1}
\includegraphics[width=12cm, height=6cm]{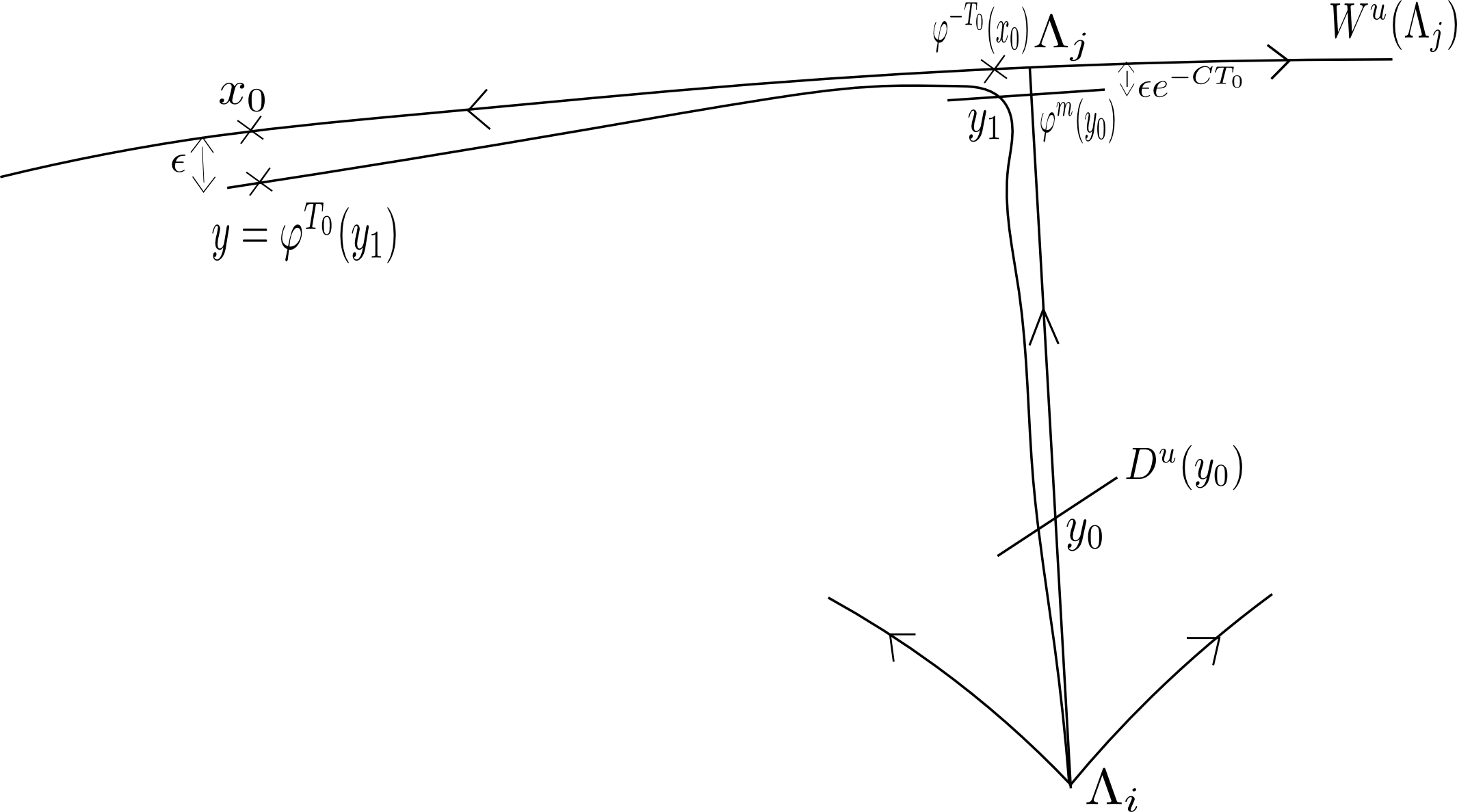}
\centering
\caption{Proof of Lemma~\ref{l:boundary}.}
\end{figure}
Recall that we said that the maximal expansion rate of 
the flow $\varphi^t$ is less than $C$. Hence, we have constructed a point $y=\varphi^{T_0}(y_1)\in W^u(\Lambda_i)$ which is 
$\eps$-close to $x_0$ as expected.

In the case where $\Lambda_j$ is a closed orbit, we shall fix $\Sigma$ to be a Poincar\'e section transversal to $\Lambda_j$. For every $T_0>0$ large enough, there will 
exist $T>T_0$ such that $\varphi^{-T}(x_0)$ belongs to $\Sigma$. As in the proof of Lemma~\ref{l:nocycle}, we have to use the $\lambda$-lemma 
for the induced Poincar\'e map and we deduce the results following the same lines as for a fixed point.
\end{proof}

We continue with the following Lemma~\cite[Lemma~3.5]{Sm60}:

\begin{lemm}\label{l:partialorder} Suppose $W^u(\Lambda_{i_1})\cap W^s(\Lambda_{i_2})\neq\emptyset$ and $W^u(\Lambda_{i_2})\cap W^s(\Lambda_{i_3})\neq\emptyset$. Then, one has
 $$W^u(\Lambda_{i_1})\cap W^s(\Lambda_{i_3})\neq\emptyset.$$
 Moreover, if $i_1=i_3$, then $i_1=i_2=i_3$. In particular, from Lemma~\ref{l:nocycle}, all these intersections are reduced to $\Lambda_{i_1}$.
\end{lemm}

\begin{proof} The argument looks very much like the proof of the no-cycle Lemma~\ref{l:nocycle}. Let $x_2$ be an element in $W^u(\Lambda_{i_2})\cap W^s(\Lambda_{i_3})$. Let 
$\tilde{D}^s$ be a small disk inside $W^s(\Lambda_{i_3})$ which is of dimension $\text{dim}(W^s(\Lambda_{i_2}))$ and which contains $x_2$. Again, this is possible according to Lemma~\ref{l:index-increase}. 
By applying $\varphi^{-t}$ with $t>0$ large enough, we have two options which follows from the $\lambda$-Lemma:
\begin{itemize}
 \item $\Lambda_{i_2}$ is a hyperbolic fixed point -- see figure~\ref{f:Smaletheo2}. If we set $D^s(-t)$ to be connected component of $\varphi^{-t}(\tilde{D}^s)\cap O$ (where $O$ 
 is a neighborhood of $\Lambda_{i_2}$) 
 containing $\varphi^{-t}(x_2)$ in its interior, then $D^s(-t)$ is $\eps$ $\ml{C}^1$-close to the local stable manifold of $\Lambda_{i_2}$. Note that any element in 
  $D^s(-t)$ belongs to $W^s(\Lambda_{i_3})$ and that this small piece of disk has dimension $\text{dim}(W^s(\Lambda_{i_2}))$.
 \begin{figure}[ht]\label{f:Smaletheo2}
\includegraphics[width=12cm, height=6cm]{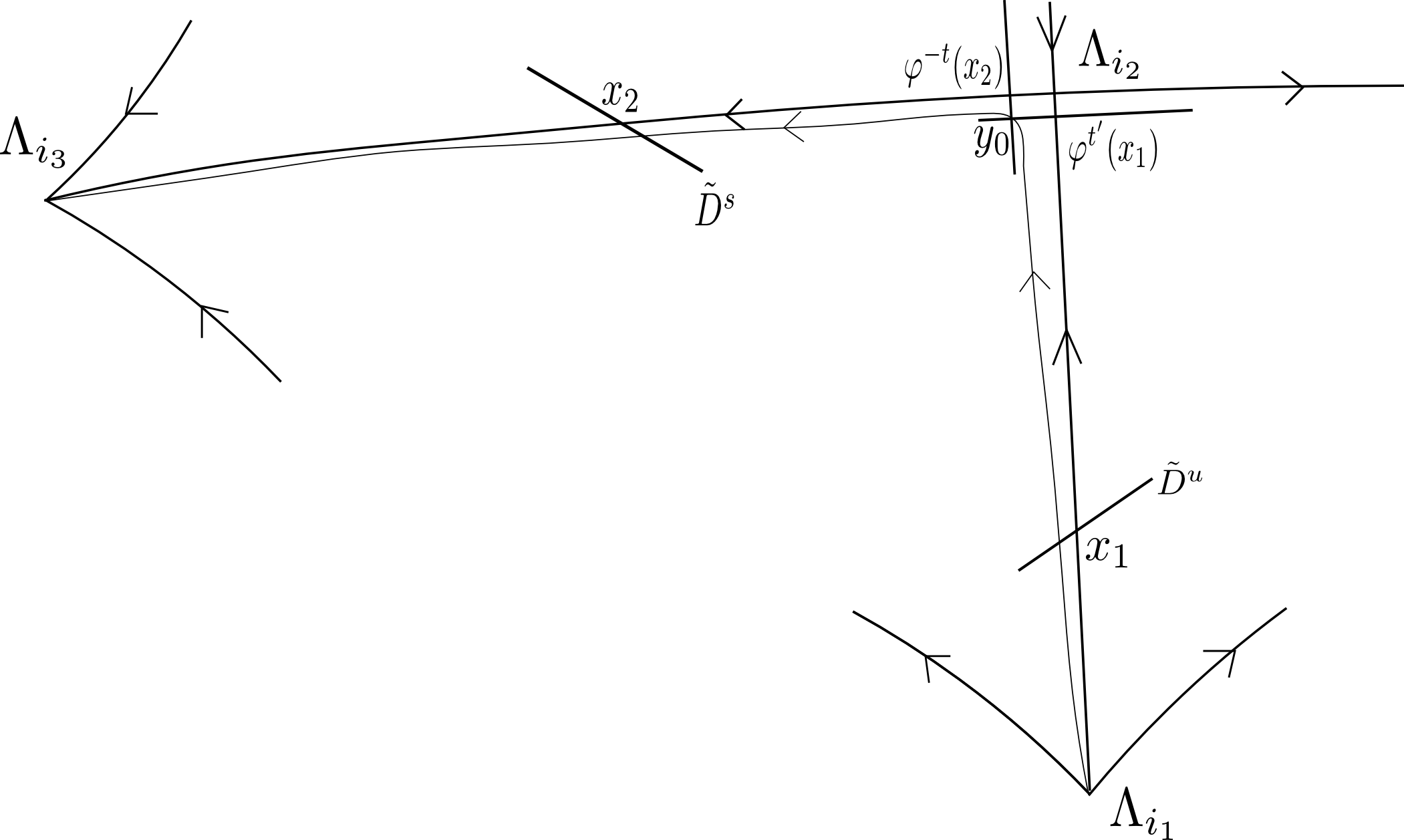}
\centering
\caption{Proof of Lemma~\ref{l:partialorder}}
\end{figure}
 \item $\Lambda_{i_2}$ is a hyperbolic closed orbit. As in the proof of Lemma~\ref{l:nocycle}, we fix a Poincar\'e section $\Sigma$ and we define a small disk $D^s(-t,\Sigma)$ inside $W^s(\Lambda_{i_3})\cap\Sigma$ 
 which is $\eps$ close to the stable manifold of the induced Poincar\'e map on $\Sigma$. Again, any point in $D^s(-t,\Sigma)$ is contained in the stable manifold of $\Lambda_{i_3}$. Note that this small disk has dimension $\text{dim}(W^s(\Lambda_{i_2}))-1$.
\end{itemize}
Then, we fix $x_1$ be an element in $W^u(\Lambda_{i_1})\cap W^s(\Lambda_{i_2})$. Working with positive times, we can similarly define either $D^u(t')$ (which is of dimension $\text{dim}(W^u(\Lambda_{i_2}))$) or $D^u(t',\Sigma)$ (which is of dimension $\text{dim}(W^u(\Lambda_{i_2}))-1$) inside 
$W^u(\Lambda_{i_1})$. Using the relations on the dimension as in the proof of Lemma~\ref{l:nocycle}, we know 
that these small pieces of disks have at least one point of intersection $y_0$. Hence, the point $y_0$ belongs to $W^u(\Lambda_{i_1})\cap W^s(\Lambda_{i_3})$ which 
concludes the proof of the first part.

Now, if $i_1=i_3$, then we have constructed a point $y_0$ in a neighborhood of $\Lambda_{i_2}$ which is contained in 
$W^u(\Lambda_{i_1})\cap W^s(\Lambda_{i_1})$. Then, the no-cycle Lemma implies finally $i_2=i_1$ as expected.
\end{proof}

The combination of Lemmas~\ref{l:boundary},~\ref{l:partialorder} and~\ref{l:boundary2} gives the proof of Smale's Theorem~\ref{t:smale} as we shall now explain. First, we note that, according to Lemma~\ref{l:partition}, 
one can find, for every $x\in M$, a unique $1\leq j(x)\leq K$ such that $x\in W^u(\Lambda_{j(x)})$. To prove the first part of Smale's Theorem, we can verify that, for every $1\leq i\leq K$,
$$\overline{W^u(\Lambda_i)}=\bigcup_{x\in \overline{W^u(\Lambda_i)}}W^u(\Lambda_{j(x)}).$$
Note that that the first inclusion $\subset$ is obvious by definition while the other one follows from Lemma~\ref{l:boundary2}. It now remains to show the partial order 
relation on the unstable manifolds inside $M$. We start with 
the transitivity property. Suppose that $W^u(\Lambda_{i_1})\preceq W^u(\Lambda_{i_2})$ and that $W^u(\Lambda_{i_2}) \preceq W^u(\Lambda_{i_3})$. This exactly means that $W^u(\Lambda_{i_1})\subset \overline{W^u(\Lambda_{i_2})}$ 
and $W^u(\Lambda_{i_2})\subset \overline{W^u(\Lambda_{i_3})}.$ Applying Lemma~\ref{l:boundary2} to $(i_1,i_2)$ and to $(i_2,i_3)$, we find a sequence $i_3=j_1,\ldots, j_m=i_1$ such that, for every $1\leq p\leq m-1$, 
$W^u(\Lambda_{j_{p}})\cap W^s(\Lambda_{j_{p+1}})\neq \emptyset$. 
Then, a combination of Lemmas~\ref{l:boundary} and~\ref{l:partialorder} shows that $W^u(\Lambda_{i_1})\subset \overline{W^u(\Lambda_{i_3})}$ which implies that $W^u(\Lambda_{i_1}) \preceq W^u(\Lambda_{i_3})$.
For the reflexivity, we suppose that both $W^u(\Lambda_i) \preceq W^u(\Lambda_j)$ and $W^u(\Lambda_j) \preceq W^u(\Lambda_i)$ hold true for 
some couple $(i,j)$ with $i\neq j$. Then, applying Lemmas~\ref{l:boundary2} and~\ref{l:partialorder} in this order, and as $i\neq j$, it provides the existence of a point $x$ accumulating to $\Lambda_i$ in the past and to $\Lambda_j$ in the future. Similarly, we can find a 
point $y$ such that the roles of $i$ and $j$ are reversed. Proceeding as in the proof of Lemma~\ref{l:partialorder}, we can find a point $z$ outside $\Lambda_i$ which accumulate to $\Lambda_i$ in the past and in the 
future. This contradicts Lemma~\ref{l:nocycle} and concludes the proof of the reflexivity.

\end{document}